\documentclass{article}
\usepackage{amsmath}
\usepackage{mathrsfs}
\usepackage{algorithmic}
\usepackage{algorithm}
\usepackage{amssymb,multirow,array,float,bm,bbm,pifont}
\usepackage{textcomp,subfigure}
\usepackage{stfloats}
\usepackage{url,multicol}
\usepackage{appendix}
\usepackage[authoryear]{natbib}
\usepackage{graphicx,enumerate}
\usepackage{appendix,epstopdf}
\usepackage[colorlinks,citecolor=blue,urlcolor=blue,linkcolor=red]{hyperref}
\usepackage[capitalize]{cleveref}

\newtheorem{definition}{Definition}[section]
\newtheorem{proposition}{Proposition}[section]
\newtheorem{assumption}{Assumption}[section]
\newtheorem{lemma}{Lemma}[section]
\newtheorem{theorem}{Theorem}[section]

\newtheorem{example}{Example}[section]

\newcommand{\bigO}{\mathcal{O}}
\newcommand{\dee}{\mathrm{d}}
\renewcommand{\Pr}{\mathbb{P}}

\def\EE{\mathbb{E}}

\DeclareMathOperator*{\cov}{Cov}

\title{Model Specification Test for Stationary Functional Time Series}
\date{}

\author{\small
Yan Cui\\
\small Department of Mathematics, Illinois State University
\and
\small Holger Dette\\
\small Department of Mathematics, Ruhr-University Bochum
\and
\small Zhou Zhou\\
\small Center for Data Science and School of Management, Zhejiang University
}

\begin{document}

\maketitle

\begin{abstract}
 We develop a general framework for model specification testing in stationary functional time series. The approach is based on an autoregressive approximation that represents a broad class of stationary functional processes through coefficient kernels whose dimension and autoregressive order may increase with the sample size. Different model assumptions induce different structural restrictions on these kernels, and our tests are constructed by measuring deviations from the corresponding restrictions. We illustrate this principle for three problems: testing a prescribed order of a functional autoregressive model, testing a functional autoregressive moving-average specification, and testing separability of autoregressive coefficient kernels. The resulting statistics are based on weighted $\mathcal{L}^2$-distances, and the critical values are obtained by a multiplier bootstrap. We establish a quantitative bootstrap approximation that is uniform over a class of weight functions and prove asymptotic validity and consistency of the proposed tests. The methodology allows for data-adaptive weighting and is illustrated  by simulations and a data example.
\end{abstract}
\noindent%

\section{Introduction}
\label{sec1}
\def\theequation{1.\arabic{equation}}	
\setcounter{equation}{0}

Functional time series arise whenever a sequence of random curves is observed
over time, with applications including intraday financial records, electricity
demand profiles, environmental trajectories, and high-frequency measurements
aggregated into functional observations. Numerous modeling approaches from
classical time series analysis have been extended to the functional setting
to describe their temporal dependence structure.
 Functional autoregressive (FAR) models and their extensions provide
natural analogues of classical linear time series models and have become
important tools for prediction and inference; see, for example,
\cite{Bosq2000,Mas2007,horvath2012inference,Aue2015}. Functional autoregressive
moving-average (FARMA) models provide a more flexible class by allowing both
autoregressive and moving-average components. These models have been studied from theoretical, estimation, and prediction perspectives; see, for example,
\cite{Bosq2014,KlepschEtAl2017,Kuenzer2024, KuehnertRiceAue2026} and the references therein.
 As in classical time series analysis, however, the
usefulness of inferential and forecasting procedures based on such models
depends critically on whether the assumed dynamic structure provides an
adequate description of the data. This raises fundamental questions concerning
the order of a functional autoregression, the adequacy of more general FARMA
specifications, and possible structural restrictions on the
coefficient operators.

A substantial part of the literature on model diagnostics for functional time
series is based on tests for white noise or serial independence, which can in
particular be applied to residuals from fitted functional time series models.
Early contributions include the portmanteau test of
\cite{GabrysKokoszka2007}, while subsequent approaches have considered the
problem from both the time and frequency domains; see, for example,
\cite{zhang2016white,characiejus2020general,
HlavkaEtAl2021} and the review \cite{kim2023}.  More direct goodness-of-fit tests
for autoregressive models have recently been proposed by
\cite{kim2024projection}, \cite{alvarez2025goodness}, and
\cite{gonzalez2023testing}. Along a different direction,
\cite{MESTRE2021107108} introduced functional autocorrelation and partial
autocorrelation functions as graphical tools for the identification and
diagnosis of FARMA models.

Despite these developments, formal specification testing for functional time
series remains relatively less developed. While procedures for determining
the order of FAR models are available, see, for example,
\cite{KokoszkaReimherr2013}, existing approaches do not provide a common framework for testing FAR and FARMA specifications and structural restrictions on the autoregressive coefficient operators. The purpose of this paper is to develop such a framework.

The main ingredient of our approach is a general autoregressive approximation
for a broad class of stationary functional time series. Its construction is
based on a finite-dimensional representation of the functional observations,
whose dimension is allowed to increase with the sample size. Under suitable
smoothness, moment, and short-range dependence conditions, we show that the
resulting autoregressive representation approximates the original process
uniformly, with an explicit error bound accounting for both the truncation of
the functional observations and the autoregressive approximation. Thus, the
autoregressive representation is not imposed as a model assumption, but
provides a general approximation device for a substantially larger class of
stationary functional processes.

A key observation in our approach  is that particular functional time series models impose
specific structures on the coefficient kernels appearing in this
autoregressive representation. This suggests a general approach to model specification: rather than relying
on diagnostic properties of the fitted residuals, we directly test whether the
estimated autoregressive coefficient kernels satisfy the structural
restrictions implied by the model under consideration.
 In this way, different specification problems can be
treated within the same framework by identifying the corresponding restrictions
on the autoregressive coefficients and measuring deviations from these
restrictions.

We illustrate this principle for three testing problems. First, if the process
follows a FAR model of order $k$, the autoregressive coefficient kernels must
vanish for all lags larger than $k$. This leads to a test of a prescribed FAR
order based on the size of the estimated higher-order coefficient kernels.
Second, a causal and invertible FARMA$(p^*,q^*)$ process admits an
infinite-order autoregressive representation whose coefficients satisfy
restrictions determined by its AR and MA operators. We therefore test a
prescribed FARMA specification by comparing the coefficient kernels obtained
from the unrestricted autoregressive approximation with those implied by the
fitted FARMA model. Third, we  test whether  the  individual coefficient kernels are separable across their two functional arguments. In all
three cases, deviations from the null hypothesis are quantified by weighted $\mathcal{L}^2$-distances between estimated coefficient kernels and their structured
counterparts.

To obtain critical values for the
resulting statistics, we develop  a  multiplier bootstrap. The coefficient kernels are estimated through the finite-dimensional vector autoregressive approximation, and the bootstrap is
constructed from estimated residuals and score vectors. We establish a quantitative bootstrap approximation that is uniform over a class of weight functions and allows both the dimension of the finite-dimensional
approximation and the autoregressive order to increase with the sample size.
As a consequence, the proposed tests have asymptotically correct sizes under
the null hypotheses and are consistent against fixed alternatives.  In particular, the uniformity with respect to the weight functions permits
the use of data-adaptive weights that account for variation in estimation
uncertainty over the functional domain.

The remainder of the paper is organized as follows. In Section~\ref{sec2} we  develop the autoregressive approximation.  Section~\ref{sec3} introduces the specification tests, their estimation and bootstrap implementation.  Section~\ref{sec4} establishes the theoretical validity and consistency of the
procedures. The finite-sample properties of the proposed methodology are investigated by means of a simulation study in Section~\ref{sec5}, while Section~\ref{sec6} presents
an  empirical application. The technical arguments and auxiliary results are
collected in the supplementary material.

\section{Autoregressive approximation of functional time series}
\label{sec2}
\def\theequation{2.\arabic{equation}}	
\setcounter{equation}{0}

\textbf{Convention}: Throughout this paper, let $\mathcal{L}^2([0,1])$ be the separable Hilbert space of all square-integrable functions on $[0,1]$ with inner product $\langle x,y \rangle=\int_0^1 x(u)y(u)\dee u$  and corresponding $\mathcal{L}^2$-norm $|x|_{\mathcal{L}^2}^2=\int_0^1 x^2(u)\dee u$. A square-integrable random function $Y$ in $\mathcal{L}^2([0,1])$  has a second moment, if   $\EE|Y|_{\mathcal{L}^2}^2<\infty$. We also denote by $\mathcal{C}^d([0,1])$ the collection of functions that are $d$-times continuously differentiable with absolutely continuous $d$-th derivative on the interval $[0,1]$.  The notation $\Vert\cdot\Vert$ signifies the operator norm (i.e., largest singular value) when applied to matrices and the Euclidean norm when applied to vectors. We use $|\cdot|_F$ to denote the Frobenius norm of a matrix or  vector. Also, we use $\lambda_{\min}(\cdot)$ and $\lambda_{\max}(\cdot)$ to signify the smallest and largest eigenvalues of the matrices. Throughout this paper, the symbol $C$ denotes a generic finite constant that is independent of $n$ and may vary from place to place.

In this section we 
introduce a unified functional autoregressive (AR) approximation  for a rich class of stationary functional time series, which will be the fundamental basis for the statistical methodology developed in this paper.
 Our approach is  applicable to  a broad class of stationary  time series 
$\{Y_i \}_{i=1}^n$  
in $ \mathcal{L}^2([0,1])$,
which  satisfy   $\EE|Y_i|_{\mathcal{L}^2}^2<\infty$ (here $| \cdot |_{\mathcal{L}^2}$ denotes the common norm in $\mathcal{L}^2 ([0,1])$. To keep the notation simple, we also assume that $\{Y_i \}_{i=1}^n$ is centered, that is 
$\EE[Y_i(u)] = 0$ for all $u\in[0,1]$,  throughout this article. Let $\{\alpha_k\}_{k=1}^\infty$ be a set of pre-determined orthonormal basis functions in $\mathcal{L}^2([0,1])$, then each element of functional time series $\{Y_i\}_{i=1}^n$ admits a  Karhunen-Lo\`{e}ve type expansion of the form
\begin{equation}\label{infinite}
Y_i =\sum_{k=1}^\infty r_{i,k}\alpha_k =\sum_{k=1}^\infty x_{i,k}f_k\alpha_k,
\end{equation}
where, for any given $k\ge 1$, 
\begin{align} \label{det2}
r_{i,k}=\int_0^1 Y_i(u)\alpha_k(u)\dee u
\end{align}
is the $k$th  (random) basis coefficient  of $Y_i$ (with respect to the basis $\{\alpha_k\}_{k=1}^\infty$),  
 $f_k^2:={\rm Var}(r_{i,k})$ denotes the variance  of $r_{i,k}$  and $x_{i,k}:=r_{i,k}/f_k$ if $f_k\neq 0$ (we  set $x_{i,k}=0$ when $f_k=0$).
 The representation in \eqref{infinite} is widely   considered in the functional data analysis literature; see for instance \cite{horvath2012inference} and \cite{hsing2015theoretical}. Note that $f_k$ captures the decay rate of $r_{i,k}$ as $k$ increases and if $f_k\neq 0$, the random coefficients $x_{i,k}$
remain at the same magnitude with variance $1$ as $k$ increases.

Throughout this paper $\| X\|_p$ denotes the $L_p$-norm of a real-valued random variable $X$ ($p\geq 1$).
To facilitate the investigation of functional autoregressive approximation theory, we impose an assumption on the decay rate of the basis expansion coefficients $r_{i,k}$.
\begin{assumption}\label{ass_conti_u}
	The functional time series $\{Y_i\}_{i=1}^n$ is a sequence in  $\mathcal{C}^{d}([0,1])$ (a.s.), where $d \in \mathbb{N}$ is some integer, and  the random coefficients  in the expansion \eqref{infinite}  satisfy $\left\Vert r_{i,k}\right\Vert_2\le Ck^{-(d+1-\theta)}$ for $k=1,2, \ldots $ where the constant $\theta$ is either  $0$ or $1/2$.
\end{assumption}

For a function $f \in \mathcal{C}^d([0,1])$  with  $d \in \mathbb{N} $, the fastest  rate for its $k$th Fourier  coefficient is of order $\bigO(k^{-(d+1)})$ for a wide class of basis functions \citep[see, for example,][]{Chen07}. For example, the Fourier bases (for periodic functions), the weighted Chebyshev polynomials \citep{Trefethen2008} and the orthogonal wavelets with degree $m\ge d$ \citep{Meyer90} admit the latter decay rate under some additional but mild assumptions on the behavior of the function's $d$-th derivative. On the other hand, the basis expansion coefficients may decay at a slower rate for some other orthonormal bases. An example is the basis of normalized Legendre polynomials, where the corresponding basis expansion coefficients decay at a rate of order  $\bigO(k^{-(d+1/2)})$ \citep[see, for example,][]{WX11}. To accommodate a general decay rate of $r_{i,k}$ in the current framework, we therefore consider the order $k^{-(d+1-\theta)}$ with  $\theta=0$ or $1/2$ in Assumption \ref{ass_conti_u}.

To establish a rigorous   functional AR approximation, we will truncate the infinite representation \eqref{infinite} to a finite (but diverging)  linear combination  of basis functions, that is 
\begin{equation}\label{truncation}
\begin{split}
	Y_i(u)& = Y_i^{(p)}(u)+\bigO_\Pr(p^{-d+2\theta}) ,  \\
Y_i^{(p)}(u) &:= \sum_{k=1}^{p} x_{i,k}f_{k}\alpha_k(u), 
    \end{split}
\end{equation}
where $p=p_n$ is the truncation order and $\theta=0$ or $1/2$ is the constant in \cref{ass_conti_u}. This truncated expansion in \eqref{truncation} serves as a strategy of dimension reduction for our theoretical investigation, which is a common approach in functional time series analysis. For example, one could apply the initial dimension reduction by functional principal component analysis \citep{Shang14,KR17}, or explore the properties of linear regression estimators \citep{Hall06,LiHsing07}. Some existing work suggests projecting infinite dimensional objects onto a fixed dimensional subspace to facilitate statistical calculations \citep{Aue2015}, that is, the truncation number is a fixed constant. However, there is increasing interest in allowing the truncation number to grow to infinity with the sample size $n$ to make the truncation adaptive to the smoothness of functional observations \citep[see][]{Hall06,LiHsing07,chiou2016multivariate,chang2024autocovariance}. Throughout this paper, we assume that the truncation number diverges to infinity at a relatively slow speed, that is  $p\asymp n^{c_1} $ for some  $ c_1\in(0,1)$.

Since the functional time series is centered, we have $\EE[x_{i,k}]=0$ for all $i=1,...,n,~k\ge 1$. Ignoring the error in the truncated approximation  \eqref{truncation} for a moment, the problem reduces to studying the $p$-dimensional stationary time series $\{\bm{x}_i\}_{i=1}^n$, where $\bm{x}_i:=(x_{i,1},...,x_{i,p})^\top$ in the representation of $Y_i^{(p)}$. When $i \ge 2$, the best linear prediction (in terms of the mean squared prediction error) of $\bm{x}_i$ from the data  $\bm{x}_1,...,\bm{x}_{i-1}$, can be expressed in the form
\begin{equation}\label{ar_best_pre}
\widehat{\bm{x}}_i=\sum_{j=1}^{i-1}\bm{\Phi}_{j}\bm{x}_{i-j},
\end{equation}
where $\{\bm{\Phi}_{j}\}_{j=1}^{i-1}\in \mathbb{R}^{p\times p}$ are the prediction coefficient matrices. By construction, $\bm{\epsilon}_i:=\bm{x}_i-\widehat{\bm{x}}_i$ is a white noise process with mean $\bm{0}$ and some covariance matrix, say  $\bm{\Sigma}_{\epsilon}\in\mathbb{R}^{p\times p}$. Let 
\begin{align} \label{det27}
\bm{\Gamma}(j)
:=  \big (\gamma_{k,\ell}(j) \big ) _{k,\ell =1,\ldots, p} := \big (\EE[x_{i,k}x_{i+j,\ell}] \big ) _{k,\ell =1,\ldots, p} 
\in \mathbb{R}^{p\times p}
\end{align}
denote  the auto-covariance matrix of  the time series $\{ \bm{x}_i\}_{i=1}^n $ at some lag $j\in\mathbb{Z}$,  
then it follows from the  basis expansion in \eqref{infinite}
that the 
auto-covariance function  $\cov(Y_i(u), Y_{i+j}(v))$ of the stationary functional time series $\{Y_i\}_{i=1}^n$  satisfies 
$$
\int_0^1 \int_0^1 \cov(Y_i(u),Y_{i+j}(v))\alpha_k(u)\alpha_l(v)\dee u\dee v=\gamma_{k,l}(j)f_kf_l~.
$$
 The latter representation indicates that the covariance structure of the functional time series $\{Y_i \}_{i=1}^n$ can be determined by the auto-covariance of the scaled multivariate time series $\{\bm{x}_i\}_{i=1}^n$, provided  that $Y_i^{(p)}$ is a good approximation of $Y_i$ such that the error  in  \eqref{truncation} is small.

We now formulate  in an informal way  a functional AR approximation for  a stationary functional time series (with short-range dependence). A detailed justification of this approximation will be given in Section \ref{sec4}. For this purpose, let $b=b_n$ denote the order of the functional AR approximation. For theoretical and practical purposes, $b$ is required to be much smaller than the sample size $n$ to achieve a parsimonious approximating model. In \cref{approx_ar} of \cref{sec4} we will demonstrate that a wide class of stationary functional time series can be efficiently approximated by stationary functional AR($b$) processes with a slowly diverging order $b$.
More specifically, for a centered stationary  process $\{Y_i\}_{i \in \mathbb{Z}}$  in  $ \mathcal{C}^d([0,1])$ satisfying Assumption \ref{ass_conti_u} and $\EE|Y_i|_{\mathcal{L}^2}^2<\infty$,  we prove that for each fixed $i \geq 2$ and $d\ge 2$,
\begin{equation}\label{approx_far}
	Y_i(u)=\sum_{j=1}^{\min\{i-1,b\}}\int_0^1 
	\psi_j(u,v)Y_{i-j}(v)\dee v+\varepsilon_i(u)
	+\bigO_\Pr\left(p^{1/2}b^{-\tau+2}(\log b)^{\tau-1}+p^{-d+2\theta}\right)
	\end{equation}
	uniformly with respect to $u \in [0,1] $, where   $\psi_j \in\mathcal{L}^2([0,1]^2)$ is an autoregressive operator kernel defined by  the truncated basis expansion 
   \begin{align}
   \label{approx_farA}
    \psi_j(u,v):=
 \sum_{k,l=1}^p  \psi_{j,kl}\alpha_k(u)\alpha_l(v)
   \end{align}
 and the error process $\{\varepsilon_i\} _{i=1}^n$ is  a functional white noise process
 defined by  $\varepsilon_i(u):=\bm{\alpha}_{f+}^\top(u)\bm{\epsilon}_i$, where $\bm{\alpha}_{f+}(u)=(\alpha_1(u)f_1,...,\alpha_p(u)f_p)^\top$.

 Notice that the first error term on the right-hand side of \eqref{approx_far} describes the approximation error of the functional AR process truncated at the order 
$\min\{i-1,b\}$, and the last term reflects the truncation error in the basis expansion \eqref{truncation} (this is a consequence of the proof of Theorem \ref{thm_gaussian} in the \cref{app_proof}). The approximation result in \eqref{approx_far} also reveals that the error bound is adaptive to the smoothness of the functional curves ($d$), and the strength of the temporal dependence in functional observations  ($\tau$). 

\section{Model specification tests}
\label{sec3}
\def\theequation{3.\arabic{equation}}	
\setcounter{equation}{0}

Assessing the adequacy of
a functional linear time series model is of vital importance and is an integral part of model building
and our  AR approximation \eqref{approx_far} result for stationary functional time series  will serve as a vital tool to perform model specification tests for stationary functional time series. More specifically, we  illustrate the power of the AR-approximation \eqref{approx_far} developing 
  multiplier bootstrap tests for 
three testing problems.  First, we use the results to determine  the order $j^\ast$ of a functional AR model by testing whether the coefficient kernels $\psi_j$ in the functional AR approximation representation vanish  for lags $j>j^\ast$. 
Second, we use the results to develop specification tests for  the adequacy of some classic functional linear models such as FARMA models by comparing  the coefficient operator kernels of a causal representation of  the underlying FARMA-model  with those of the approximating  functional AR process. Third, we develop  a test for the hypothesis that  the functional coefficient kernels  $\psi_j$ in the expansion \eqref{approx_far} can be expressed as a separable product of univariate functions. Such a restriction substantially reduces the complexity of the coefficient operator and yields a more parsimonious description of the temporal dynamics. Our purpose in this section is to provide  $\mathcal{L}^2$-type statistics, which can be used to measure deviations from these hypotheses. 
We present our approach on a methodological level and discuss the  statistical validity of the proposed procedures in Section \ref{sec4}.

\subsection{Null hypotheses and test statistics} \label{sec30}

We develop three different types of model specification tests based on estimates of the coefficient functions $\psi_1, \ldots , \psi_b$ in the stochastic representation \eqref{approx_far}. For  this purpose we define for  a given vector  of weight functions $\bm{w}_n =(w_{n,1}, \ldots , w_{n,b})^\top : [0,1]^2 \to \mathbb{R}^b$ by  
$$
|\bm{f}|_{\mathcal{L}^2,\bm{w}_{n}}={1 \over b} \sum_{j=1}^b|f_j/w_{n,j}|_{\mathcal{L}^2}
$$ 
a weighted $\mathcal{L}^2$-norm of a vector-valued function  $\bm{f} =(f_1, \ldots , f_b)^\top : [0,1]^2 \to \mathbb{R}^b$. The weight functions are used to stabilize the contribution of each component in $\bm{f}$ and we assume that 
\begin{align*}
\bm{w}_n \in 
\mathcal{W}&=\Big  \{\bm{g} =(g_1,\ldots,g_b)^\top: [0,1]^2 \to \mathbb{R}^b \Big | ~g_j \in  C([0,1]^2);~
 \inf_{1\le j\le b}\inf_{u,v\in[0,1]}g_j(u,v)\ge \kappa_2 >0\Big  \}.
\end{align*}
(here $\kappa_2$ is a universal positive  constant). A data-adaptive choice of the weight functions will be discussed in Section \ref{sec33} below.
\smallskip

First, we consider the null hypothesis 
\begin{align}\label{h0far}
    &H_0^{\rm FAR}: \quad 
    \{Y_i\}_{i=1}^n \text{~follows~a~FAR}(k) ~\text{process}
    \end{align}
    for a fixed index $k$. For this purpose, we  denote by $\widehat{\psi}_j$ an estimator  for the $j$th  coefficient function $\psi_j$   ($j=1, \ldots , b$) in the functional AR approximation \eqref{approx_far}, which will be precisely defined in Section \ref{sec31}. We  propose to reject $H_0^{\rm FAR}$ for large values of the statistic  
\begin{align} \label{testfar} 
T_1= \sum_{j= k+1}^b \big|\widehat{\psi}_j 
\big|_{\mathcal{L}^2, {w}_{n,j}} 
\end{align}
(later we will use data-adaptive weights in the definition of the norm).
Note that under the null hypothesis we  expect $T_1$ to be small. 
An appropriate  critical value for this decision rule  will be determined by a multiplier bootstrap procedure which will be developed in Section \ref{sec32} (see also Algorithm \ref{alg:FAR_test} for details). 
\smallskip

Next we consider the hypothesis 
\begin{align} \label{h0farma}
    &H_0^{\rm FARMA}: \quad 
    \{Y_i\}_{i=1}^n \text{~follows~a~FARMA}(p^\ast,q^\ast) ~\text{process}.
    \end{align}
If the null hypothesis holds,
and the process is causal and invertible \citep[see][]{horvath2012inference}, there exists an infinite-order AR representation of  the underlying functional ARMA specification. We define by $\widetilde{\psi}_j $ the common estimate of the  $j$th functional AR coefficient in this representation  and  propose to reject  $H_0^{\rm FARMA}$ for large values of the statistic 
\begin{align}
    \label{testfarma}
T_2= 
\sum_{j=1}^b \big|\widehat{\psi}_j -\widetilde{\psi}_j
\big|_{\mathcal{L}^2,{w}_{n,j}} ,
\end{align}
where the critical value  is  determined by a  multiplier bootstrap (see Algorithm \ref{alg:FARMA_test}, which also uses data-adaptive weights).
Note again that under $H_0^{\rm FARMA}$,   we expect that  the test statistic $T_2$ is small. 
\smallskip

Finally, we consider  for a  fixed index $j \in \{ 1,\ldots,b \}$, the hypothesis that the $j$th coefficient function in the FAR approximation \eqref{approx_far} is separable, that is 
\begin{equation}
\label{h0sep}
\begin{split}
    H_0^{{\rm sep},j}: & \quad \psi_j(u,v)=C_jg_j(u)h_j(v)  ~ \text{ for~~constants } 
    C_j\neq 0 \text{ and functions  }  g_j,h_j : [0,1] \to \mathbb{R} \\
    &  \quad \text{such that } \int_0^1 g_j(u)\dee u=\int_0^1 h_j(v)\dee v=1. 
\end{split}
\end{equation}
Such a restriction yields a more parsimonious description of the temporal dynamics. For example, for a FAR(1) process, simple calculations yield that $H_0^{{\rm sep},j}$ implies that the covariance function can be written as $\mbox{Cov}(Y_i(u),Y_{i+k}(v))=Cg_1(u)g_1(v)\rho^k$ for all $k\ge 1$, where $-1<\rho<1$ and $C$ is a constant. This is a simple representation of the functional covariance structure. In \eqref{h0sep}, we assume $C_j\ne 0$ to ensure identification and a well-defined integral-based normalization (we set $\psi_j(u,v) =0$ when $C_j=0$).
We propose to reject $H_0^{{\rm sep},j}$, whenever 
\begin{align}
    \label{testsep}
T_{3,j}= \big |\widehat{\psi}_j^{\rm sep} -\widehat\psi_j  
\big |_{\mathcal{L}^2,{w}_{n,j}}  ,
\end{align}
where the critical value   is again obtained by a multiplier bootstrap (see Algorithm \ref{alg:separability_test}) and $\widehat{\psi}_j^{\rm sep}$ is an estimate of the $j$th coefficient function in the FAR approximation \eqref{approx_far} under the null hypothesis of separability, that is 
\begin{align}
    \label{det21}
\widehat{\psi}_j^{\rm sep}(u,v):=
\widehat{C}_j\widehat{g}_j(u)
\widehat{h}_j(v)
\end{align}
with $$\widehat{C}_j=\int_0^1 \int_0^1 \widehat{\psi}_j(u,v)\dee u\dee v,\quad \widehat{g}_j(u)=\int_0^1 \widehat{\psi}_j(u,v)\dee v/\widehat{C}_j,\quad \widehat{h}_j(v)=\int_0^1 \widehat{\psi}_j(u,v)\dee u/\widehat{C}_j.
$$ 
In the following section we provide details for the construction of the estimators $\widehat{\psi}_j$ and the calculation of 
corresponding quantiles by multiplier bootstrap for testing the hypotheses \eqref{h0far}, \eqref{h0farma} and \eqref{h0sep}.

\subsection{Estimation of the coefficient  kernels} \label{sec31}

In this section we define the estimators of the coefficient functions $\psi_1, \ldots , \psi_b$ in the AR-approximation \eqref{approx_farA} which are required for the calculation of the test statistics $T_1,$ $T_2$ and $T_{3,j}$ in \cref{sec30}. For this  purpose, we denote by
$$
\bm{\Psi}_j = \big ( \psi_{j,kl} \big )_{k,l=1, \ldots , p} 
$$
the  $p\times p$ matrix of the coefficients  defined by \eqref{approx_farA}. Then the $j$th coefficient kernel can be represented as 
$$
\psi_j(u,v)= \bm{\alpha}^\top(u)\bm{\Psi}_j\bm{\alpha}(v), 
$$ 
where $\bm{\alpha}(\cdot)=(\alpha_1(\cdot),\ldots,\alpha_p(\cdot))^\top$ is the vector of basis functions.  It follows from  the proof of \cref{approx_ar} in the supplementary material that the coefficient matrices $\bm{\Psi}_j$ 
can be expressed in terms of the  prediction coefficients 
$\bm{\Phi}_j$ of the best linear prediction \eqref{ar_best_pre}, that is  
\begin{align}
    \label{det20}
\bm{\Psi}_j={\rm diag}(f_1,\cdots,f_p)\bm{\Phi}_j{\rm diag}(1/f_1,\cdots,1/f_p).
\end{align}
Consequently, we can write 
\begin{align}
\label{det4}
\psi_j(u,v)=
\bm{\alpha}_{f+}^\top(u)
\bm{\Phi}_j
\bm{\alpha}_{f-}(v)~,
\end{align}
where 
\begin{equation}
    \begin{split}
        \bm{\alpha}_{f+}(\cdot)& ={\rm diag}(f_1,\ldots,f_p)\bm{\alpha}(\cdot) ,\\
        \bm{\alpha}_{f-}(\cdot)& ={\rm diag}(1/f_1,\ldots,1/f_p)\bm{\alpha}(\cdot) . 
    \end{split}
    \label{det3}
\end{equation}
Therefore, estimating the coefficient kernels 
$\psi_j $ is  equivalent to estimating the unknown coefficient matrices $\bm{\Phi}_j$ in the best linear prediction \eqref{ar_best_pre}.

Now,  Proposition \ref{var_approx} in Section \ref{secD} of the supplemental material provides a  vector autoregressive (VAR) approximation for the time series $\{\bm{x}_i\}_{i=1}^n$, that is 
\begin{equation}\label{var}
    \bm{x}_i=\sum_{j=1}^{\min\{i-1,b\}} \bm{\Phi}_j\bm{x}_{i-j}+\bm{\epsilon}_i+\bigO_\Pr\left(p^{1/2}b^{-\tau+2}(\log b)^{\tau-1}\right).
\end{equation}
We assume that $i \geq b+1$, denote by $\bm{\Phi}=(\bm{\Phi}_1, \ldots , \bm{\Phi}_b)^\top$ the (unknown) $bp\times p$ block coefficient matrix and define an $(n-b)\times bp$  design matrix by  
\begin{equation}
    \label{det11}
\bm{X} = \begin{pmatrix}
    \bm{x}_{b}^\top & \bm{x}_{b-1}^\top & \ldots & \bm{x}_{1}^\top \\
    \bm{x}_{b+1}^\top &\bm{x}_{b}^\top & \ldots & \bm{x}_{2}^\top \\
    \vdots & \vdots & \vdots  & \vdots\\
    \bm{x}_{n-1}^\top &\bm{x}_{n-2}^\top & \ldots &\bm{x}_{n-b}^\top 
\end{pmatrix} \in \mathbb{R}^{(n-b)\times bp}~.
\end{equation}
Denoting by $\bm{Y}=(\bm{x}_{b+1},\ldots,\bm{x}_n)^\top
\in \mathbb{R}^{(n-b)\times p}$  the response matrix and ignoring the approximation error in  the stochastic expansion \eqref{var}, we obtain the linear model
\begin{equation}\label{linear_mod}
 \bm{Y} = \bm{X} \bm{\Phi} + \bm{\varepsilon} ,
\end{equation}
where $\bm{\varepsilon}=(\bm{\epsilon}_{b+1},\ldots,\bm{\epsilon}_n)^\top$.
By \cref{ass_design_matrix} in \cref{sec_multiboots}, together with Lemmas D.3 and E.8 in the supplemental material, $\bm{X}^\top\bm{X}$ is invertible with probability tending to one. Consequently, the least squares estimator $\widehat{\bm{\Phi}} =(\widehat{\bm{\Phi}}_1, \ldots , \widehat{\bm{\Phi}}_b)^\top $ exists and  we define 
$$
\widehat{\bm{\psi}}(u,v):=(\widehat\psi_1(u,v),\ldots,\widehat\psi_b(u,v))^\top
$$
as  the estimate of the vector of coefficient kernels 
$\bm{\psi}(u,v):=(\psi_1(u,v),\ldots,\psi_b(u,v))^\top$, where 
\begin{align}
\label{det5}
\widehat \psi_j(u,v)=
\bm{\alpha}_{f+}^\top(u)
\widehat {\bm{\Phi}}_j
\bm{\alpha}_{f-}(v)
\end{align}
and $\bm{\alpha}_{f-}$  and $\bm{\alpha}_{f+}$  are defined in \eqref{det3}.

\subsection{Critical values via multiplier bootstrap}
\label{sec32}
Standard arguments from least squares estimation show that 
$$
\widehat{\bm{\Phi}}-\bm{\Phi}=  (\bm{X}^\top\bm{X})^{-1}\bm{X}^\top\bm{\epsilon}, 
$$
and observing the representations \eqref{det4}
 and \eqref{det5},
 we can derive a distributional approximation for the difference $\widehat{\bm{\psi}} -  {\bm{\psi}}$. More specifically, 
we define $\bm{x}_i^{(b)}=(\bm{x}_{i-1}^\top,\ldots,
\bm{x}_{i-b}^\top)^\top\in \mathbb{R}^{bp}$,  then joint statistical inference for  the unknown operator kernels $\bm{\psi}(u,v)$ boils down to evaluating the distributional
behavior of the process 
\begin{align}
\label{det1}
\sqrt{n}
\left (\widehat{\bm{\psi}}(u,v)-\bm{\psi}(u,v)\right)
\approx 
\bm{A}_{f}(v)
\left(\frac{\bm{X}^\top\bm{X}}{n}\right)^{-1}
\left(\frac{1}{\sqrt{n}}\sum_{i=b+1}^n
\bm{x}_i^{(b)}\bm{\epsilon}_i^\top
\right)\bm{\alpha}_{f+}(u)~, 
\end{align}
where $\bm{\alpha}_{f+}$  is given  in \eqref{det3} and 
\begin{align}
    \label{det22}
\bm{A}_f(v)=
\begin{pmatrix}
\bm{\alpha}_{f-}^\top(v) & \bm{0} &\cdots  & \bm{0} \\
\bm{0} & \bm{\alpha}_{f-}^\top(v) &\cdots & \bm{0}\\
\vdots & \vdots & \vdots & \vdots \\
\bm{0} & \bm{0} & \cdots & \bm{\alpha}_{f-}^\top(v)~
\end{pmatrix}  \in \mathbb{R}^{b\times bp}.
\end{align}
 For a weight function $\bm{w}_n$
we obtain for  the  corresponding  
 weighted $\mathcal{L}^2$-norm  of \eqref{det1} the distributional approximation
\begin{equation}
\Lambda_{n}^w:=\sqrt{n}
\big |\widehat{\bm{\psi}} -\bm{\psi}\big |_{\mathcal{L}^2,
\bm{w}_{n}}~
\approx 
\left|\bm{A}_{f} 
\left(\frac{\bm{X}^\top\bm{X}}{n}\right)^{-1}
\left(\frac{1}{\sqrt{n}}\sum_{i=b+1}^n
\bm{x}_i^{(b)}\bm{\epsilon}_i^\top
\right)\bm{\alpha}_{f+}
\right|_{\mathcal{L}^2,
\bm{w}_{n}}, \label{sup_diff}
\end{equation}
To motivate the multiplier bootstrap, 
define 
\begin{align}
\label{det23}
\bm{z}_i = (z_{i,1},\ldots,z_{i,bp^2})^\top := 
{\rm vec}\big(\bm{x}_i^{(b)}\bm{\epsilon}_i^\top\big)\in\mathbb{R}^{bp^2} 
\end{align}
as the vectorized version of the matrix $ \bm{x}_i^{(b)}\bm{\epsilon}_i^\top\in \mathbb{R}^{bp \times p}$ where $\bm{x}_i^{(b)}=(\bm{x}_{i-1}^\top,\ldots,
\bm{x}_{i-b}^\top)^\top\in \mathbb{R}^{bp}$,  and
\begin{align}
  \label{det12}  
 \bm{Z}_n^b=(Z_{n,1}^b,\ldots,Z_{n,bp^2}^b)^\top={1 \over \sqrt{n}} \sum_{i=b+1}^n 
\bm{z}_i .
\end{align}
Then
the right hand side of \eqref{det1} 
can be rewritten as 
\begin{equation}\label{eq_target}
\widetilde{\bm{B}}_{n}^{\rm oracle,1}(u,v)
:=
\bm{A}_{f}(v)
\big(\bm{X}^\top\bm{X}/n\big)^{-1}
\widetilde{\bm{I}}{\rm diag}(\bm{Z}_n^b)\widetilde{\bm{E}}
\bm{\alpha}_{f+}(u)
,
\end{equation}
where ${\rm diag}(\bm{Z}_n^b)$ is the $bp^2 \times bp^2$ squared diagonal matrix with  diagonal entries  $Z_{n,k}^b $ defined by \eqref{det12}, 
$\widetilde{\bm{I}}=(\bm{I}_{bp},\ldots,\bm{I}_{bp})\in \mathbb{R}^{bp\times bp^2}$, $\bm{I}_{bp} \in \mathbb{R}^{bp \times bp} $ is the identity matrix and $$\widetilde{\bm{E}}=
\begin{pmatrix}
\bm{1}_{bp} & \bm{0}_{bp} & \cdots & \bm{0}_{bp}\\
\bm{0}_{bp} & \bm{1}_{bp} & \cdots & \bm{0}_{bp}\\
\vdots & \vdots & \ddots & \vdots\\
\bm{0}_{bp} & \bm{0}_{bp} & \cdots & \bm{1}_{bp}
\end{pmatrix}
\in \mathbb{R}^{bp^2\times p}
$$ 
with
$\bm{1}_{bp}=(1,\ldots,1)^\top \in \mathbb{R}^{bp}$ and $\bm{0}_{bp}=(0,\ldots,0)^\top \in \mathbb{R}^{bp}$.
Observing the approximation  \eqref{sup_diff} and these notations, it follows that  the distribution of the weighted $\mathcal{L}^2$-distance $\sqrt{n}
\big|\widehat{\bm{\psi}} -\bm{\psi}\big|_{\mathcal{L}^2, \bm{w}_n}$ can be approximated by the distribution   of $|\widetilde{\bm{B}}_{n}^{\rm oracle, 1} |_{
\mathcal{L}^2, \bm{w}_n}$.
Next, we define a multiplier bootstrapped sum given a block size $m$ as  
\begin{align}
    \label{det13}
\bm{Z}_n^{\ast,b}=(Z_{n,1}^{\ast,b},\ldots,Z_{n,bp^2}^{\ast,b})^\top:=\frac{1}{\sqrt{n-m-b+1}}\sum_{j=b+1}^{n-m+1}\Big (\frac{1}{\sqrt{m}}\sum_{i=j}^{j+m-1}\bm{z}_i\Big )N_j,
\end{align}
where $\{N_j\}_{j \in \mathbb{N}}$ is a sequence of independent standard normal random variables which is independent of the random coefficients $\Upsilon_{b+1}^{n}:=\{\bm{r}_i=(r_{i,1},\ldots,r_{i,p})^\top\}_{i=b+1}^n$.
According to \eqref{eq_target}, we define the bootstrapped statistic as 
\begin{align}
    \label{det25} 
\widetilde{\bm{B}}_{n}^{\rm oracle,2}(u,v)
:=\bm{A}_{f}(v)\big(\bm{X}^\top\bm{X}/n\big)^{-1}\widetilde{\bm{I}}
{\rm diag}(\bm{Z}_n^{\ast,b})
\widetilde{\bm{E}}\bm{\alpha}_{f+}(u) \in \mathbb{R}^{b},
\end{align}
where ${\rm diag}(\bm{Z}_n^{\ast,b})$ is the $bp^2 \times bp^2$ squared diagonal matrix with diagonal entries  $Z_{n,k}^{\ast,b}$ defined by \eqref{det13}. We  will show in Theorems \ref{thm_gaussian} and \ref{thm_boots} of the supplementary
material that, conditional on the data, $\bm{Z}_n^{\ast,b}$ approximates $\bm{Z}_n^b$ in distribution in large samples with high probability. Consequently, we claim that the law of $ |\widetilde{\bm{B}}_{n}^{\rm oracle,1}|_{\mathcal{L}^2, \bm{w}_{n}}$ will be well approximated by the conditional distribution of $|\widetilde{\bm{B}}_{n}^{\rm oracle,2}|_{\mathcal{L}^2, \bm{w}_{n}}$ uniformly over all quantiles and weight functions in $\mathcal{W}$.

However,  this ``oracle'' multiplier bootstrap cannot be directly implemented 
as the definition of the vectors $\bm{z}_i = {\rm vec}\big(\bm{x}_i^{(b)}\bm{\epsilon}_i^\top\big)$ involves the unknown  vectors  $\bm{x}_i:=(x_{i,1},...,x_{i,p})^\top$  and the  ``errors'' $\bm{\epsilon}_i$ in the truncated approximation \eqref{var}. To address this issue, we estimate the vectors 
$\bm{x}_i$ by 
\begin{align}
\label{det7}
\widehat{\bm{x}}_i =\Big (\frac{r_{i,1}}{\widehat{f}_1},\ldots,\frac{r_{i,p}}{\widehat{f}_p}\Big )^\top
\end{align}
and calculate the  residuals by
\begin{align}
\label{det8}
\widehat{\bm{\epsilon}}_i=\bm{\widehat{x}}_i-\sum_{j=1}^b \widehat{\bm{\Phi}}_j
    \bm{\widehat{x}}_{i-j}, ~~~~~~~(   i=b+1,\ldots,n),
    \end{align}
where $\widehat{\bm{\Phi}}_j$ is the least squares estimator in \eqref{linear_mod}. Furthermore, we define  $\widehat{\bm{z}}_i={\rm vec}\big(
\widehat{\bm{x}}_i^{(b)}\widehat{\bm{\epsilon}}
_i^\top\big)$, let
\begin{align}
    \label{det9}
    \widehat{\bm{Z}}_n^{\ast,b}=(\hat Z_{n,1}^{\ast,b},\ldots, \hat Z_{n,bp^2}^{\ast,b})^\top:=
\frac{1}{\sqrt{n-m-b+1}}\sum_{j=b+1}^{n-m+1}\Big (\frac{1}{\sqrt{m}}\sum_{i=j}^{j+m-1}\widehat{\bm{z}}_i\Big )N_j
\end{align}
and introduce 
\begin{align}
    \label{det10}
    \widehat{\bm{B}}_{n}(u,v) =\big  ( 
\widehat{B}_{n,1}(u,v) , \ldots, \widehat{B}_{n,{b}}(u,v) \big )^\top =\
    \widehat{\bm{A}}_{f}(v)\big(\widehat{\bm{X}}^\top
\widehat{\bm{X}}/n\big)^{-1}
\widetilde{\bm{I}}{\rm diag}(\widehat{\bm{Z}}_n^{\ast,b})
\widetilde{\bm{E}}\widehat{\bm{\alpha}}
_{f+}(u),  
\end{align}
where  $\widehat{\bm{X}}$ is defined as the design matrix  $\bm{X}$ in \eqref{det11} with $\bm x_j^\top $  replaced by the  estimate $\widehat{\bm x}_j^\top $ defined in \eqref{det7}. Similarly, $\widehat{\bm{A}}_f(v)$ and $\widehat{\bm{\alpha}}_{f+}(u)$ are defined as $\bm{A}_f(v)$ in \eqref{det22} and $\bm{\alpha}_{f+}(u)$ in \eqref{det3} with the standard deviation $f_k$ replaced by its estimate $\widehat{f}_k$.  We finally define  the bootstrap statistics 
\begin{align}
    \label{det14}
    M_n^1 & = \sum_{j=k+1}^b\big|\widehat{B}_{n,j}\big|_{\mathcal{L}^2,w_{n,j}},
    \\
    M_n^2 &= \big |\widehat{\bm{B}}_{n}\big|_{\mathcal{L}^2, \bm{w}_{n}} ,   \label{det15} \\
    M_{n,j} &=\big|\widehat{B}_{n,j} \big|_{\mathcal{L}^2, w_{n,j}},
       \label{det16}
\end{align}
and use the conditional quantiles of the statistic $M_n^1, M_n^2$ and $M_{n,j}$ for the tests \eqref{testfar}, \eqref{testfarma} and \eqref{testsep}, respectively,  that is
$$
\mathbb{P} \big ( M_n^1 \leq q_{1-\alpha}^{1,*} \mid \Upsilon_{b+1}^n\big ) =
\mathbb{P} \big ( M_n^2 \leq q_{1-\alpha}^{2,*} \mid \Upsilon_{b+1}^n\big ) =
\mathbb{P} \big ( M_{n,j} \leq q_{1-\alpha}^{3,j,*} \mid \Upsilon_{b+1}^n\big ) = 1- \alpha .
$$
In the following section we describe how this procedure is implemented to obtain valid statistical tests for the hypotheses \eqref{h0far}, \eqref{h0farma} and \eqref{h0sep}.

\subsection{Practical implementation}
\label{sec33}

We will first discuss tuning parameter selection in our methodology, and then outline the implementation algorithms for  testing the aforementioned hypotheses.

\subsubsection{Choice of the weight function}
 One could simply choose $\bm{w}_{n}$ as some fixed and equal weights such as $\bm{w}_n(u,v)\equiv \bm{1}$ for all $u,v\in [0,1]$. On the other hand, it is advantageous to choose each weight component as \begin{align}
     \label{det17}
 w_{n,j}(u,v)=\left[{\rm Std}(\widehat{\psi}_j(u,v))\big/ \int_{[0,1]^2} {\rm Std}(\widehat{\psi}_j(u,v))\dee u\dee v\right]^\sigma
 \end{align}
 where $\sigma \in (0,1]$. This data-adaptive choice of weights yields much smaller weighted $\mathcal{L}^2$ norm $|\widehat{\psi}_j(u,v)-\psi_j(u,v)|_{\mathcal{L}^2,w_{n,j}}$ compared to those of the fixed choice. Furthermore, the data-adaptive weighted deviation also captures the standard deviation of $\widehat{\psi}_j(u,v)$ which provides direct information on the estimation uncertainty at each pair $(u,v)\in [0,1]^2$. For practical implementation, we set $\sigma=1/3$, in line with the choice adopted in related studies (\cite{zhang2024simultaneous}). In addition, ${\rm Std}(\widehat{\psi}_j(u,v))$ needs to be estimated in practice. More specifically, we estimate $\text{Std}(\widehat{\psi}_j(u,v))$ using the sample standard deviation of the empirical bootstrap replicates $\{\widehat{B}_{n,j}^{(r)}(u,v)\}_{r=1}^B$. We refer to \cref{prac_impl} for more details.

\subsubsection{Tuning parameter selection}\label{sec_parameter}
There are three parameters to choose throughout our framework, including the truncation number $p$ in the basis expansion, the functional AR order $b$ and the window size $m$ in the multiplier bootstrap. 

(a) The truncation number $p$ can be selected by the cumulative percentage of total variance (CPV) criterion. Let $\bar{Y}=\sum_{i=1}^{n}Y_i/n$ be  the sample mean function and 
$\widehat{c}(u,v)
=\sum_{i=1}^{n}
\left(Y_i(u)-\bar{Y}(u)\right)
\left(Y_i(v)-\bar{Y}(v)\right)/n$
be the sample covariance function, which induces the sample covariance operator
$$
(\widehat{\mathcal{C}}f)(u)
=
\int_0^1
\widehat{c}(u,v)f(v)\,dv,
\quad
f\in \mathcal L^2([0,1]).
$$ 
If $\hat{\lambda}_1, \hat{\lambda}_2,\ldots$ denote the empirical eigenvalues of $\widehat{\mathcal{C}}$, the CPV($p$) is defined as 
$${\rm CPV}(p):=\sum_{i=1}^p \hat{\lambda}_i\Big/ 
\sum_{i=1}^\infty \hat{\lambda}_i
$$ 
(note that $\sum_{i=1}^\infty
\hat{\lambda}_i = 
\int_0^1 \widehat{c}(u,u)\,du$).
We recommend choosing $p$ such that the \rm CPV($p$) exceeds the  predetermined cutoff value 90\% or 95\%.

(b) We determine the optimal functional AR order $b$ using the Akaike information criterion (AIC).   More specifically, for a given  set of candidate orders $\mathcal{B}$, we fit for each $b \in \mathcal{B}$ a VAR($b$) model to the resulting $p$-dimensional time series $\{\bm{x}_i\}$. Let $\widehat{\bm{\Sigma}}_\epsilon(b)$ denote the estimated covariance matrix of the residuals from the fitted VAR($b$) model, then the AIC is defined by 
$${\rm AIC}(b) = \log{\rm det}\big(\widehat{\bm{\Sigma}}_\epsilon(b)\big)+2bp^2 /(n-b),$$
and the ``optimal'' functional AR order can be selected by
$\hat{b}=\underset{b\in\mathcal{B}}{\mathrm{argmin}}\text{AIC}(b)$.

(c) To select the window size $m$, we apply the plug-in method which balances the bias and variance of an  estimator $$\widetilde{\bm{\Xi}}_{m,n}:=\frac{1}{(n-m-b+1)m}\sum_{j=b+1}^{n-m+1}\Big (
\sum_{i=j}^{j+m-1} \bm{z}_i\Big )\Big (
\sum_{i=j}^{j+m-1} \bm{z}_i^\top\Big )$$ of the long-run variance matrix $\bm{\Xi}_n:=\EE[\bm{Z}_n^b(\bm{Z}_n^b)^\top]$. According to the discussion in Section B.5 of \cite{wu2024frequency}, assessing the accuracy of the estimator $\widetilde{\bm{\Xi}}_{m,n}$ for $\bm{\Xi}_n$ in the high dimensional matrix framework is computationally intensive. Therefore, we choose a block size $m$ that minimizes the average mean squared error (AMSE) across dimensions. More specifically, denote $\Xi_{i,n}$ and $\widetilde{\Xi}_{i,m,n}$ as the $i$th diagonal elements of $\bm{\Xi}_n$ and $\widetilde{\bm{\Xi}}_{m,n}$, respectively, for $i=1,\ldots,bp^2$. Then our goal is to find the optimal $m$ that minimizes ${\rm AMSE}(m)=\sum_{i=1}^{bp^2} \EE[\widetilde{\Xi}_{i,m,n}-\Xi_{i,n}]^2/bp^2$. By the same techniques as in
Section B.2-B.3 of \cite{wu2024frequency}, we can estimate the variance and bias of AMSE($m$), denoted by $\bar{C}_{1,n}^2$ and $\bar{C}_{2,n}$, respectively. Consequently, the plug-in block size selector chooses the optimal $m$ by
$$m^\ast = \max\left\{1, \left\lfloor \Big(\frac{2\bar{C}_{2,n}^2}{\bar{C}_{1,n}^2}\Big)n^{1/3}\right\rfloor\right\}.$$   
Further details on estimating the variance and bias of AMSE($m$) can be found in Section B of \cite{wu2024frequency}. Here, we also refer to alternative methods, such as the minimum volatility approach \citep{CZ2022} and the automatic selection methodology \citep{politis2004automatic}, which can also be used to determine the block size $m$.

\subsubsection{Algorithms}\label{prac_impl}

Here, we describe the implementation procedures for the proposed tests when the
data-driven weight function is employed. We first summarize the calculation of the residuals in Algorithm \ref{alg:residuals} and state the bootstrap procedure and the calculation of the weights in Algorithm \ref{boot}. The output of both algorithms  will be the input to Algorithms \ref{alg:FAR_test}--\ref{alg:separability_test}, which define the tests for the hypotheses \eqref{h0far}, \eqref{h0farma} and \eqref{h0sep}, respectively.  

\begin{algorithm}[H]
\caption{Computation of the residual process}
\label{alg:residuals}
\begin{algorithmic}[1]
\STATE Select the truncation number $p$.
\STATE Calculate the scaled multivariate time series
$\{\widehat{\bm{x}}_i\}_{i=1}^n$ defined by \eqref{det7}
\STATE Fit a VAR model to $\{\widehat{\bm{x}}_i\}_{i=1}^n$  is selected by AIC.
\STATE Determine  the AR coefficient matrices
$\{\widehat{\bm{\Phi}}_j\}_{j=1}^b$
\STATE Compute the residual process
$\{\widehat{\bm{\epsilon}}_i\}_{i=b+1}^n$ by \eqref{det8}.
\end{algorithmic}
\end{algorithm}

\vspace{-0.35cm}
\begin{algorithm}[H]
\caption{Multiplier bootstrap and calculation of the weights}
\label{boot}
\begin{algorithmic}[1]
\STATE Apply Algorithm~\ref{alg:residuals} to obtain
$\{\widehat{\bm{x}}_i\}_{i=1}^n$,
$\{\widehat{\bm{\Phi}}_j\}_{j=1}^b$ and
$\{\widehat{\bm{\epsilon}}_i\}_{i=b+1}^n$.
\STATE Choose the window size $m$ using the plug-in selection method in \cref{sec_parameter}.
\STATE Generate $B$ sets of i.i.d. standard normal random variables
$\{N_j^{(r)}\}_{j=b+1}^{n-m+1}$, $r=1,\ldots,B$.
\STATE For each $r=1,\ldots,B$, \\
 \quad - Compute $  \widehat{\bm{Z}}_{n,(r)}^{\ast,b}$ by  \eqref{det9},
where the random variables $\{N_j\}_{j=b+1}^{n-m+1}$ are replaced by $\{N_j^{(r)}\}_{j=b+1}^{n-m+1}$ \\
\quad -  Compute the bootstrapped process $ \widehat{B}_{n,j}^{(r)}$ by  \eqref{det10}, where $\widehat{\bm{Z}}_{n}^{\ast,b}$ is replaced by $\widehat{\bm{Z}}_{n,(r)}^{\ast,b}$.
\STATE Compute  the standard deviation 
$\widehat{\operatorname{Std}}\{\widehat{B}_{n,j}(u,v)\}$
of the bootstrap sample 
$\{\widehat{B}_{n,j}^{(r)}(u,v)\}_{r=1}^B$.
\STATE Compute the weight function
\[
    \widehat{w}_{n,j}(u,v)
    =
    \Bigg [
    \frac{
    \widehat{\operatorname{Std}}\{\widehat{B}_{n,j}(u,v)\}
    }{
    \int_0^1\int_0^1
    \widehat{\operatorname{Std}}\{\widehat{B}_{n,j}(u,v)\}
    \dee u\,\dee v
    }
    \Bigg]^{1/3}   ~.
\]
\end{algorithmic}
\end{algorithm}

\vspace{-0.35cm}
\begin{algorithm}[H]
\caption{Test for FAR$(k)$ models}
\label{alg:FAR_test}
\begin{algorithmic}[1]
\STATE Apply Algorithm~\ref{alg:residuals} and \cref{boot} to obtain 
$\{\widehat{\bm{x}}_i\}_{i=1}^n$,
$\{\widehat{\bm{\Phi}}_j\}_{j=1}^b$ and
$\{\widehat{\bm{\epsilon}}_i\}_{i=b+1}^n$, the bootstrap processes $\{\widehat{B}_{n,j}^{(r)}\}_{r=1}^B$ and the weight functions $ \widehat{w}_{n,j} $ ($j=k+1, \ldots, b$).
\STATE Compute the estimates $
\widehat \psi_j$  in \eqref{det5} and the test statistic 
\begin{align} \label{testfarweight} 
\widehat{T}_1= \sum_{j= k+1}^b \big |\widehat{\psi}_j 
\big |_{\mathcal{L}^2, {\widehat{w}}_{n,j}} 
\end{align}
\STATE Compute the bootstrap statistics
\[
    \widehat{M}_{n,r}^1
    =
    \sum_{j=k+1}^b
    \left|
        \widehat{B}_{n,j}^{(r)}(u,v)
    \right|_{\mathcal{L}^2,\widehat{w}_{n,j}},
    \qquad r=1,\ldots,B .
\]
\STATE Let $\widehat{q}_{1-\alpha}^{1,\ast}$ be the empirical
$(1-\alpha)$-quantile of the bootstrap sample 
$\{\widehat{M}_{n,r}^1\}_{r=1}^B$.
Reject $H_0^{\rm FAR}$  if
\begin{align}
\label{testfarnew} 
   \sqrt{n}  \widehat{T}_1 > \widehat{q}_{1-\alpha}^{1,\ast}.
\end{align}
\end{algorithmic}
\end{algorithm}

\begin{algorithm}[H]
\caption{Test for FARMA$(p^\ast,q^\ast)$ models}
\label{alg:FARMA_test}
\begin{algorithmic}[1]
\STATE  Apply Algorithm~\ref{alg:residuals} and \cref{boot} to obtain 
$\{\widehat{\bm{x}}_i\}_{i=1}^n$,
$\{\widehat{\bm{\Phi}}_j\}_{j=1}^b$ and
$\{\widehat{\bm{\epsilon}}_i\}_{i=b+1}^n$, the bootstrap processes $\{\widehat{B}_{n,j}^{(r)}\}_{r=1}^B$ and the weight functions $ \widehat{w}_{n,j} $ ($j=1, \ldots, b$).

\STATE Fit a VARMA$(p^\ast,q^\ast)$ model to the time series  
$\{\widehat{\bm{x}}_i\}_{i=1}^n$ and estimate the corresponding AR and MA
coefficient matrices, denoted by
$\{\widetilde{\bm{\Phi}}_i\}_{i=1}^{p^\ast}$ and
$\{\widetilde{\bm{\Theta}}_j\}_{j=1}^{q^\ast}$.
\STATE Convert the fitted VARMA$(p^\ast,q^\ast)$ process into its VAR$(\infty)$
representation, with AR coefficient matrices denoted by
$\{\widetilde{\bm{\Psi}}_j\}_{j=0}^\infty$.
\STATE Estimate the corresponding coefficient functions by $
    \widetilde{\psi}_j(u,v)
    =
    \bm{\alpha}_{f+}^{\top}(u)
    \widetilde{\bm{\Psi}}_j
    \bm{\alpha}_{f-}(v).
$
\STATE Compute the  estimates $
\widehat \psi_j$  in \eqref{det5} and the  test 
statistic 
\begin{align}
    \label{testfarmaweight}
\widehat{T}_2= 
\sum_{j=1}^b \big |\widehat{\psi}_j -\widetilde{\psi}_j
\big |_{\mathcal{L}^2,\widehat{w}_{n,j}} 
\end{align}
\STATE Compute the bootstrap statistics
\[
    \widehat{M}_{n,r}^2
    =
    \sum_{j=1}^b
    \left|
        \widehat{B}_{n,j}^{(r)}(u,v)
    \right|_{\mathcal{L}^2,\widehat{w}_{n,j}},
    \qquad r=1,\ldots,B .
\]
\STATE Let $\widehat{q}_{1-\alpha}^{2,\ast}$ be the empirical
$(1-\alpha)$-quantile of the bootstrap sample
$\{\widehat{M}_{n,r}^2\}_{r=1}^B$.
Reject $H_0^{\rm FARMA}$ if
\begin{align}
\label{testfarmanew} 
   \sqrt{n} \widehat{T}_2 > \widehat{q}_{1-\alpha}^{2,\ast}.
\end{align}
\end{algorithmic}
\end{algorithm}

\vspace{-0.35cm}
\begin{algorithm}[H]
\caption{Test for separability of the coefficient function}
\label{alg:separability_test}
\begin{algorithmic}[1]
\STATE Apply Algorithm~\ref{alg:residuals} and \cref{boot} to obtain 
$\{\widehat{\bm{x}}_i\}_{i=1}^n$,
$\{\widehat{\bm{\Phi}}_j\}_{j=1}^b$ and
$\{\widehat{\bm{\epsilon}}_i\}_{i=b+1}^n$, the bootstrap processes $\{\widehat{B}_{n,j}^{(r)}\}_{r=1}^B$ and the weight function $ \widehat{w}_{n,j} $.
\STATE Compute the  estimates $
\widehat \psi_j$  in \eqref{det5}, the separable approximation
$\widehat{\psi}_j^{\rm sep}$ in \eqref{det21}  
and the test statistic 
\begin{align}
    \label{testsepweight}
\widehat{T}_{3,j}= \big |\widehat{\psi}_j^{\rm sep} -\widehat\psi_j  
\big |_{\mathcal{L}^2,{\widehat{w}}_{n,j}} 
\end{align}
\STATE For each $r=1,\ldots,B$, compute
\[
    \widehat{M}_{n,r}^{3,j}
    =
    \left|
        \widehat{B}_{n,j}^{(r)}(u,v)
    \right|_{\mathcal{L}^2,\widehat{w}_{n,j}} .
\]
\STATE Let $\widehat{q}_{1-\alpha}^{3,j,\ast}$ be the empirical
$(1-\alpha)$-quantile of the bootstrap sample 
$\{\widehat{M}_{n,r}^{3,j}\}_{r=1}^B$.
Reject $H_0^{{\rm sep},j}$ if
\begin{align}
\label{testsepnew} 
    \sqrt{n}\,\widehat{T}_{3,j}
    >
    \widehat{q}_{1-\alpha}^{3,j,\ast}.
\end{align}
\end{algorithmic}
\end{algorithm}
\section{Main results} \label{sec4}
\def\theequation{4.\arabic{equation}}	
\setcounter{equation}{0}
In this section, we present the main theoretical results which show that the proposed methodology in Section \ref{sec3} yields valid statistical inference.

\subsection{FAR approximation}We begin with a rigorous statement for the AR-approximation \eqref{approx_far} in Section \ref{sec2} of a centered stationary time series $\{Y_i\}_{i=1}^n$ in $\mathcal{L}^2([0,1])$ with  $\EE|Y_i|_{\mathcal{L}^2}^2<\infty$. 
 For this purpose the following assumption is required.
\begin{assumption}\label{weak_depen_components}	There exist universal constants $C>0$, $\tau >1$ 
such that the auto-covariances in \eqref{det27} satisfy 
	$\Vert \bm{\Gamma}(j)\Vert\le C(|j|+1)^{-\tau}$ for all 
 $j\in \mathbb{Z}$.
\end{assumption}
\cref{weak_depen_components} states that the correlation among the components of $\bm{x}_i$ is relatively weak. This condition is generally mild and can be fulfilled in most cases, as the random components $x_{i,k}$ typically exhibit weak dependence across different $k$ under appropriate basis expansions. Furthermore, simple calculations show that \cref{weak_depen_components} implies $\max_{k,l}|{\rm Cov}(x_{i,k}, x_{i+j,l})| \le C(|j|+1)^{-\tau}$, which provides a polynomial decay rate of the covariance structure of scalar random variables across time. On the other hand, to avoid erratic behavior of the functional AR approximation, the smallest eigenvalue of the covariance matrix of time series $\{\bm{x}_i\}_{i=1}^n$ should be bounded away from zero. Here, we put forth a uniform-positive-definite-in-covariance (UPDC) assumption for stationary multivariate time series below.
\begin{assumption}\label{updc}
	There exists a universal constant $\kappa_1>0$ such that the smallest eigenvalue of the 
$pn \times pn$ covariance  matrix $\cov(\bm{{\rm x}})$ of the vector 
 $\bm{{\rm x}}= (\bm{x}_1^\top,...,
\bm{x}_n^\top)^\top\in\mathbb{R}^{np}$
is bounded from below  by $\kappa_1  $.
 Moreover,
      the stationary time series   $\{x_{i,k}\}_{i \in \mathbb{Z}}$  from the expansion \eqref{infinite} satisfy
    $\sup_{k\ge 1}\Vert x_{i,k}\Vert_q<\infty$ for some  $q> 9$. 
\end{assumption}
The first part of this condition is necessary to avoid an ill-conditioned matrix  $\cov(\bm{{\rm x}})$ and hence facilitates construction of the functional AR approximation. Note that it is a mild requirement and has been widely used in the statistical literature for covariance and precision matrix estimation; see for instance, \cite{CXW13}, \cite{CLZ16} and references therein. In particular, \citep[Theorem 11.8.1]{BD91} states that this UPDC condition holds if its spectral density matrix is uniformly bounded below by a positive constant.

\begin{theorem}\label{approx_ar}
    If Assumptions \ref{ass_conti_u}--\ref{updc} are satisfied and $d\ge 2$, we have for $i\ge 2$ 
\begin{equation}\label{approx_farnew}
	Y_i(u)=\sum_{j=1}^{\min\{i-1,b\}}\int_0^1 
	\psi_j(u,v)Y_{i-j}(v)\dee v+\varepsilon_i(u)
	+\bigO_\Pr\big (p^{1/2}b^{-\tau+2}(\log b)^{\tau-1}+p^{-d+2\theta}\big )
	\end{equation}
	uniformly with respect to $u$, where the error process $\{\varepsilon_i\} _{i=1}^n$ is  a functional white noise process
 defined by  $\varepsilon_i(u):=\bm{\alpha}_{f+}^\top(u)\bm{\epsilon}_i$. Here $\psi_j \in\mathcal{L}^2([0,1]^2)$ is the $j$th  autoregressive operator kernel defined in \eqref{approx_farA}
 and satisfies $\sum_{j=1}^\infty|\psi_j|_{\mathcal{L}^2}<\infty$. 
 \end{theorem}

We continue formulating a general model for the functional time series $\{Y_i\}_{i\in\mathbb{Z}}$, which describes the standardized univariate stationary time series $\{x_{i,k}\}_{i\in\mathbb{Z}}$ appearing in the Karhunen–Loève-type expansion \eqref{infinite} from the perspective of nonlinear systems and  will be essential for our theoretical analysis of the multiplier bootstrap.

\begin{definition}\label{def1}
	Let $\{Y_i\}_{i\in\mathbb{Z}}$ denote a centered stationary functional time series with $\EE|Y_i|_{\mathcal{L}^2}<\infty$.  We say that $\{Y_i\}_{i\in\mathbb{Z}}$ admits a physical representation if for each $k$   there exists   a measurable function $G_{k}$ such that the elements of  $\{x_{i,k}\}_{i \in \mathbb{Z}}$  from the expansion \eqref{infinite} can be represented in the form  $
	x_{i,k}=G_{k}(\mathcal{F}_i),$ where   $\mathcal{F}_i=(...,\eta_{i-1},\eta_i)$ is a one-sided shift process and $\{\eta_i\}_{i \in \mathbb{Z}}$ is a sequence of  i.i.d. random elements. For $l\ge 0$, define the $l$-th physical dependence measure for the functional time series $\{Y_i\}_{i\in\mathbb{Z}}$ with respect to the basis  $\{\alpha_k \}_{k=1}^\infty$  as 
    $$
	\delta_x(l,q)=\sup_{1\le k<\infty}\Vert G_{k}(\mathcal{F}_i)-G_{k}(\mathcal{F}_{i,l})\Vert_q,
    $$
	where $\mathcal{F}_{i,l}=(\mathcal{F}_{i-l-1},\eta_{i-l}^\ast,\eta_{i-l+1},...,
	\eta_i)$ is a coupled process with $\eta_{i-l}^\ast$ an i.i.d. copy of $\eta_{i-l}$. 
\end{definition}
Note that in the above definition, the time series $\{x_{i,k}\}_{i \in \mathbb{Z}}$ is considered as a physical system where functions $G_{k}$ are the underlying data generating mechanisms and $\{\eta_i\}_{i \in \mathbb{Z}}$ are innovations that drive the system. The coefficient 
$\delta_x(l,q)$ does not depend on $i$ and measures the temporal dependence of the functional time series $\{Y_i\}_{i \in \mathbb{N}}$ by quantifying the corresponding changes in the system's output uniformly across all basis expansion coefficients when the shock of the system $l$ steps ahead is changed to an independent copy. We refer to \cite{Wu05} for more discussion of the physical dependence measures with examples on how to calculate them for a wide range of linear and nonlinear univariate time series models.

Definition \ref{def1} covers a wide class of commonly used functional time series models and we  
refer the  reader to Section \ref{app_examples} in the supplemental
material  for some examples.
It is related to the class of functional time series model formulated in \cite{zhou2023}. The difference is that in \eqref{infinite} we separate the standard deviation $f_{k}$ from $r_{i,k}$ and the functional time series model in \cite{zhou2023} is formulated without this extra step. Standardization of the basis expansion coefficients is needed in the fitting of the VAR approximation \eqref{var} to avoid near singularity of the design matrix.  Furthermore, Definition \ref{def1} is also related to the concept of $m$-approximable functional time series introduced in \cite{hormann2010weakly}, as both formulations utilize the concepts of Bernoulli shifts and coupling. The difference lies in our adaptation of the basis expansion,  which separates the functional index $u$ and time index $i$ and hence makes it easier technically to investigate the behavior of various estimators of $\psi_j(u,v)$ uniformly with respect to $u$ and $v$. 

\subsection{Multiplier bootstrap}\label{sec_multiboots}
To establish the main results, we need the following assumptions. Let $\bm{\Sigma}=\EE[\bm{X}^\top\bm{X}/n]$ denote the population version of the matrix $\bm{X}^\top\bm{X}/n$ in \eqref{eq_target}.

\begin{assumption}\label{ass_dep}
	There exists some constant $\tau>5$ such that for some universal constant $C>0$, the physical dependence measure satisfies
	$\delta_x(l,q) \le C(l+1)^{-\tau},~l\ge 0.$
\end{assumption}

\begin{assumption}\label{ass_moment}
The elements of the vector $\bm z_i$ in \eqref{det23} satisfy  $\max_{1\le k\le bp^2}\EE|z_{i,k}|^q\le C_q<\infty$ for each $i$, where $q>9$ is the same moment order  as in \cref{updc}, and $C_q>0$ is a constant.
\end{assumption}

\begin{assumption}
\label{ass_design_matrix}
The smallest eigenvalue of the matrix $\bm{\Sigma}$ is bounded from below by a constant $\kappa_3>0$.
\end{assumption}

\begin{assumption}\label{f_bound}
	For sufficiently large $k$, $|f_{k}|\ge Ck^{-(d+1-\theta)}$, where $C>0$ is a universal finite constant and $\theta=0$ or $1/2$.
\end{assumption}

\cref{ass_dep} is a mild short-range dependence assumption which asserts that the temporal dependence of the stationary functional time series $\{Y_i\}_{i\in\mathbb{Z}}$ decays at a sufficiently fast polynomial rate. We provide two examples in Section C of the supplemental material on how to calculate $\delta_x(l,q)$ for a class of functional MA$(\infty)$ and functional AR(1) processes. \cref{ass_moment} imposes a moment condition on the components of the random vector $\bm{z}_i$ in \eqref{det23}, while \cref{ass_design_matrix} ensures positive definiteness of the design matrix in order to avoid multicollinearity. The constraint on the standard deviation $f_k$
in \cref{f_bound} is commonly assumed in the functional data literature
and imposes a lower bound on $f_k$; see, for example, \cite{Hall07}.

We denote  by 
$$
\widetilde{\bm{\Xi}}_n=\EE[(\bm{Z}_n^{\ast,b})(\bm{Z}_n^{\ast,b})^\top \mid \Upsilon_{b+1}^{n}]
$$
the conditional covariance matrix of $\bm{Z}_n^{\ast,b}$ given $\Upsilon_{b+1}^{n}$ and 
employ the empirical estimates of the random components to implement the multiplier bootstrap procedure. Specifically, recall the notation of $\widehat{\bm{Z}}_n^{\ast,b}$  in \eqref{det9} and the definition of the process $\widehat{\bm{B}}_n$  in \eqref{det10}, then the following theorem provides a bound on the error of the empirical multiplier bootstrap approximation.
    \begin{theorem}
    \label{empirical_thm}
    Let Assumptions \ref{ass_conti_u}--\ref{f_bound} be satisfied and
    assume that the smallest eigenvalue of the matrix $\bm{\Xi}_n$ defined in \cref{sec_parameter} is bounded below by some constant $\kappa_4>0$ and $m=\bigO\big((n/p)^{1/3}\big)$. For some finite constant $C>0$, define 
\begin{align}
    \label{det26}
       \mathcal{A}_n=\big \{\omega:\Delta_n(\omega):=
	 |\widetilde{\bm{\Xi}}_n-\bm{\Xi}_n |_F\le Cbp^{17/6}n^{-1/3}h_n\big \},
\end{align}
    where $\omega$ represents the element in the probability space, $h_n$ diverges to infinity at an arbitrarily slow rate, then $\Pr(\mathcal{A}_n)=1-o(1)$. Moreover, on the event $\mathcal{A}_n$, we have 
        	\begin{align}
        &\sup \limits_{x\in\mathbb{R},\bm{w}_n\in \mathcal{W}}\Big|
		\Pr\Big(
  \big|\widehat{\bm{B}}_n
  \big|
  _{\mathcal{L}^2,\bm{w}_n}
  \le x\mid \Upsilon_{b+1}^{n}\Big)-\Pr\Big(\sqrt{n}|\widehat{\bm{\psi}}-\bm{\psi}|_{\mathcal{L}^2,\bm{w}_n}\le x\Big)\Big|
  \notag \\ 
 & \le
C\left(
        \min\left\{
        (bp^2)^{7/4}n^{-\frac{1}{2}+\frac{9}{2q}+\frac{2}{\tau-1}}, (bp^2)^{\frac{5}{6}} n^{\frac{2}{3 r}-\frac{1}{3}}(\log n)^{\frac{1}{3}}\right\}+b^{7/8}p^{3/2}n^{-1/4+1/q}+bp^{17/6}
        n^{-1/3}h_n\right),
        \label{eq_final}
		\end{align}
where 
        $$
        \frac{1}{r}=
    \max\left\{\frac{1}{q},\frac{1}{q\sqrt{\tau+1} }+\left(\frac{1}{2}-\frac{1}{q\sqrt{\tau+1} }\right) \max \Big\{\frac{2}{q\sqrt{\tau+1}}, \frac{1}{\tau}\big(\frac{1}{2}-\frac{1}{q}\big)\Big\}\right\}.
    $$
\end{theorem}    
From  \cref{empirical_thm} we conclude that, under certain regularity conditions, the law of  $\sqrt{n}|\widehat{\bm{\psi}}-\bm{\psi}|_{\mathcal{L}^2,\bm{w}_n}$ will be well approximated by the empirical conditional distribution of $|\widehat{\bm{B}}_n|_{\mathcal{L}^2, \bm{w}_n}$. Notably, the first term on the right-hand side of \eqref{eq_final} reflects the Gaussian approximation error, the second term represents the empirical multiplier bootstrap approximation error, and the last term relates to the covariance-matrix comparison error between $\widetilde{\bm{\Xi}}_n$ and $\bm{\Xi}_n$. In particular, if the functional time series has rapidly decaying memory ($\tau \to \infty$) and finite moments of arbitrarily large order ($q\to \infty$), the FAR order can be chosen as $b \asymp \log(n)$ and the truncation number grows as large as $p \asymp n^{2/17}$ such that the approximation error in \eqref{eq_final} converges  to zero. 

By \cref{empirical_thm}, the bootstrap critical values consistently estimate the quantiles of the limiting distributions of our test statistics under the null hypotheses. Note that the test statistics \eqref{testfar}, \eqref{testfarma} and \eqref{testsep} in \cref{sec30} can be viewed as functionals of $|\widehat{\bm{\psi}}-\bm{\psi}|_{\mathcal{L}^2,
\bm{w}_n}$, and their sampling distributions are well approximated by the laws of $\widehat{M}_n^1$, $\widehat{M}_n^2$ and $\widehat{M}_{n,j}$ according to \cref{empirical_thm}. A formal statement of the validity and consistency of the proposed procedures is provided in the next subsection.

\subsection{Validity and consistency of the specification tests}\label{sec_consistency}
 The following propositions establish that the proposed tests \eqref{testfarweight}, \eqref{testfarmaweight} and \eqref{testsepweight} have asymptotically correct size under the null hypotheses and are consistent under the alternative hypotheses.

\begin{proposition}\label{consistency_test_far}
Suppose Assumptions \ref{ass_conti_u}--\ref{f_bound} hold. On the event $\mathcal{A}_n$ defined in \eqref{det26} we have\\
(a) If the null hypothesis $H_0^{\rm{FAR}}$ holds, then the test \eqref{testfarweight} satisfies
$$
\lim_{n\to \infty} \Pr(\sqrt{n}  \widehat{T}_1>\widehat{q}_{1-\alpha}^{1,\ast} \mid \Upsilon_{b+1}^{n} )  =  \alpha. 
$$
(b) If the alternative $H_a^{\rm{FAR}}$ holds and $p \ll \big(\frac{n}{b^3\log(n)}\big)^{\frac{1}{d+3-\theta}}$, then the test \eqref{testfarweight} satisfies
$$
\lim_{n\to \infty} \Pr(\sqrt{n} \widehat{T}_1>\widehat{q}_{1-\alpha}^{1,\ast} \mid \Upsilon_{b+1}^{n} )  =  1.
$$
\end{proposition}

\begin{proposition}
\label{consistency_test_farma}
Suppose Assumptions \ref{ass_conti_u}--\ref{f_bound} hold. On the event $\mathcal{A}_n$ defined in \eqref{det26} we have \\
(a) If the null hypothesis $H_0^{\rm{FARMA}}$ holds, then the test \eqref{testfarmaweight} satisfies
$$
\lim_{n\to \infty} \Pr(\sqrt{n}  \widehat{T}_2>\widehat{q}_{1-\alpha}^{2,\ast} \mid \Upsilon_{b+1}^{n} )  =  \alpha.
$$
(b) If the alternative $H_a^{\rm{FARMA}}$ holds and $p \ll \big(\frac{n}{b^3\log(n)}\big)^{\frac{1}{d+3-\theta}}$, then the test \eqref{testfarmaweight} satisfies
$$
\lim_{n\to \infty} \Pr(\sqrt{n} \widehat{T}_2>\widehat{q}_{1-\alpha}^{2,\ast} \mid \Upsilon_{b+1}^{n} )  =  1.
$$
\end{proposition}

\begin{proposition}
\label{consistency_test_sep}
Suppose Assumptions \ref{ass_conti_u}--\ref{f_bound} hold. On the event $\mathcal{A}_n$ defined in \eqref{det26} we have \\
(a) If the null hypothesis $H_0^{{\rm sep},j}$ holds, then the test \eqref{testsepweight} satisfies
$$
\lim_{n\to \infty} \Pr(\sqrt{n}  \widehat{T}_{3,j}>\widehat{q}_{1-\alpha}^{3,\ast} \mid \Upsilon_{b+1}^{n} )  =  \alpha.
$$
(b) If the alternative $H_a^{{\rm sep},j}$ holds and $p \ll \big(\frac{n}{b^3\log(n)}\big)^{\frac{1}{d+3-\theta}}$, then the test \eqref{testsepweight} satisfies
$$
\lim_{n\to \infty} \Pr(\sqrt{n} \widehat{T}_{3,j}>\widehat{q}_{1-\alpha}^{3,\ast} \mid \Upsilon_{b+1}^{n} )  =  1.
$$
\end{proposition}

The additional constraint on $p$ under the alternatives is mild. For example, if the functional time series has rapidly decaying memory ($\tau\to \infty$) based on the normalized Legendre polynomials ($\theta=1/2$) and the FAR order $b\asymp \log(n)$, then the truncation number should be chosen as $p\ll n^{\frac{1}{2d+5}}$.

\section{Simulation studies}
\label{sec5}
\def\theequation{5.\arabic{equation}}	
\setcounter{equation}{0}
Here, we employ the multiplier bootstrap procedure in Algorithm \ref{boot} to evaluate the proposed specification tests for several functional time series structures.

\subsection{Type I errors on model specification tests}\label{simu_size}
First,  we generate functional time series via a general basis expansion framework
$Y_i(u)=\bm{\alpha}_\ast^\top(u)\bm{r}_i$, where $\bm{\alpha}_\ast(u)=(\alpha_1(u),\alpha_2(u),\ldots)^\top$ and $\bm{r}_i=(r_{i,1},r_{i,2},\ldots)^\top$. Different models for multivariate random coefficients are specified below. Note that the normalized Legendre polynomial basis is used in Cases (1) and (5), whereas the Fourier basis is employed in Cases (2)–(4).
\begin{enumerate}[(1)]
    \item \label{case1} {\it Stationary AR(1) model.} Let $\bm{r}_i=(r_{i,1},...r_{i,10})^\top$, consider $$\bm{r}_i=
    \bm{A}_1\bm{r}_{i-1}+\bm{\epsilon}_{i},$$ where $\bm{A}_1$ is a ten-dimensional matrix with $0.3$ at its diagonal and $0.1$ at its off-diagonals. Furthermore, let $\bm{e}_i=(e_{i,1},e_{i,2},\ldots,e_{i,10})^\top$ i.i.d. follow a multivariate normal distribution $\mathcal{N}(0,\bm{\Sigma}_1)$, where the diagonal elements of the matrix  $\bm{\Sigma}_1$ are given by  $1$   and the elements on the off-diagonal are $ 0.4 $. The innovation process is generated by $\epsilon_{i1}=e_{i1}, \epsilon_{i2}=0.8e_{i2}, \epsilon_{i3}=-0.5e_{i3}, \epsilon_{ik}=k^{-2}e_{ik}$ for $k\ge 4,~i=1,...,n$.
\item \label{case2} {\it Stationary AR(2) model.}  Let $\bm{r}_i=(r_{i,1},r_{i,2})^\top$, consider $$\bm{r}_i=\bm{B}_1\bm{r}_{i-1}+\bm{B}_2\bm{r}_{i-2}+\bm{\epsilon}_i,
$$ where 
$$
\bm{B}_1=\begin{pmatrix}
	0.5 & 0.2\\
	-0.2 & -0.5
	\end{pmatrix} ~~\text{ and } ~~\bm{B}_2=\begin{pmatrix}
	-0.3 & 0.7\\
	-0.1 & 0.3
	\end{pmatrix}.
    $$
Further, we assume that the errors $\bm{e}_i=(e_{i1},e_{i2})^\top$ are  i.i.d. following  a centered multivariate $t$-distribution with  $8$  degrees of freedom and the scale parameter 
$\bm{\Sigma}_2=\begin{pmatrix}
1 & 0.4\\
0.4 & 1\end{pmatrix}$. The innovation process is generated by $\epsilon_{i1}=e_{i1},~\epsilon_{i2}=0.5e_{i2},~i=1,\ldots,n$.

    \item \label{case3} {\it Stationary ARMA(1,1) model.}  Let $\bm{r}_i=(r_{i,1},...r_{i,10})^\top$, consider 
    $$
    \bm{r}_i=
    \bm{B}\bm{r}_{i-1}+\bm{\epsilon}_{i}+
    \bm{\Theta}\bm{\epsilon}_{i-1},
    $$
    where $\bm{B}$ is a ten-dimensional diagonal matrix with $0.4$ on its diagonal, $\bm{\Theta}$ is a matrix with $0.4$ on the diagonal and $0.2$ on the off-diagonal. The innovation process is the same as in Case (1). 
\item \label{case4} {\it Stationary ARMA(2,1) model.} Let $\bm{r}_i=(r_{i,1},r_{i,2})^\top$, consider $$\bm{r}_i=\bm{B}_1\bm{r}_{i-1}+\bm{B}_2\bm{r}_{i-2}+\bm{\epsilon}_i+
\bm{\Theta}_1\bm{\epsilon}_{i-1},$$ where 
$$
\bm{B}_1=\begin{pmatrix}
	0.3 & 0\\
	0 & 0.3
	\end{pmatrix}~, ~~ 
    \bm{B}_2=\begin{pmatrix}
	0.2 & 0.1\\
	-0.3 & 0.5
	\end{pmatrix} ~~\text{  and  } ~~
    \bm{\Theta}_1=
    \begin{pmatrix}
	-0.3 & 0.4\\
	0.2 & -0.5
	\end{pmatrix}.
    $$
    The innovation process is of the same form as in Case (2)  except that the i.i.d. innovations  $\bm{e}_i=(e_{i,1},e_{i,2})^\top$  have  a centered multivariate $t$-distribution with $6$ degrees of freedom. 
    \item {\it Stationary FAR(1) model with separable coefficient.}  Consider $$Y_i(u)=\int_0^1 \psi(u,v)Y_{i-1}(v)\dee v+\epsilon_i(u),$$
where $\psi(u,v)=
0.3\sum_{j,k=1}^{5} f_jf_k\alpha_j(u)
\alpha_k(v)$ with $f_1=1, f_2=0.7, f_3=-0.5$ and $f_j=j^{-1}$ for $j\ge 4$. The error process is given by $\epsilon_i(u)=\bm{\alpha}^\top(u)\bm{e}_i,$ where the random variables $\bm{e}_i=(e_{i,1},\ldots,e_{i,5})^\top$ are independent  centered multivariate normal distributed with a covariance matrix $\bm{\Sigma}_3$ having $1$ on the diagonal and $0.4$ on the off-diagonal.
\end{enumerate}

\cref{tab:type1} reports the simulated Type I error
of the tests \eqref{testfarnew}, \eqref{testfarmanew}
 and \eqref{testsepnew} under the respective null hypotheses \eqref{h0far}, \eqref{h0farma} and \eqref{h0sep}, where the observed data are sampled from the aforementioned model settings for sample sizes $n=400$ and $n=800$. All results are based on $1000$ simulation runs and $B=1000$ bootstrap replications. We consider both constant weights and data-adaptive standard-deviation weights at nominal levels $0.05$ and $0.1$. From \cref{tab:type1}, the simulated Type I errors are generally close to their nominal levels, supporting the finite-sample validity of the proposed procedures.

\begin{table*}[htbp!]
	\centering
	\caption{\it Simulated Type I error of the tests \eqref{testfarnew}, \eqref{testfarmanew}
 and \eqref{testsepnew} under the respective null hypotheses \eqref{h0far}, \eqref{h0farma} and \eqref{h0sep}.} 
	\label{tab:type1}
	\vspace*{0.1in}
	\begin{tabular}{|c|c|c|c|c|}
		\hline
	 Model & Weight & $\alpha$ & $n=400$ & $n=800$\\
        \hline
        \multirow{4}{*}{FAR(1)} & \multirow{2}{*}{1} & 0.05 & 0.066 & 0.059 \\
        && 0.1 & 0.113 & 0.114 \\
        \cline{2-5} & \multirow{2}{*}{Std} & 0.05 & 0.070 & 0.064\\
        & & 0.1 & 0.126 & 0.118 \\
        \hline
        \multirow{4}{*}{FAR(2)} & 
        \multirow{2}{*}{1} & 
        0.05 & 0.062 & 0.057 \\
        & & 0.1 & 0.126 & 0.110 \\
        \cline{2-5}
        & \multirow{2}{*}{Std} & 0.05 & 0.074 & 0.057\\
        & & 0.1 & 0.123 & 0.118\\
        \hline
        \multirow{4}{*}{FARMA(1,1)}
        & \multirow{2}{*}{1} & 
        0.05 & 0.058 & 0.054 \\ 
        & & 0.1 & 0.101 & 0.097\\
        \cline{2-5}
        & \multirow{2}{*}{Std} & 0.05 & 0.060 & 0.053\\
        & & 0.1 &  0.090 & 0.088\\
        \hline
        \multirow{4}{*}{FARMA(2,1)}
        & \multirow{2}{*}{1} & 
        0.05 & 0.040 & 0.042 \\ 
        & & 0.1 & 0.085 & 0.102\\
        \cline{2-5}
        & \multirow{2}{*}{Std}
        & 0.05 & 0.043 & 0.044 \\
        & & 0.1 & 0.090  & 0.097\\
        \hline
        \multirow{4}{*}{Separable FAR(1)}
        & \multirow{2}{*}{1} &
        0.05 & 0.057 & 0.064 \\ 
        & & 0.1 & 0.092 & 0.096\\
        \cline{2-5}
        & \multirow{2}{*}{Std}
        & 0.05 & 0.072 & 0.058 \\
        & & 0.1 & 0.103  & 0.106\\
        \hline
	\end{tabular}
\end{table*}

\subsection{Empirical power of model specification tests}
Here, we apply the proposed method to conduct specification tests on typical functional time series models and evaluate their  powers in comparison with existing approaches.

\subsubsection{Goodness-of-fit tests for FAR(2) models}
Now, we investigate the power of the  test \eqref{testfarnew} for FAR(2) models. Although the methodologies proposed in \cite{zhang2016white} and \cite{kokoszka2017inference} can in principle be extended to the goodness-of-fit test for FAR($p$) models with $p>1$, their theoretical justification remains unclear, with respect to the explicit formulation of the effect on estimating the functional AR operators. Instead, we compare our approach with the data-adaptive weight to a recent goodness-of-fit test for autoregressive Hilbertian (ARH) models proposed by \cite{alvarez2025goodness}. The latter
test is formulated in terms of a Cram\'{e}r-von-Mises norm with calibration achieved via
a wild bootstrap resampling procedure, and we denote it  by CvM-test in the following discussion. We consider the null hypothesis
\begin{equation}
\label{h0far2}H_{0}^{\rm FAR}: \{ Y_i \}_{i=1}^n \text{~is~a~FAR(2)~process}
\end{equation}
We suppose that the functional time series is driven by the following model
\begin{equation}\label{farma_far2}
Y_i(u)=\int_0^1 \phi_1(u,v)Y_{i-1}(v)\dee v+\int_0^1
\phi_2(u,v)Y_{i-2}(v)\dee v
+ \varepsilon_i(u)
-\int_0^1 \theta(u,v)\varepsilon_{i-1}(v)\dee v,
\end{equation}
where $\phi_1$, $\phi_2$ and $\theta$ are Gaussian kernels ($c\exp\{-(u^2+v^2)/2\}$ for some constant  $c$) with $|\phi_1(u,v)|_{\mathcal{L}^2} =0.5$, $|\phi_2(u,v)|_{\mathcal{L}^2}=0.3$ and $|\theta(u,v)|_{\mathcal{L}^2}=\delta$. We allow $\delta$ to vary over the interval $[0,0.6]$ in increments of $0.1$ and the innovations  $\{\varepsilon_i\}_{i=1}^n$  are  i.i.d. standard Brownian motion on the interval $[0,1]$. Note  that the null hypothesis of  a FAR($2$) process  corresponds to the choice  $\delta=0$. When $\delta>0$, the statistical power of the test is then evaluated based on the 
data-generating process FARMA($2,1$).

\begin{figure}[htbp!]
	\centering
	\vspace{-0.2cm}
	{\includegraphics[width=7cm,height=6cm]{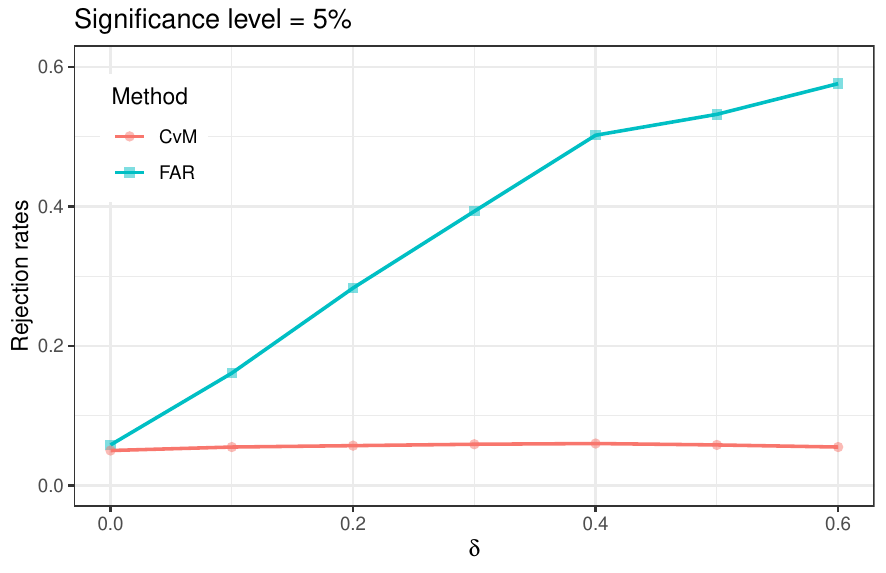}
	\includegraphics[width=7cm,height=6cm]{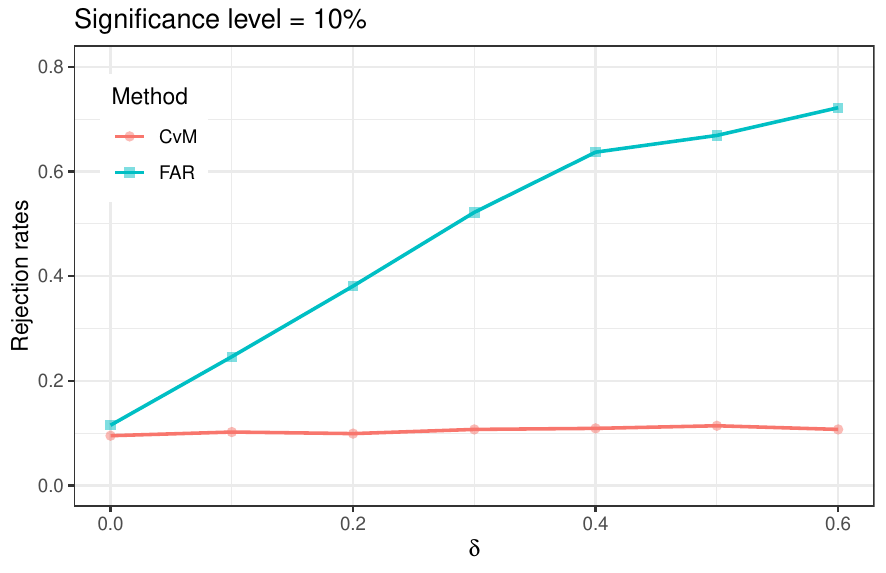}
	}
\vspace{-0.5cm}
\caption{\small \it 
Simulated rejection rates of different tests  for the hypothesis \eqref{h0far2} with nominal levels  $\alpha=5\%$ (Left) and $\alpha=10\%$ (Right). FAR: the test \eqref{testfarnew} proposed in this paper. CvM: the test  proposed by \cite{alvarez2025goodness}.}
\label{fig_power3}
\end{figure}

To implement our multiplier bootstrap and the wild bootstrap for the CvM test, we simulate functional time series with sample size $n=400$ in $1000$ Monte Carlo replicates and $B=1000$ bootstrap resamples. The simulated rejection rates as a function of $\delta$ are shown in \cref{fig_power3}. We observe that both our test and the CvM test achieve nearly nominal empirical sizes at $\alpha=5\%, 10\%$ when $\delta=0$. Furthermore, our test achieves steadily increasing rejection rates as $\delta$ grows. In contrast, the CvM test remains relatively stable as $\delta$ increases, indicating limited power under the considered FARMA($2,1$) alternative. This also demonstrates that the CvM test exhibits certain limitations against FARMA alternatives, even though its algorithm does not explicitly assume an ARH process under the alternative hypothesis.

\subsubsection{Model specification tests for FARMA models}
Next, we will apply our approach with the data-adaptive weights to perform model specification tests for FARMA processes. We consider the hypothesis 
\begin{equation}
\label{h0farma1}H_{0}^{\rm FARMA}: \{Y_i\}_{i=1}^n \text{~is~a~FARMA(1,1)~process}
\end{equation}
where the functional time series is given by
\begin{equation}\label{farma_farma}
Y_i(u)=\int_0^1 \phi(u,v)Y_{i-1}(v)\dee v+\varepsilon_i(u)+\int_0^1 \theta_1(u,v)\varepsilon_{i-1}(v)\dee v
-\int_0^1 \theta_2(u,v)\varepsilon_{i-2}(v)\dee v.
\end{equation}
Here the coefficient kernels $\phi$, $\theta_1$ and $\theta_2$ are Gaussian kernels with $|\phi(u,v)|_{\mathcal{L}^2} =0.3$, $|\theta_1(u,v)|_{\mathcal{L}^2}=0.5$ and $|\theta_2(u,v)|_{\mathcal{L}^2} =\delta$. The parameter $\delta$ varies over the interval $[0,0.5]$ with increment $0.1$ and the innovations $\{\varepsilon_i(u)\}$ are i.i.d. standard Brownian motions on $[0,1]$. Notice that the choice  $\delta=0$ corresponds to the null hypothesis of a FARMA($1,1$) process. 

We performed our multiplier bootstrap procedure in Algorithm \ref{alg:FARMA_test} with $B=1000$ bootstrap replications and sample size $n=400$ ($1000$ simulation runs). The empirical rejection rates as a function of $\delta$ are shown in \cref{fig_power4}. We notice that our test achieves nearly
nominal empirical sizes at $\alpha=5\%, 10\%$ when $\delta=0$ and exhibits a notably increasing power as $\delta$ grows.

\begin{figure}[htbp!]
	\centering
	\vspace{-0.2cm}
	\includegraphics[width=7cm,height=6cm]{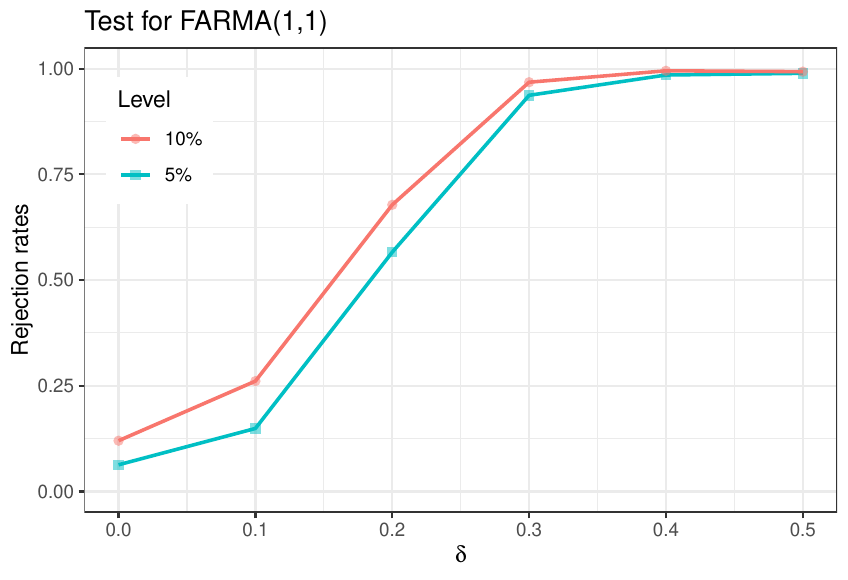}
  \vspace{-0.5cm}
\caption{\small \it Simulated rejection rates of the test \eqref{testfarmanew} for the hypothesis \eqref{h0farma1} with nominal levels $\alpha=5\%$ and $10\%$.}
\label{fig_power4}
\end{figure}  

\subsubsection{Test for separability}

Finally, we perform the separability test on the coefficient function 
\begin{equation}
    \label{h0sep1}H_{0}^{\rm sep}: \psi(u,v)=Cg(u)h(v)
    \end{equation} for some constant $C$ and $\int_0^1 g(u)\dee u=\int_0^1 h(v)\dee v =1$. We consider the FAR(1) model
\begin{equation}\label{eq_sepa}
Y_i(u)=\int_0^1 \psi(u,v)Y_{i-1}(v)\dee v+\varepsilon_i(u),
\end{equation}
where $\psi(u,v)=
0.3\left[\sum_{j,k=1}^{5} f_jf_k\alpha_j(u)
\alpha_k(v)-
\delta \sum_{j=1}^{5}
f_j\alpha_j^2(u)\right]$ with $f_1=1, f_2=0.7, f_3=-0.5$ and $f_j=j^{-1}$ for $j\ge 4$. The error process is the same as in Case (5) in \cref{simu_size}. The parameter $\delta$ is set to range between 0 and 0.6 with increments of 0.05, maintaining stationarity and a moderate level of dependence in the FAR(1) model.
Particularly when $\delta=0$, the coefficient function $\psi $ is separable and \eqref{eq_sepa} implies a dynamic factor structure. \cref{fig_power5} displays the rejection rates of the  multiplier bootstrap test \eqref{testsepnew} with the data-adaptive weight as a function of  $\delta$. It turns out that when $\delta=0$, the empirical sizes are close to their nominal levels $\alpha=0.05$ and $0.1$. On the other hand, rejection rates increase steadily as $\delta$ grows, indicating that our test effectively detects departures from separability.

\begin{figure}[htbp!]
	\centering
	\vspace{-0.2cm}
	\includegraphics[width=7cm,height=6cm]{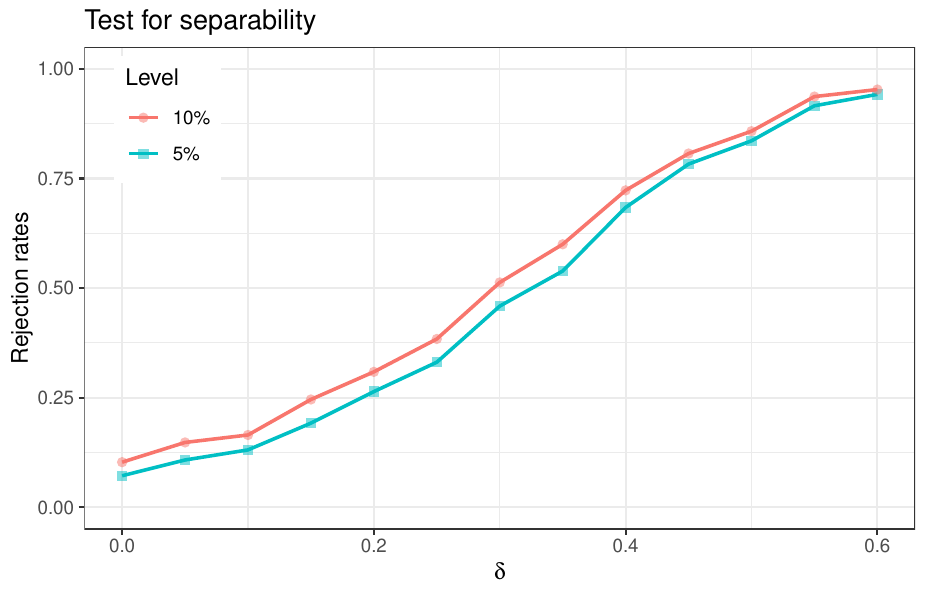}
  \vspace{-0.5cm}
\caption{\small  \it Simulated rejection rates of the test \eqref{testsepnew} for the hypothesis \eqref{h0sep1} with nominal levels $\alpha=5\%$ and $10\%$.}
\label{fig_power5}
\end{figure} 

\section{Empirical illustration}
\label{sec6}
\def\theequation{A.\arabic{equation}}	
\setcounter{equation}{0}

We consider the daily curves of electricity demands (MWh) in Spain from January 1, 2015 to December 31, 2017. These data can be obtained from  \href{https://www.esios.ree.es/en}{Red El\'{e}ctrica de Esp\~{a}na} system operator. Since the daily electricity demands on weekdays and weekends differ, in this paper we focus on the weekday curves (from Monday to Friday) with $n=782$ days. 

The original data are recorded by 10-minute intervals from 00:00--23:50 on each day, which consists of 144 observations. We consider the daily log-transformed real demand curves by smoothing and rescaling them to a continuous interval $[0,1]$. By applying the white noise test of \cite{kokoszka2017inference}, we reject the null hypothesis and conclude that the dataset exhibits temporal dependence. Meanwhile, the stationarity test of \cite{horvath14} fails to reject the null hypothesis at the 5\% significance level, which further supports the application of the proposed methodology to this dataset.

We first implement the  test \eqref{testfarnew}  for a FAR($k$)-model sequentially for $k=1,2,3,\ldots $. The bootstrap sample size is $B=10000$ and we use the normalized Legendre polynomials as basis functions. 
The $p$-values of the different tests are displayed in the left part of Table \ref{T1}.
It turns out that our methodology rejects FAR models with orders $k=1,...,4$ for both constant and data-driven weights. Furthermore, the test \eqref{testfarnew}  rejects the null hypothesis of FAR($5$) at the 10\% significance level under the data-driven weight, but fails to reject the null under the constant weight.  For $k=6$, the test \eqref{testfarnew}  fails to reject the null hypothesis of a FAR($6$) process under both constant and data-driven weights at the 10\% significance level. The corresponding AIC of the fitted model is $-5.474$. These results suggest that the FAR($6$) model offers an appropriate fit to this dataset.

\begin{table*}[htbp!]
	\centering
	\caption{\it P-values for potential FAR($k$) and FARMA($p^\ast,q^\ast$) models.}
	\label{T1}
	\vspace*{0.1in}
	\begin{tabular}{|c|c|c||c|c|c|}
    	\hline
 \multicolumn{3}{|c||}{ FAR test \eqref{testfarnew}} & \multicolumn{3}{c|}{FARMA test \eqref{testfarmanew}} \\
		\hline
	 & \multicolumn{2}{c||}{Weight} &  & \multicolumn{2}{c|}{Weight} \\
     \hline
     $k$ & 1 & {\rm Std} & $(p^\ast,q^\ast)$ & 1 & {\rm Std} \\
    \hline
    1 & 0.0000 &   0.0000 & $(3,1)$ & 	0.0338 & 0.0075 \\
    2  & 0.0054 & 0.0005 & $(3,2)$ & 0.0592 & 0.0124\\
    3 & 0.0034 & 0.0002 & $(4,1)$ & 0.0739 & 0.0346\\ 
    4 & 0.0091 & 0.0015 & $(4,2)$ & 0.1608 & 0.0954\\
    5 & 0.1340 & 0.0943 & $(5,1)$ & \textbf{0.4264}
    & \textbf{0.4102} \\
    6 & \textbf{0.1406} & \textbf{0.1264} & $(6,1)$ & \textbf{0.4599} & \textbf{0.4724}
    \\
    \hline
	\end{tabular}
\end{table*}

Next, we proceed with the proposed model specification tests for potential FARMA processes. The null hypothesis is rejected for FARMA(3,1), FARMA(3,2), and FARMA(4,1) models at the 10\% significance level, indicating that these models do not adequately describe this dataset. For the FARMA($4,2$) model, the $p$-values show rejection of the null under the data-adaptive weight but not under the constant weight at the 10\% level. On the other hand, we notice that FARMA(5,1) and FARMA(6,1) models both pass the specification test, but the FARMA(5,1) model yields the smallest AIC ($-5.504$) among all candidate models. Therefore, combining the specification tests for both the FAR and FARMA models with the AIC criterion, we favor the FARMA(5,1) model over competing candidate models. 

Compared with the existing goodness-of-fit tests reviewed in Section \ref{sec5}, both the KRS and $Z_n$ tests reject the null hypothesis of a FAR(1) process. We further apply the CvM tests to FAR($k$) models for $k=1,...,7$, all candidate FAR models are likewise rejected, with $p$-values equal to zero under $B=10000$ bootstrap replications. This suggests that a pure FAR specification may be inadequate. In conjunction with our proposed FARMA tests, this finding supports the inclusion of a moving-average component.

In conclusion, the selected FARMA(5,1) model indicates that the weekday electricity demand curve depends on demand patterns observed during the preceding working week, while the moving-average component captures the short-term persistence of unexpected demand fluctuations. The autoregressive dependence reflects the regular weekly operating cycle of electricity consumption, whereas the moving-average component may account for transient shocks such as weather conditions, holidays, or other atypical events. Although the data may retain some longer-term seasonal variation, the FARMA(5,1) model provides the most parsimonious representation of the temporal dependence among the weekday demand curves.

\begin{table*}[htbp!]
	\centering
	\caption{\it $P$-values of separability tests for the fitted FAR(6) model.}
	\label{T3}
	\vspace*{0.1in}
	\begin{tabular}{|c|c|c|c|c|c|c|}
		\hline
		& \multicolumn{6}{c|}{FAR lag} \\
		\hline
		Weight & 1 & 2 & 3 & 4 & 5 & 6\\ 
		\hline
		1 &0.0672 & 0.9196 & 0.9765 & 0.9820 & 0 & 0.0051  \\
		\hline
		{\rm Std} & 0.0187 & 0.8917 & 0.9692 & 0.9693 & 0 & 0 \\
		\hline
	\end{tabular}
\end{table*}

We further apply the proposed methodology to test the separability of the FAR coefficient kernels. According to Algorithm \ref{alg:residuals}, we select $b=6$ for the AR order. Table \ref{T3} reports the corresponding $p$-values of the fitted FAR(6) coefficient kernels under the data-adaptive and constant weight functions. For the first coefficient kernel, separability is rejected at the 10\% level under both weight functions. Separability is strongly rejected for the fifth and sixth coefficient kernels under both weighting schemes, whereas the tests provide no evidence against separability for the second, third and fourth coefficient kernels. Overall, departures from separability are concentrated at lags 1, 5, and 6, suggesting more complex dependence at the immediate and approximately weekly horizons.

\bigskip

{\bf Acknowledgements.} The work  of Holger Dette  has been partially supported
by the Deutsche Forschungsgemeinschaft (DFG):
TRR 391 {\it Spatio-temporal Statistics for the Transition of Energy and Transport} (520388526);
Research unit 5381 \textit{Mathematical Statistics in the Information Age} (460867398).

\clearpage

\appendix

\section*{\centering Supplemental Material for ``Model Specification Test for Stationary Functional Time Series''}

\bigskip

\section{Additional simulation studies}
\def\theequation{A.\arabic{equation}}	
\setcounter{equation}{0}

We  investigate the  power of our proposed test for the hypothesis \eqref{h0far}  and compare it with  two existing goodness-of-fit tests  proposed by \cite{zhang2016white} and \cite{kokoszka2017inference}. More specifically, we consider the null hypothesis
\begin{equation}
    \label{h0far1}
H_{01}^{\rm FAR}: \{ Y_i\}_{i=1}^n\text{~is~a~FAR(1)~process}
\end{equation}
and evaluate the power of our specification test \eqref{testfarnew} under a FAR(2) alternative. 
The functional time series are generated by $Y_i(u)=\sum_{k=1}^3 r_{i,k}\alpha_k(u)$, where the Fourier basis is used for $\{\alpha_k(u)\}$ and $\bm{r}_i=(r_{i,1},r_{i,2},r_{i,3})^\top$ follows a VAR process given by
\begin{equation}\label{eq_far}
\bm{r}_i=\bm{B}_1\bm{r}_{i-1}+\delta \bm{B}_2\bm{r}_{i-2}+\bm{\epsilon}_i,
\end{equation}
(where $\delta=0$ corresponds to the null hypothesis) and
$$
\bm{B}_1
=\begin{pmatrix}
0.10 & 0.05 & -0.05 \\
-0.05 & 0.10 & 0.00\\
0.00 & -0.10 & 0.05
\end{pmatrix}~~\text{  and } 
~~\bm{B}_2=
\begin{pmatrix}
    0.20 & -0.15 & 0.10\\
    0.05 & 0.25 & -0.20\\
    -0.10 & 0.0 & 0.30\\
\end{pmatrix}.
$$
The innovation is constructed as follows. Let $\{\bm{e}_i\}_{i=1}^n$
with $\bm{e}_i=(e_{i1},e_{i2},e_{i3})^\top$ be a 
BEKK$(1,0)$ model defined by 
\begin{align*}
\bm{e}_i &= \bm{H}_i^{1/2} \bm{\eta}_i,\\
\bm{H}_i &= 0.1 \bm{I}_3 + \bm{C}\bm{e}_{i-1}\bm{e}_{i-1}^\top\bm{C}^\top,
\end{align*}
where $\bm{I}_3$ is a $3\times 3$ identity matrix, 
$$
\bm{C}=\begin{pmatrix}
    0.3 & 0 & 0\\
    0.1 & 0.2 & 0\\
    0 & 0.1 & 0.25
\end{pmatrix}
$$
and $\{\bm{\eta}_i\}_{i=1}^n$ are independent $3$-dimensional standard normal distributed random variables. Note that this model allows for conditional heteroscedasticity. The innovations $\bm{\epsilon}_i =(\epsilon_{i1},\epsilon_{i2},\epsilon_{i3})^\top $  are then defined by $\epsilon_{i1}=e_{i1}, \epsilon_{i2}=-0.5e_{i2}, \epsilon_{i3}=0.25e_{i3}$.

\begin{figure}[htbp!]
	\centering
	\vspace{-0.2cm}
	{\includegraphics[width=7cm,height=6cm]{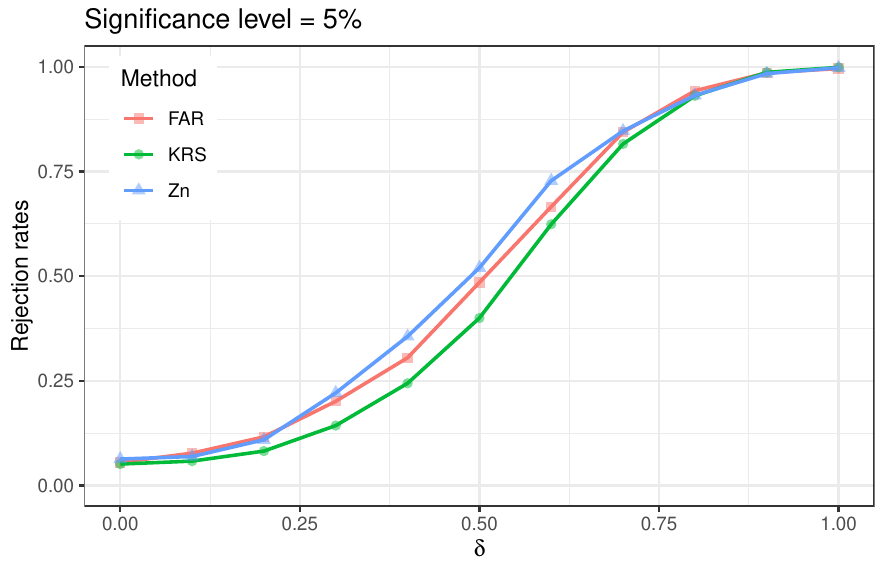}
	\includegraphics[width=7cm,height=6cm]{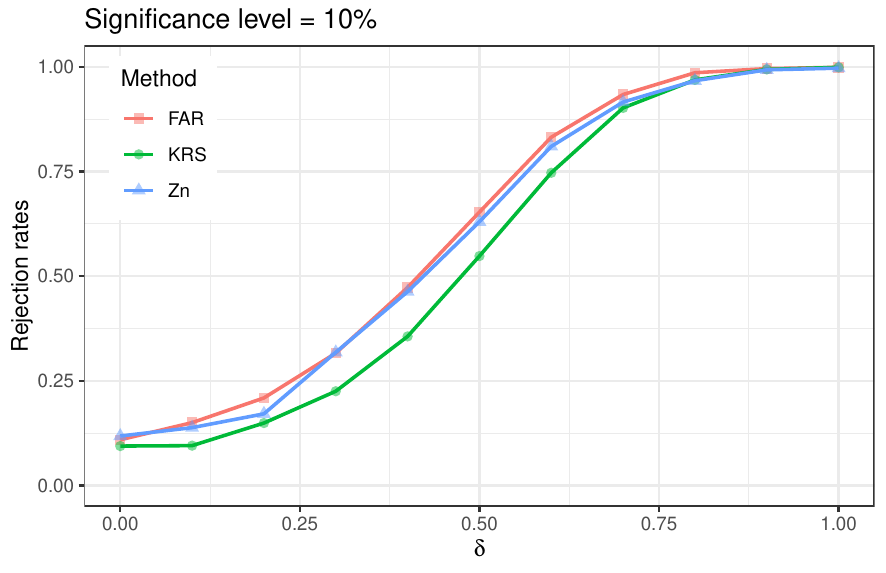}
	}
\vspace{-0.5cm}
\caption{\small \it 
Simulated rejection rates of different tests for the hypothesis \eqref{h0far1} with nominal levels  $\alpha=5\%$ (Left) and $\alpha=10\%$ (Right). FAR: the test \eqref{testfarnew} proposed in this paper. $Z_n$: the test  proposed by \cite{zhang2016white}. KRS: the test suggested in  \cite{kokoszka2017inference}. } 
\label{fig_power2}
\end{figure}

We implement our multiplier bootstrap method \eqref{testfarnew} with the data-adaptive weight to test $H_{01}^{\rm FAR}$ across a range of $\delta$ from $0$ to $1$ in increments of $0.1$ (denoted by FAR test in the following). Moreover, note that the model \eqref{eq_far} reduces to the vector AR(1) model when $\delta=0$ and remains stationary at $\delta=1$. In contrast, we also compare the Ljung-Box type test investigated by \cite{kokoszka2017inference} (denoted by KRS test) and the spectral-based test proposed by \cite{zhang2016white}, which we denote by $Z_n$-test in the following discussion. Specifically, the KRS test is a portmanteau test to assess the cumulative significance of residual auto-covariance operators up to some lag and we choose the maximum lag $H=5$ in our study. The $Z_n$-test proposed by \cite{zhang2016white} considers a
Cram\'{e}r-von-Mises functional based on the functional periodogram to test the uncorrelatedness of the residuals, where a block bootstrap procedure is also required to approximate the null distribution. 

The sample size is set as $n=400$ and the bootstrap procedure is based on $B=1000$ replicates for $1000$ Monte Carlo runs. For  the $Z_n$ test  we used  $B=100$  bootstrap replications to keep computation manageable. In \cref{fig_power2} we display  the empirical rejection rates of  the three tests 
for a nominal level of $5\%$ and $10\%$. Among the three tests, the KRS test exhibits the lowest  power, while our test outperforms both the  KRS and the  $Z_n$-test if the significance level is $\alpha=10\%$. On the other hand,  $Z_n$-test shows a slight superiority compared to our test for moderate values of $\delta$ between $0.4$ and $0.7$ if the nominal level is chosen as $\alpha = 5\%$. Overall, our FAR method delivers comparable goodness-of-fit results to the benchmark methods for testing for a FAR(1) model. 

\section{Technical proofs of main results}\label{app_proof}
\def\theequation{B.\arabic{equation}}	
\setcounter{equation}{0}

\noindent
\textbf{Proof of \cref{approx_ar}}.
The two terms  on the right hand side of  \eqref{approx_farnew} correspond to the approximation and truncation error, and we examine these errors 
separately.  The last term in \eqref{approx_farnew} comes from the truncation error in basis expansion of the functional time series. Given some basis function and the truncation number $p$, it boils down to computing the order of the remaining term $\sum_{k=p+1}^\infty r_{i,k}\alpha_k(u)$. Define the supremum norm of a function in $\mathcal{L}^2([0,1])$ as $|f|_\infty=
\sup_{u\in[0,1]}|f(u)|$. Then for trigonometric polynomials, univariate polynomial splines or orthonormal wavelet
bases, we have $|\alpha_k|_\infty \le C$ for $k=1,2,\ldots$. While for normalized Legendre polynomials, it follows that $|\alpha_k|_\infty \le C\sqrt{k}$ for $k=1,2,\ldots$. In line with our convention for $\theta$, we write the general bound of a basis function as $|\alpha_k|_\infty \le Ck^{\theta}$ where $\theta=0$ or $1/2$. Using  \cref{ass_conti_u}, one can establish that
$$ \Bigg \Vert
\sup_{u\in[0,1]} \Big|\sum_{k=p+1}^\infty 
r_{i,k}\alpha_k(u)
\Big| \Bigg\Vert_2\le 
\sqrt{C\sum_{k=p+1}^\infty
\sum_{l=p+1}^\infty (kl)^{-(d+1-\theta)} (kl)^{\theta}} \le Cp^{-d+2\theta}.$$
Therefore, we have $Y_i(u)=Y_i^{(p)}(u)+\bigO_\Pr(p^{-d+2\theta})$ uniformly with respect to $u$.

On the other hand, the first term in \eqref{approx_farnew} represents the approximation error of functional AR model truncated at $\min\{i-1,b\}$. From \cref{phi_rate} in Section \ref{secD} of the  in the supplementary
material, we have
\begin{align*}
&\left\Vert \sup_{u\in[0,1]} \Big|\bm{\alpha}_{f+}^\top(u)\sum_{j=b+1}^{i-1}\bm{\Phi}_j\bm{x}_{i-j}\Big|\right\Vert_2\le \sup_{u\in[0,1]} |\bm{\alpha}_{f+}^\top(u)|\left\Vert \sum_{j=b+1}^{i-1}\bm{\Phi}_j\bm{x}_{i-j}\right\Vert_2\\
\le & \sup_{u\in [0,1]} |\bm{\alpha}_{f+}(u)| 
\sqrt{
\sum_{j=b+1}^{i-1}
\sum_{k=b+1}^{i-1}
{\rm Tr}\left(
\bm{\Phi}_k\bm{\Gamma}(j-k)\bm{\Phi}_j^\top\right)}\le Cp^{1/2}b^{-\tau+2}(\log b)^{\tau-1}.
\end{align*}

By \cref{weak_depen_components} and using the basis expansion to recover the functional time series, we can derive
\begin{align*}
Y_i^{(p)}(u)
=& \bm{\alpha}_{f+}^\top(u)\bm{x}_i 
=\bm{\alpha}_{f+}^\top(u) \Big [\sum_{j=1}^b \bm{\Phi}_j\bm{x}_{i-j}+\sum_{j=b+1}^{i-1}\bm{\Phi}_j\bm{x}_{i-j}+\bm{\epsilon}_i \Big ]\\
=&\bm{\alpha}_{f+}^\top(u)\sum_{j=1}^b \bm{\Phi}_j\bm{x}_{i-j}+\bm{\alpha}_{f+}^\top(u)\bm{\epsilon}_i+\bigO_\Pr\big(p^{1/2}b^{-\tau+2}(\log b)^{\tau-1}\big)
\end{align*}
Consequently, it turns out that 
\begin{equation}
\label{intermediate}
Y_i(u)=\bm{\alpha}_{f+}^\top(u)\sum_{j=1}^b \bm{\Phi}_j\bm{x}_{i-j}+\bm{\alpha}_{f+}^\top(u)\bm{\epsilon}_i+\bigO_\Pr\big(p^{1/2}b^{-\tau+2}(\log b)^{\tau-1}+p^{-d+2\theta}\big).
\end{equation}

Now, it remains to verify that the first term on the right hand side of \eqref{approx_farnew} is equivalent to the corresponding term in \eqref{intermediate}. For 
this purpose we recall the representation of the
 kernel $\psi_j $ in \eqref{approx_farA},  the definition of the  associated coefficient matrix as $\bm{\Psi}_j=\{\psi_{j,kl}\}_{k,l=1}^p$ for $j=1,\ldots,b$ and the definition of $\bm{\alpha}_{f+}^\top$
 in \eqref{det3}. Then considering the basis expansion 
of  $Y_{i-j}$ in \eqref{infinite} and using the orthogonality  of the basis functions, we have
\begin{align}\label{thm_right_eq}
	\sum_{j=1}^{\min\{i-1,b\}}\int_0^1 \psi_j(u,v)Y_{i-j}(v)\dee v 
	=&\sum_{j=1}^{\min\{i-1,b\}}\sum_{k=1}^p\sum_{l=1}^p\psi_{j,kl}f_lx_{i-j,l}\alpha_k(u)  \\
=&\bm{\alpha}_{f+}^\top(u)\sum_{j=1}^{\min\{i-1,b\}}{\rm diag}(1/f_1,\cdots,1/f_p)
	\bm{\Psi}_j{\rm diag}(f_1,\cdots,f_p)\bm{x}_{i-j}.
	\notag
    \end{align} 
    By \eqref{det20}, we have 
    $\bm{\Phi}_j={\rm diag}(1/f_1,\cdots,1/f_p) 
    \bm{\Psi}_j{\rm diag}(f_1,\cdots,f_p)$, which yields 
    \begin{equation}
    \label{eq_last}
    \sum_{j=1}^{\min\{i-1,b\}}\int_0^1 \psi_j(u,v)Y_{i-j}(v)\dee v=\bm{\alpha}_{f+}^\top(u)
    \sum_{j=1}^{\min
    \{i-1,b\}}
    \bm{\Phi}_j\bm{x}_{i-j}.
    \end{equation}
 The assertion now follows combining Eqs. \eqref{intermediate} and \eqref{eq_last}. 
 $\hfill \square$
\medskip

\noindent
 \textbf{Proof of \cref{empirical_thm}}. Denote $1\le k,l\le bp^2$ as indices of the coordinates of the covariance matrix. Similarly to the proof of Lemma 9 in \cite{CZ2022} and equipped with \cref{eps_depend}, we can obtain $\Vert \widetilde{\bm{\Xi}}_{n,kl}-\bm{\Xi}_{n,kl}\Vert_2=\bigO(\frac{\sqrt{p}}{m}+p\sqrt{\frac{m}{n}})$. By balancing these two error terms, we have $|\widetilde{\bm{\Xi}}_n-\bm{\Xi}_n|_F=\bigO_\Pr(bp^{17/6}n^{-1/3})$ by choosing $m=\bigO((n/p)^{1/3})$. Given a sequence $h_n$ that diverges to infinity at an arbitrarily slow rate, then $\Pr(\mathcal{A}_n)=1-o(1)$ by the Markov's inequality. 

 Define the target distance as
 $$\widehat{\mathcal{D}}(\widehat{\bm{Z}}_n^{\ast,b},\bm{Z}_n^b):= \sup_{x\in\mathrm{R}, \bm{w}_n\in\mathcal{W}}
 \left|\Pr\left(\left|\widehat{\bm{B}}_n\right|_{\mathcal{L}^2,\bm{w}_n}\le x \mid \Upsilon_{b+1}^n\right)
 -\Pr\left(\left|\widetilde{\bm{B}}_n^{\rm oracle,1}\right|_{\mathcal{L}^2,\bm{w}_n}\le x\right)\right|.$$
Similar to the proof of \cref{thm_boots} in \cref{add_theory}, we first decompose $\widehat{\mathcal{D}}(\widehat{\bm{Z}}_n^{\ast,b},\bm{Z}_n^b)$ into several distance terms. For any two events $A$ and $B$ and given a $\sigma$-field $\mathcal{G}$, note that the inequality $\Pr(A\mid\mathcal{G})\le \Pr(A\cap B \mid \mathcal{G})+\Pr(B^c\mid \mathcal{G})$ holds. Moreover, denote $\widehat{\bm{\Sigma}}=\widehat{\bm{X}}^\top\widehat{\bm{X}}/n$. Given positive sequences $\delta_1=\delta_{1,n}$ and $\delta_2=\delta_{2,n}$ that converge to zero as $n\to \infty$, then we have
{\footnotesize
\begin{align*}
    &\sup_{x\in\mathbb{R},\bm{w}_n\in\mathcal{W}}\left|
    \Pr\left(\left|\widehat{\bm{B}}_n(u,v)\right|_{\mathcal{L}^2,
    \bm{w}_n}\le x \mid  \Upsilon_{b+1}^n\right)-\Pr\left(\left|\widetilde{\bm{B}}_n^{\rm oracle,1}(u,v)\right|_{\mathcal{L}^2,\bm{w}_n}\le x\right)\right| \\
    \le & \sup_{x\in\mathbb{R},\bm{w}_n\in\mathcal{W}}\left|
    \Pr\left(\left|\bm{A}_f(v)\widehat{\bm{\Sigma}}^{-1}
    \widetilde{\bm{I}}{\rm diag}(\widehat{\bm{Z}}_n^{\ast,b})\widetilde{\bm{E}}\widehat{\bm{\alpha}}_{f+}(u)\right|_{\mathcal{L}^2,
    \bm{w}_n}\le x+\delta_1 \mid \Upsilon_{b+1}^n\right)-\Pr\left(\left|\widetilde{\bm{B}}_n^{\rm oracle,1}(u,v)\right|_{\mathcal{L}^2,\bm{w}_n}\le x\right)\right|\\
    &\qquad + \sup_{\bm{w}_n\in\mathcal{W}}\left|
    \Pr\left(\left|[\widehat{\bm{A}}_f(v)-\bm{A}_f(v)]\widehat{\bm{\Sigma}}^{-1}
    \widetilde{\bm{I}}{\rm diag}(\widehat{\bm{Z}}_n^{\ast,b})\widetilde{\bm{E}}\widehat{\bm{\alpha}}_{f+}(u)\right|_{\mathcal{L}^2,
    \bm{w}_n}\ge \delta_1 \mid \Upsilon_{b+1}^n \right)\right|\\
    \le & \sup_{x\in\mathbb{R},\bm{w}_n\in\mathcal{W}}\left|
    \Pr\left(\left|\bm{A}_f(v)\widehat{\bm{\Sigma}}^{-1}
    \widetilde{\bm{I}}{\rm diag}(\widehat{\bm{Z}}_n^{\ast,b})\widetilde{\bm{E}}\bm{\alpha}_{f+}(u)\right|_{\mathcal{L}^2,
    \bm{w}_n}\le x+\delta_1+\delta_2 \mid  \Upsilon_{b+1}^n\right)-\Pr\left(\left|\widetilde{\bm{B}}_n^{\rm oracle,1}(u,v)\right|_{\mathcal{L}^2,\bm{w}_n}\le x\right)\right|\\
    &\qquad + \sup_{\bm{w}_n\in\mathcal{W}}\left|
    \Pr\left(\left|\left[\widehat{\bm{A}}_f(v)-\bm{A}_f(v)\right]\widehat{\bm{\Sigma}}^{-1}
    \widetilde{\bm{I}}{\rm diag}(\widehat{\bm{Z}}_n^{\ast,b})\widetilde{\bm{E}}\widehat{\bm{\alpha}}_{f+}(u)\right|_{\mathcal{L}^2,
    \bm{w}_n}\ge \delta_1 \mid \Upsilon_{b+1}^n \right)\right|\\
    &\qquad + \sup_{\bm{w}_n\in\mathcal{W}}\left|
    \Pr\left(\left|\bm{A}_f(v)\widehat{\bm{\Sigma}}^{-1}
    \widetilde{\bm{I}}{\rm diag}(\widehat{\bm{Z}}_n^{\ast,b})\widetilde{\bm{E}}\left[\widehat{\bm{\alpha}}_{f+}(u)-\bm{\alpha}_{f+}(u)\right]\right|_{\mathcal{L}^2,
    \bm{w}_n}\ge \delta_2 \mid \Upsilon_{b+1}^n \right)\right|\\
    \le & \sup_{x\in\mathbb{R},\bm{w}_n\in\mathcal{W}}\left|
    \Pr\left(\left|\bm{A}_f(v)\widehat{\bm{\Sigma}}^{-1}
    \widetilde{\bm{I}}{\rm diag}(\widehat{\bm{Z}}_n^{\ast,b})\widetilde{\bm{E}}\bm{\alpha}_{f+}(u)\right|_{\mathcal{L}^2,
    \bm{w}_n}\le x+\delta_1+\delta_2 \mid  \Upsilon_{b+1}^n\right)- \Pr\left(\left|\widetilde{\bm{B}}_n^{\rm oracle,2}(u,v)\right|_{\mathcal{L}^2,
    \bm{w}_n}\le x+\delta_1+\delta_2 \mid  \Upsilon_{b+1}^n\right)\right|\\
    &\qquad + \sup_{x\in\mathbb{R},\bm{w}_n\in\mathcal{W}}\left|\Pr\left(\left|\widetilde{\bm{B}}_n^{\rm oracle,2}(u,v)\right|_{\mathcal{L}^2,\bm{w}_n}\le x+\delta_1+\delta_2 \mid \Upsilon_{b+1}^n \right)-\Pr\left(\left|\widetilde{\bm{B}}_n^{\rm oracle,2}(u,v)\right|_{\mathcal{L}^2,\bm{w}_n}\le x \mid  \Upsilon_{b+1}^n\right)\right|\\
    &\qquad + \sup_{x\in\mathbb{R},\bm{w}_n\in\mathcal{W}}\left|\Pr\left(\left|\widetilde{\bm{B}}_n^{\rm oracle,2}(u,v)\right|_{\mathcal{L}^2,\bm{w}_n}\le x \mid \Upsilon_{b+1}^n \right)-\Pr\left(\left|\bm{B}_n^u(u,v)\right|_{\mathcal{L}^2,\bm{w}_n}\le x \right)\right|\\
    &\qquad + \sup_{x\in\mathbb{R},\bm{w}_n\in\mathcal{W}}\left|\Pr\left(\left|\bm{ B}_n^u(u,v)\right|_{\mathcal{L}^2,\bm{w}_n}\le x \right)-\Pr\left(\left|\bm{B}_n^{\rm oracle}(u,v)\right|_{\mathcal{L}^2,\bm{w}_n}\le x \right)\right|\\
    &\qquad + \sup_{x\in\mathbb{R},\bm{w}_n\in\mathcal{W}}\left|\Pr\left(\left|\bm{B}_n^{\rm oracle}(u,v)\right|_{\mathcal{L}^2,\bm{w}_n}\le x \right)-\Pr\left(\left|\widetilde{\bm{B}}_n^{\rm oracle,1}(u,v)\right|_{\mathcal{L}^2,\bm{w}_n}\le x \right)\right|\\
    &\qquad + \sup_{\bm{w}_n\in\mathcal{W}}\left|
    \Pr\left(\left|\left[\widehat{\bm{A}}_f(v)-\bm{A}_f(v)\right]\widehat{\bm{\Sigma}}^{-1}
    \widetilde{\bm{I}}{\rm diag}(\widehat{\bm{Z}}_n^{\ast,b})\widetilde{\bm{E}}\widehat{\bm{\alpha}}_{f+}(u)\right|_{\mathcal{L}^2,
    \bm{w}_n}\ge \delta_1 \mid \Upsilon_{b+1}^n \right)\right|\\
    &\qquad + \sup_{\bm{w}_n\in\mathcal{W}}\left|
    \Pr\left(\left|\bm{A}_f(v)\widehat{\bm{\Sigma}}^{-1}
    \widetilde{\bm{I}}{\rm diag}(\widehat{\bm{Z}}_n^{\ast,b})\widetilde{\bm{E}}\left[\widehat{\bm{\alpha}}_{f+}(u)-\bm{\alpha}_{f+}(u)\right]\right|_{\mathcal{L}^2,
    \bm{w}_n} \ge \delta_2 \mid \Upsilon_{b+1}^n \right)\right|\\
    =:&~
    I+II+III+IV+V+VI+VII.
    \end{align*}}
    
   Since $\left|\widetilde{\bm{B}}_n^{{\rm oracle},2}(u,v)\right|_{\mathcal{L}^2,\bm{w}_n}$ is a Gaussian random variable conditional on $\Upsilon_{b+1}^n$, we have $II\le  C(\delta_1+\delta_2)$. In addition, from Theorems \ref{thm_gaussian} and \ref{thm_boots}, we obtain that $III\le C\left(b^{7/8}p^{5/4}n^{-1/4}\log^{1/2}(n)+bp^{17/6}n^{-1/3}h_n\right)$ and $IV\le C\left(\min\left\{
        (bp^2)^{7/4}n^{-\frac{1}{2}+\frac{9}{2q}+\frac{2}{\tau-1}}, (bp^2)^{\frac{5}{6}} n^{\frac{2}{3 r}-\frac{1}{3}}(\log n)^{\frac{1}{3}}\right\}\right)$. Now we turn to derive the bound of the error term 
        $I$. Recall $\widetilde{\bm{V}}_n^m={\rm vec}(\widetilde{\bm{\Sigma}}^{-1}\widetilde{\bm{I}}{\rm diag}(\bm{Z}_n^{\ast,b})\widetilde{\bm{E}})$, define the estimated vector as $\widehat{\bm{V}}_n^m={\rm vec}(\widehat{\bm{\Sigma}}^{-1}\widetilde{\bm{I}}{\rm diag}(\widehat{\bm{Z}}_n^{\ast,b})\widetilde{\bm{E}})\in\mathbb{R}^{bp^2}$ and an intermediate random vector as $\bar{\bm{V}}_n^m={\rm vec}(\widehat{\bm{\Sigma}}^{-1}
    \widetilde{\bm{I}}{\rm diag}(\bm{Z}_n^{\ast,b})
    \widetilde{\bm{E}})$. Then followed by \cref{comparison},
    the distance of $I$ becomes
        {\footnotesize
        \begin{align}
        & \sup_{x\in\mathbb{R},\bm{w}_n\in\mathcal{W}}\left|
    \Pr\left(\left|\bm{A}_f(v)
    \widetilde{\bm{I}}{\rm diag}(\widehat{\bm{V}}_n^m)\widetilde{\bm{E}}\bm{\alpha}_{f+}(u)\right|_{\mathcal{L}^2,
    \bm{w}_n}\le x \mid \Upsilon_{b+1}^n\right)-\Pr\left(\left|\bm{A}_f(v)
    \widetilde{\bm{I}}{\rm diag}(\widetilde{\bm{V}}_n^m)\widetilde{\bm{E}}\bm{\alpha}_{f+}(u)\right|_{\mathcal{L}^2,\bm{w}_n}\le x \mid \Upsilon_{b+1}^n \right)\right| \notag\\
    &\le 4(bp^2)^{1/4}\sigma + \sup_{T\in \widetilde{\mathcal{T}}_n}
    \left|\EE\left[h_{T,\sigma}(\widehat{\bm{V}}_n^m)-h_{T,\sigma}(\widetilde{\bm{V}}_n^m)\mid \Upsilon_{b+1}^n\right]\right| \notag\\
    &\le 4(bp^2)^{1/4}\sigma
    +\frac{C}{\sigma}\left(\EE\left[\left|\widehat{\bm{V}}_n^m-\bar{\bm{V}}_n^m\right|\mid \Upsilon_{b+1}^n\right]+
    \EE\left[\left|\bar{\bm{V}}_n^m-\widetilde{\bm{V}}_n^m\right|\mid  \Upsilon_{b+1}^n\right]\right) \notag\\
    &\le 4(bp^2)^{1/4}\sigma
    +\frac{C}{\sigma}\left(\sqrt{\EE\left[(\widehat{\bm{V}}_n^m-\bar{\bm{V}}_n^m)^\top
    (\widehat{\bm{V}}_n^m-\bar{\bm{V}}_n^m)\mid \Upsilon_{b+1}^n\right]}+
    \sqrt{\EE\left[(\bar{\bm{V}}_n^m-\widetilde{\bm{V}}_n^m)^\top(\bar{\bm{V}}_n^m-\widetilde{\bm{V}}_n^m)\mid  \Upsilon_{b+1}^n\right]}\right) \notag\\
    &= 4(bp^2)^{1/4}\sigma+
    \frac{C}{\sigma}(I_1+I_2). \label{error_i}
        \end{align}}
        By \cref{consis_epsilon} and the fact ${\rm vec}(AB) = (I\otimes B){\rm vec}(A)$, we have
        
        \begin{align*}
            I_1 \le & \sqrt{{\rm Tr}\left(\left[\bm{I}_p \otimes \widehat{\bm{\Sigma}}^{-1}\right]^2\frac{1}{n-m-b+1}
    \sum_{i=b+1}^{n-m+1}\left(\frac{1}{\sqrt{m}}\sum_{j=i}^{i+m-1}(\widehat{\bm{z}}_j-\bm{z}_j)\right)
    \left(\frac{1}{\sqrt{m}}\sum_{j=i}^{i+m-1}(\widehat{\bm{z}}_j-\bm{z}_j)^\top\right)\right)}\\
    \le & \left\Vert \bm{I}_p\otimes \widehat{\bm{\Sigma}}^{-1}\right\Vert \sqrt{\frac{1}{n-m-b+1}
    \sum_{i=b+1}^{n-m+1}\left|\frac{1}{\sqrt{m}}\sum_{j=i}^{i+m-1}
    \left[(\widehat{\bm{x}}_i^{(b)}-\bm{x}_i^{(b)})\otimes \widehat{\bm{\epsilon}}_i + \bm{x}_i^{(b)}\otimes (\widehat{\bm{\epsilon}}_i-\bm{\epsilon}_i)\right]\right|_F^2}\\
    \le & C\left(b^{1/2}p^{3/2}n^{-1/2}+b^{3/2}p^{5/2}n^{-1/2+2/q}\right)\le Cb^{3/2}p^{5/2}n^{-1/2+2/q}.
        \end{align*}
        On the other hand, by \cref{consis_epsilon} we find that
        \begin{align*}
        & \left\Vert \widehat{\bm{\Sigma}}^{-1}-\widetilde{\bm{\Sigma}}^{-1}\right\Vert =
        \left\Vert \widehat{\bm{\Sigma}}^{-1}\big(\widetilde{\bm{\Sigma}}-\widehat{\bm{\Sigma}}\big)\widetilde{\bm{\Sigma}}^{-1}\right\Vert\\
        \le & \left\Vert \widehat{\bm{\Sigma}}^{-1}\right\Vert
        \left\Vert \widetilde{\bm{\Sigma}}-\widehat{\bm{\Sigma}}\right\Vert \left\Vert \widetilde{\bm{\Sigma}}^{-1}\right\Vert \le Cbp/\sqrt{n}.
    \end{align*}
    Consequently, it yields that
        \begin{align*}
           I_2 & \le  \sqrt{{\rm Tr}\left(\left[\bm{I}_p \otimes (\widehat{\bm{\Sigma}}^{-1}-\widetilde{\bm{\Sigma}}^{-1})\right]^2
    \frac{1}{n-m-b+1}\sum_{i=b+1}^{n-m+1}\left(\frac{1}{\sqrt{m}}\sum_{j=i}^{i+m-1}\bm{z}_j\right)
    \left(\frac{1}{\sqrt{m}}\sum_{j=i}^{i+m-1}\bm{z}_j^\top\right)\right)}\\
    & \le \left\Vert \bm{I}_p \otimes (\widehat{\bm{\Sigma}}^{-1}-\widetilde{\bm{\Sigma}}^{-1})\right\Vert \sqrt{\frac{1}{n-m-b+1}
    \sum_{i=b+1}^{n-m+1}\left|\frac{1}{\sqrt{m}}\sum_{j=i}^{i+m-1}
    \bm{z}_j\right|_F^2}\\
    & \le \frac{Cbp}{\sqrt{n}}\cdot b^{1/2}p \le Cb^{3/2}p^{2}n^{-1/2}.
        \end{align*}
        As a result, $I_1+I_2 \le Cb^{3/2}p^{5/2}n^{-1/2+2/q}$. Now going back to \eqref{error_i}, we choose $\sigma=\bigO(b^{5/8}pn^{-1/4+1/q})$ to obtain the optimal error bound 
        $I \le Cb^{7/8}p^{3/2}n^{-1/4+1/q}$. For the term $VI$, we denote $\widehat{\bm{G}}_{n,j}=\bm{I}^\ast {\rm diag}(\widehat{\bm{Z}}_{n,j}^m)\bm{E}^\ast$ where $\bm{I}^\ast \in \mathbb{R}^{p\times p^2}$ with its $k$th $p$ squared blocks being identity matrices $\bm{I}_p$, $\widehat{\bm{Z}}_{n,j}^m=\frac{1}{n-m-b+1}\sum_{k=b+1}^{n-m+1}(\frac{1}{\sqrt{m}}\sum_{i=k}^{k+m-1} \widehat{\bm{x}}_{i-j}\otimes \widehat{\bm{\epsilon}}_i)\odot (\frac{1}{\sqrt{m}}\sum_{i=k}^{k+m-1} \widehat{\bm{x}}_{i-j}\otimes \widehat{\bm{\epsilon}}_i)$ with $\odot$ representing the Hadamard (elementwise) product for matrices and $\bm{E}^\ast \in \mathbb{R}^{p^2\times p}$ has 1 at locations $\{(k-1)p+1,\ldots,kp\}$ on each $k$th column. Then using the fact $\min_{1\le j\le b}\inf_{u,v\in [0,1]}w_{n,j}(u,v)\ge \kappa_2$ and the Markov inequality, we can obtain
        {\footnotesize
        \begin{align*}
    &\EE\left[\max_{1\le j\le b}
    \left|\bm{E}_j^\top
    \left[\widehat{\bm{A}}_f(v)-\bm{A}_f(v)\right]
    \widehat{\bm{\Sigma}}^{-1}\widetilde{\bm{I}}{\rm diag}(\widehat{\bm{Z}}_n^{\ast,b})\widetilde{\bm{E}}\widehat{\bm{\alpha}}_{f+}(u)\big/w_{n,j}\right|_{\mathcal{L}^2}^2\mid 
    \Upsilon_{b+1}^n \right]\\
    \le & \EE  \left\{ \max_{1\le j\le b}\int_0^1\int_0^1
    {\rm Tr}\left(\bm{E}_j^\top
    \left[\widehat{\bm{A}}_f(v)-\bm{A}_f(v)\right]
    \widehat{\bm{\Sigma}}^{-1}\widetilde{\bm{I}}{\rm diag}(\widehat{\bm{Z}}_n^{\ast,b})\widetilde{\bm{E}}
    \widehat{\bm{\alpha}}_{f+}(u)
    \widehat{\bm{\alpha}}_{f+}^\top(u)
    \widetilde{\bm{E}}^\top
    {\rm diag}(\widehat{\bm{Z}}_n^{\ast,b})
    \widetilde{\bm{I}}^\top
    \widehat{\bm{\Sigma}}^{-1}
    \left[\widehat{\bm{A}}_f(v)-\bm{A}_f(v)\right]^\top\bm{E}_j
    \right)\dee u\dee v\mid  \Upsilon_{b+1}^n \right\}\\
    \le & Cb \left|\bm{1}^\top {\rm diag}\{(1/\widehat{f}_1-1/f_1)^2,\ldots,(1/\widehat{f}_p-1/f_p)^2\}\widehat{\bm{G}}_{n,j}
    {\rm diag}(\widehat{f}_1^2,\ldots,\widehat{f}_p^2)\bm{1}\right| \le Cbp^2/\sqrt{n}.
\end{align*}
}   
By the Markov inequality, we have 
$$VI \le \frac{\EE\left[\max_{1\le j\le b}
    \left|\bm{E}_j^\top
    \left[\widehat{\bm{A}}_f(v)-\bm{A}_f(v)\right]
    \widehat{\bm{\Sigma}}^{-1}\widetilde{\bm{I}}{\rm diag}(\widehat{\bm{Z}}_n^{\ast,b})\widetilde{\bm{E}}\widehat{\bm{\alpha}}_{f+}(u)\right|_{\mathcal{L}^2}^2\Big|
    \Upsilon_{b+1}^n
    \right]}{\kappa_2^2\delta_1^2}
    \le \frac{Cb^2p^4}{n\delta_1^2}.$$
In a similar fashion, we can obtain that
{\footnotesize
        \begin{align*}
    &\EE\left[\max_{1\le j\le b}
    \left|\bm{E}_j^\top\bm{A}_f(v)
    \widehat{\bm{\Sigma}}^{-1}\widetilde{\bm{I}}{\rm diag}(\widehat{\bm{Z}}_n^{\ast,b})\widetilde{\bm{E}}\left[
    \widehat{\bm{\alpha}}_{f+}(u)-\bm{\alpha}_{f+}(u)\right]\big/w_{n,j}\right|_{\mathcal{L}^2}^2\Big|
    \Upsilon_{b+1}^n \right]\\
    \le & \EE  \left\{ \max_{1\le j\le b}\int_0^1\int_0^1
    {\rm Tr}\left(\bm{E}_j^\top
    \bm{A}_f(v)\widehat{\bm{\Sigma}}^{-1}
    \widetilde{\bm{I}}{\rm diag}(\widehat{\bm{Z}}_n^{\ast,b})\widetilde{\bm{E}}
    \left[\widehat{\bm{\alpha}}_{f+}(u)-\bm{\alpha}_{f+}(u)
    \right]\left[\widehat{\bm{\alpha}}_{f+}(u)-\bm{\alpha}_{f+}(u)\right]^\top
    \widetilde{\bm{E}}^\top
    {\rm diag}(\widehat{\bm{Z}}_n^{\ast,b})
    \widetilde{\bm{I}}^\top
    \widehat{\bm{\Sigma}}^{-1}
    \bm{A}_f^\top(v)\bm{E}_j
    \right)\dee u\dee v\Big| \Upsilon_{b+1}^n \right\}\\
    \le & Cb \left|\bm{1}^\top {\rm diag}(1/f_1^2,\ldots,1/f_p^2)\widehat{\bm{G}}_{n,j}
    {\rm diag}((\widehat{f}_1-f_1)^2,\ldots,(\widehat{f}_p-f_p)^2)\bm{1}\right| \le Cbp^2/\sqrt{n}.
\end{align*}
}  
Therefore, 
$$VII \le \frac{\EE\left[\max_{1\le j\le b}
    \left|\bm{E}_j^\top
    \bm{A}_f(v)\widehat{\bm{\Sigma}}^{-1}\widetilde{\bm{I}}{\rm diag}(\widehat{\bm{Z}}_n^{\ast,b})\widetilde{\bm{E}}\left[
    \widehat{\bm{\alpha}}_{f+}(u)-\bm{\alpha}_{f+}(u)\right]\right|_{\mathcal{L}^2}^2\Big|
    \Upsilon_{b+1}^n
    \right]}{\kappa_2^2\delta_2^2}
    \le \frac{Cb^2p^4}{n\delta_2^2}.$$
    As a result, we choose $\delta_1=\delta_2=\bigO\big((b^2p^4/n)^{1/3}\big)$ by balancing the terms $VI, VII$ and $II$. Then it turns out that $$II+VI+VII \le Cb^{2/3}p^{4/3}n^{-1/3}.$$
    For the last distance of the term $V$, we follow the proof of \citep[Lemma D.6]{liu2025wasserstein}, first recall the convex set $$\widetilde{T}_x^{w_n}=\left\{\bm{S}\in \mathbb{R}^{bp^2}:
\frac{1}{b}\sum_{j=1}^{b} \left|\bm{E}_j^\top\bm{A}_f(v)
\widetilde{\bm{I}}{\rm diag} (\bm{S})\widetilde{\bm{E}}\bm{\alpha}_{f+}(u)\big/w_{n,j}(u,v)\right|_{\mathcal{L}^2}
	\le x\right\},$$ and $\widetilde{\mathcal{T}}_n$ is the collection of the convex sets in $\mathbb{R}^{bp^2}$. Further we define $T^{\epsilon}$ and $T^{-\epsilon}$ for $T\in \widetilde{\mathcal{T}}_n$ as 
    $$T^\epsilon=\{\bm{S}\in \mathbb{R}^{bp^2}: \text{dist}(\bm{S},T)\le \epsilon\},\qquad T^{-\epsilon}
    =\{\bm{S}\in \mathbb{R}^{bp^2}: \text{dist}(\bm{S},\mathbb{R}^{bp^2}\backslash T )> \epsilon\}.$$ Let $\bm{V}_n:={\rm vec}(\bm{\Sigma}^{-1}
    \widetilde{\bm{I}}{\rm diag}(\bm{Z}_n^b)\widetilde{\bm{E}})$ and $\widetilde{\bm{V}}_n={\rm vec}(\widetilde{\bm{\Sigma}}^{-1}
    \widetilde{\bm{I}}{\rm diag}(\bm{Z}_n^b)\widetilde{\bm{E}})$. Then we observe that
    \begin{align*}
        V=&\sup_{T\in \widetilde{\mathcal{T}}_n} \left|
        \Pr(\widetilde{\bm{V}}_n\in T)-
        \Pr(\bm{V}_n\in T)\right|\\
        \le & \Pr(|\widetilde{\bm{V}}_n-\bm{V}_n|> \epsilon)
        +\sup_{T\in \widetilde{\mathcal{T}}_n}
        \left\{\Pr(\bm{V}_n \in T^\epsilon\backslash T), 
        \Pr(\bm{V}_n \in T\backslash T^{-\epsilon})\right\}\\
        :=& V_1 + V_2.
    \end{align*}
    For the first term $V_1$, we have
    \begin{align}\label{error_v1}
    V_1 &\le \frac{1}{\epsilon} \EE\left|\bm{I}_p\otimes (\widetilde{\bm{\Sigma}}^{-1}-\bm{\Sigma}^{-1})\bm{Z}_n^b\right| \notag\\
    &\le \frac{1}{\epsilon} \left\Vert \bm{I}_p \otimes (\widetilde{\bm{\Sigma}}^{-1}-\bm{\Sigma}^{-1})\right\Vert \EE\left| \frac{1}{\sqrt{n}}\sum_{i=b+1}^n \bm{z}_i\right| \notag\\
    &\le \frac{C}{\epsilon} b^{3/2}p^{2}\log(n)/\sqrt{n}.
    \end{align}
    On the other hand, notice that
    $$\Pr(\bm{V}_n \in T^\epsilon\backslash T)=\Pr(\bm{V}_n \in T^\epsilon)- \Pr(\bm{V}_n \in T).$$ Consider the Gaussian counterpart $\bm{V}_n^u:={\rm vec}(\bm{\Sigma}^{-1}
    \widetilde{\bm{I}}{\rm diag}(\bm{U}_n^b)\widetilde{\bm{E}})$, then we have
    \begin{align*}
        & \Pr(\bm{V}_n \in T^{\epsilon}\backslash T)-
        \Pr(\bm{V}_n^u \in T^{\epsilon}\backslash T)\\
        \le & \left|\Pr(\bm{V}_n \in T^{\epsilon})-
        \Pr(\bm{V}_n^u \in T^{\epsilon})\right|+\left|
        \Pr(\bm{V}_n \in T)-\Pr(\bm{V}_n^u \in T)\right|\\
        \le & 2\mathcal{D}(\bm{V}_n, \bm{V}_n^u).
    \end{align*}
    Similarly,
    $$\Pr(\bm{V}_n \in T\backslash T^{-\epsilon})-
        \Pr(\bm{V}_n^u \in T\backslash T^{-\epsilon})
        \le  2\mathcal{D}(\bm{V}_n, \bm{V}_n^u).$$ Then by \cite{Bentkus03}, it yields that
        $$V_2 \le \sup_{T \in \widetilde{\mathcal{T}}_n}\{\Pr(\bm{V}_n^u \in T^{\epsilon}\backslash T),
        \Pr(\bm{V}_n^u \in T\backslash T^{-\epsilon})\} + 4\mathcal{D}(\bm{V}_n,\bm{V}_n^u)\le 4\left((bp^2)^{1/4}\epsilon + \mathcal{D}(\bm{V}_n,\bm{V}_n^u)\right).$$
        Putting the above results together and optimizing the error bound, we conclude that 
        $$V\le C\left(\min\left\{
        (bp^2)^{7/4}n^{-\frac{1}{2}+\frac{9}{2q}+\frac{2}{\tau-1}}, (bp^2)^{\frac{5}{6}} n^{\frac{2}{3 r}-\frac{1}{3}}(\log n)^{\frac{1}{3}}\right\} + b^{7/8}p^{5/4}n^{-1/4}\log^{1/2}(n)\right).$$
	As a result, it follows that  
    {\small
    \begin{equation}\label{final_error}
    \widehat{\mathcal{D}}(\widehat{\bm{Z}}_n^b,\bm{Z}_n)\le C\left(\min\left\{
        (bp^2)^{7/4}n^{-\frac{1}{2}+\frac{9}{2q}+\frac{2}{\tau-1}}, (bp^2)^{\frac{5}{6}} n^{\frac{2}{3 r}-\frac{1}{3}}(\log n)^{\frac{1}{3}}\right\}+b^{7/8}p^{3/2}n^{-1/4+1/q}+bp^{17/6}
        n^{-1/3}h_n\right).
        \end{equation}
        }
        Finally, we need to bound the distance between $\Pr(|\widehat{\bm{B}}_n(u,v)|_{\mathcal{L}^2, \bm{w}_n}\le x \mid \Upsilon_{b+1}^n)$ and $\Pr(\sqrt{n}|\widehat{\bm{\psi}}-\bm{\psi}|_{\mathcal{L}^2,\bm{w}_n})$. Note that $$\sqrt{n} \left[\widehat{\bm{\psi}}(u,v)-\bm{\psi}(u,v)\right]
        = \widetilde{\bm{B}}_n^{{\rm oracle},1}
        + \bm{A}_f(v) \widetilde{\bm{\Sigma}}^{-1}\left(\frac{1}{\sqrt{n}}\sum_{i=b+1}^n \bm{x}_i^{(b)}\bm{R}^\top\right)\bm{\alpha}_{f+}(u),$$
         where $\Vert \bm{R}\Vert_2 = \bigO_\Pr(b^{-\tau+2}(\log b)^{\tau-1})$.
         Similarly given a positive sequence $\delta_3=\delta_{3,n}$ that converges to 0 as $n\to \infty$, we have 

         {\footnotesize
\begin{align*}
    &\sup_{x\in\mathbb{R},\bm{w}_n\in\mathcal{W}}\left|
    \Pr\left(\left|\widehat{\bm{B}}_n(u,v)\right|_{\mathcal{L}^2,
    \bm{w}_n}\le x \mid  \Upsilon_{b+1}^n\right)-\Pr\left(\sqrt{n}\left|\widehat{
    \bm{\psi}}-\bm{\psi}\right|_{\mathcal{L}^2,\bm{w}_n}\le x\right)\right| \\
    \le & \sup_{x\in\mathbb{R},\bm{w}_n\in\mathcal{W}}\left|
    \Pr\left(\left|\widehat{\bm{B}}_n(u,v)\right|_{\mathcal{L}^2,
    \bm{w}_n}\le x+\delta_3 \mid  \Upsilon_{b+1}^n\right)-\Pr\left(\left|\widetilde{\bm{B}}_n^{\rm oracle,1}(u,v)\right|_{\mathcal{L}^2,\bm{w}_n}\le x + \delta_3\right)\right|\\
    &\quad + \sup_{\bm{w}_n\in\mathcal{W}}\left|
    \Pr\left(\left|\bm{A}_f(v)\widetilde{\bm{\Sigma}}^{-1}\left(
    \sum_{i=b+1}^n \bm{x}_i^{(b)}\bm{R}^\top/\sqrt{n}\right)\bm{\alpha}_{f+}(u)\right|_{\mathcal{L}^2,
    \bm{w}_n}\ge \delta_{3} \right)\right| + C\delta_3\\
    \le & C\left(\min\left\{
        (bp^2)^{7/4}n^{-\frac{1}{2}+\frac{9}{2q}+\frac{2}{\tau-1}}, (bp^2)^{\frac{5}{6}} n^{\frac{2}{3 r}-\frac{1}{3}}(\log n)^{\frac{1}{3}}\right\}+b^{7/8}p^{3/2}n^{-1/4+1/q}+bp^{17/6}
        n^{-1/3}h_n\right)\\
        & \quad + \sup_{\bm{w}_n\in\mathcal{W}}\left|
    \Pr\left(\left|\bm{A}_f(v)\widetilde{\bm{\Sigma}}^{-1}\left(
    \sum_{i=b+1}^n \bm{x}_i^{(b)}\bm{R}^\top/\sqrt{n}\right)\bm{\alpha}_{f+}(u)\right|_{\mathcal{L}^2,
    \bm{w}_n}\ge \delta_{3} \right)\right| + C\delta_3\\
    =&: C\left(\min\left\{
        (bp^2)^{7/4}n^{-\frac{1}{2}+\frac{9}{2q}+\frac{2}{\tau-1}}, (bp^2)^{\frac{5}{6}} n^{\frac{2}{3 r}-\frac{1}{3}}(\log n)^{\frac{1}{3}}\right\}+b^{7/8}p^{3/2}n^{-1/4+1/q}+bp^{17/6}
        n^{-1/3}h_n\right) + A_1 + A_2.
    \end{align*}
    }

    By \cref{f_bound}, we obtain that
    \begin{equation*}
    \EE\left[\max_{1\le j\le b}
    \left|\bm{E}_j^\top
    \bm{A}_f(v)
    \widetilde{
    \bm{\Sigma}}^{-1}
    \left(\sum_{i=b+1}^n\bm{x}_i^{(b)}\bm{R}^\top/\sqrt{n}\right)\bm{\alpha}_{f+}(u)\big/w_{n,j}\right|_{\mathcal{L}^2}^2\right] \le Cp^{d+5/2}b^{-\tau+3}(\log b)^{\tau-1}.
\end{equation*}  
By the Markov inequality, we have 
$$A_1 \le \frac{\EE\left[\max_{1\le j\le b}
    \left|\bm{E}_j^\top
    \bm{A}_f(v)
    \widetilde{
    \bm{\Sigma}}^{-1}
    \left(\sum_{i=b+1}^n\bm{x}_i^{(b)}\bm{R}^\top/\sqrt{n}\right)\bm{\alpha}_{f+}(u)\big/w_{n,j}\right|_{\mathcal{L}^2}^2\right]}{\kappa_2^2\delta_3^2}
    \le \frac{Cp^{d+5/2}b^{-\tau+3}(\log b)^{\tau-1}}{\delta_3^2}.$$
    By balancing $A_1$ and $A_2$, we obtain the approximation error as $\bigO\left(p^{d/3+5/6}b^{-\tau/3+1}(\log b)^{(\tau-1)/3}\right)$. Once we choose $b\asymp n^{\frac{1}{\tau-3}}$ and $d<6$, the latter error bound is negligible compared to the error bounds in \eqref{final_error}. Combining the above arguments completes the proof.
         $\hfill \square$
\medskip

\noindent
\textbf{Proof of \cref{consistency_test_far}}. First, we consider the FAR($k$) test as follows:
    $$H_0^{\rm FAR}: 
    \{Y_i\}_{i=1}^n \text{~follows~a~FAR}(k) ~\text{process}\quad v.s.\quad  H_a^{\rm FAR}: \exists j\in \{k+1,\ldots,b\} \text{~such that~} |\psi_j(u,v)|_{\mathcal{L}^2, w_{n,j}}= c>0,$$
    where $c$ is some finite constant.
    Note that $\widehat{T}_1=\sum_{j=k+1}^b \left|\widehat{\psi}_j(u,v)\right|_{\mathcal{L}^2,\widehat{w}_{n,j}}$, then it is obvious to find that $\Pr(\sqrt{n} \widehat{T}_1>
    \widehat{q}_{1-\alpha}^{1,\ast} \mid H_0^{\text{FAR}}) \to \alpha$  according to \cref{empirical_thm}. 
    
    On the other hand, if $H_a^{\text{FAR}}$ holds, we have $\widehat{T}_1=\bigO_\Pr(1)$. Then it suffices to prove $\widehat{q}_{1-\alpha}^{1,\ast}/\sqrt{n} \to 0$ as $n\to \infty$. In fact, $\widehat{q}_{1-\alpha}^{1,\ast}/\sqrt{n}$ has the same rate as $\sum_{j=k+1}^b |\widehat{\psi}_j(u,v)-\psi_j(u,v)|_{\mathcal{L}^2,\widehat{w}_{n,j}}$ on the event $\mathcal{A}_n$. By \cref{esti_consistency} and \cref{f_bound}, we notice that for each $j$
    \begin{align}
        \left| \widehat{\psi}_j(u,v)-\psi_j(u,v)\right|_{\mathcal{L}^2,\widehat{w}_{n,j}} \le C &\sqrt{\int_0^1\int_0^1\left[\widehat{\psi}_j(u,v)-\psi_j(u,v)\right]^2\dee u\dee v}
        \le C \left| \widehat{\bm{\Psi}}_j-\bm{\Psi}_j\right|_{F} \notag \\
        &= \left|\text{diag}(f_1,\ldots,f_p)\left(\widehat{\bm{\Phi}}_j-\bm{\Phi}_j\right)\text{diag}(1/f_1,\ldots,1/f_p)\right|_F \notag
        \\
         & \le Cp^{d+3/2-\theta}\sqrt{\frac{bp^3\log (n)}{n}}. \label{convergence_rate}
    \end{align}
    Consequently, we have $\widehat{q}_{1-\alpha}^{1,\ast}/\sqrt{n}=\bigO_\Pr\left(b\left|\widehat{\psi}_j-\psi_j\right|_{\mathcal{L}^2,\widehat{w}_{n,j}}\right)=o_\Pr(1)$. Therefore, as $n\to \infty$, the probability of $\sqrt{n}\widehat{T}_1>\widehat{q}_{1-\alpha}^{1,\ast}$ converges to 1 under $H_a^{\text{FAR}}$. 
    $\hfill \square$
    \medskip

\noindent
\textbf{Proof of \cref{consistency_test_farma}}. Consider the FARMA model test:
    $$H_0^{\rm FARMA}: 
    \{Y_i\}_{i=1}^n \text{~follows~a~FARMA}(p^\ast,q^\ast) ~\text{process}\quad v.s.\quad  H_a^{\rm FARMA}: \{Y_i\}_{i=1}^n \text{~is~not~a~FARMA}(p^\ast,q^\ast) ~\text{process}$$
     Recall the test statistic  $\widehat{T}_2=\sum_{j=1}^b \left|\widehat{\psi}_j(u,v)-\widetilde{\psi}_j(u,v)\right|_{\mathcal{L}^2,\widehat{w}_{n,j}}$ where $\widetilde{\psi}_j(u,v)$ is the  coefficient function that was converted from the model under $H_0^\text{FARMA}$. Now, it suffices to compare the rates of $\left|\widehat{\psi}_j-\psi_j\right|_{\mathcal{L}^2,\widehat{w}_{n,j}}$ and $\left|\widetilde{\psi}_j-\psi_j\right|_{\mathcal{L}^2,\widehat{w}_{n,j}}$ where $\psi_j$ is the AR coefficient function induced by the true model. If $H_0^{\text{FARMA}}$ holds, notice that $\left|\widetilde{\psi}_j-\psi_j\right|_{\mathcal{L}^2,\widehat{w}_{n,j}}=\bigO_\Pr(n^{-1/2})$ (\cite{lutkepohl2013introduction}) which is faster than the convergence rate of $\widehat{\psi}_j$ in \eqref{convergence_rate}. As a result, we conclude that $\Pr(\sqrt{n} \widehat{T}_2>\widehat{q}_{1-\alpha}^{2,\ast} \mid H_0^{\text{FARMA}}) \to \alpha$  based on \cref{empirical_thm}. 
     
     When $H_a^{\text{FARMA}}$ holds, there exist some $j\in \{1,\ldots,b\}$ and some constant $c>0$ such that $\Pr(|\widetilde{\psi}_j-\psi_j|_{\mathcal{L}^2,\widehat{w}_{n,j}}\ge c)> 0$. In fact, $\widetilde{\psi}_j$ converges to a pseudo-true parameter in probability, which implies $\lim_{n\to \infty}\Pr(|\widetilde{\psi}_j-\psi_j|_{\mathcal{L}^2,\widehat{w}_{n,j}}\ge c_0)= 1$ where $0<c_0\le c$. Then on the event $\mathcal{A}_n$, we can obtain that as $n\to \infty$,
    \begin{align*}
        & \Pr(\sqrt{n} \widehat{T}_2>\widehat{q}_{1-\alpha}^{2,\ast} \mid H_a^{\text{FARMA}}) = 1- \Pr(\sqrt{n} \widehat{T}_2\le \widehat{q}_{1-\alpha}^{2,\ast} \mid H_a^{\text{FARMA}})\\
         \ge & 1- \Pr\left(\sum_{j=1}^b \left| \left|\widehat{\psi}_j-\psi_j\right|_{\mathcal{L}^2,\widehat{w}_{n,j}} -  \left|\widetilde{\psi}_j-\psi_j\right|_{\mathcal{L}^2,\widehat{w}_{n,j}}\right| \le \widehat{q}_{1-\alpha}^{2,\ast}/\sqrt{n}  \mid H_a^{\text{FARMA}}\right)\to 1.
    \end{align*} 
$\hfill \square$
\medskip

\noindent
\textbf{Proof of \cref{consistency_test_sep}}. Here, we consider the separability test below:
    $$H_{0}^{{\rm sep},j}: \psi_j(u,v)=C_jg_j(u)h_j(v) \quad v.s.\quad H_{a}^{{\rm sep},j}: \psi_j(u,v)\neq C_jg_j(u)h_j(v).$$
    Recall the test statistic  $\widehat{T}_{3,j}=
    \left|\widehat{\psi}_j^{\rm sep}
    (u,v)-\widehat{\psi}_j(u,v)\right|_{\mathcal{L}^2,\widehat{w}_{n,j}}$ where $
\widehat{\psi}_j^{\rm sep}(u,v):=
\widehat{C}_j\widehat{g}_j(u)\widehat{h}_j(v)$. Similar to the FARMA test, we first need to compare the rates of $\left|\widehat{\psi}_j-\psi_j\right|_{\mathcal{L}^2,\widehat{w}_{n,j}}$ and $\left|\widehat{\psi}_j^{\rm sep}-\psi_j\right|_{\mathcal{L}^2,\widehat{w}_{n,j}}$ where $\psi_j$ is the AR coefficient function converted by the true model. Note that under $H_0^{{\rm sep}, j}$,
    \begin{align*}
        &\left| \widehat{\psi}_j^{\rm sep}(u,v) - \psi_j(u,v)\right|_{\mathcal{L}^2,\widehat{w}_{n,j}}\\
        =& \left| \widehat{C}_j\widehat{g}_j(u)
\widehat{h}_j(v) - C_jg_j(u)h_j(v)\right|_{\mathcal{L}^2,\widehat{w}_{n,j}}\\
        =&\left| (\widehat{C}_j-C_j)\widehat{g}_j(u)
        \widehat{h}_j(v) + C_j[\widehat{g}_j(u)-g_j(u)]
        \widehat{h}_j(v) + C_j g_j(u)
        [\widehat{h}_j(v)-h_j(v)]\right|_{\mathcal{L}^2,\widehat{w}_{n,j}}\\
        \le & \left| (\widehat{C}_j-C_j)\widehat{g}_j(u)
        \widehat{h}_j(v)\right|_{\mathcal{L}^2,\widehat{w}_{n,j}}+ \left|
        C_j[\widehat{g}_j(u)-g_j(u)]
        \widehat{h}_j(v)\right|_{\mathcal{L}^2,\widehat{w}_{n,j}}  + \left|C_j g_j(u)
        [\widehat{h}_j(v)-h_j(v)]\right|_{\mathcal{L}^2,\widehat{w}_{n,j}}\\
        = & I+ II +III
    \end{align*}
Denote $\widehat{\Psi}_{j,kl}$ and $\Psi_{j,kl}$ as the $(k,l)$th components of $\widehat{\bm{\Psi}}_{j}$ and $\bm{\Psi}_j$, respectively. Similarly, let $\widehat{\Phi}_{j,kl}$ and $\Phi_{j,kl}$ be the $(k,l)$th components of $\widehat{\bm{\Phi}}_{j}$ and $\bm{\Phi}_j$, respectively. For the estimator $\widehat{C}_j$, we have
\begin{align*}
    [\widehat{C}_j - C_j]^2 &= 
    \left[\int_0^1 \int_0^1 \left(\widehat{\psi}_j(u,v)-\psi_j(u,v)\right)\dee u \dee v\right]^2 =(\widehat{\Psi}_{j,11}-\Psi_{j,11})^2
    \\
    & = (\widehat{\Phi}_{j,11}-\Phi_{j,11})^2=\left|\widehat{\bm{\Phi}}_j-\bm{\Phi}_j
\right|
_{\max}^2
\le \left\Vert \widehat{\bm{\Phi}}_j-\bm{\Phi}_j
\right\Vert^2 \le \frac{Cbp^2\log n}{n}.
    \end{align*}

For the estimated separate functions $\widehat{g}_j(u)$ and $\widehat{h}_j(v)$, note that
    $$\int_0^1 \widehat{g}_j^2(u)\dee u \le C\left|\widehat{\bm{\Psi}}_j\right|_F^2, \quad \int_0^1 \widehat{h}_j^2(v)\dee v\le C\left| \widehat{\bm{\Psi}}_j\right|_F^2.$$ Then we have
    $$I \le C\sqrt{\frac{bp^2\log n}{n}} \cdot \max\left\{\left|\bm{\Psi}_j\right|_F, \left| \widehat{\bm{\Psi}}_j-\bm{\Psi}_j
\right|_{F}\right\}.$$
Similarly, one can obtain
    $$II =III \le C\sqrt{\frac{bp^2\log n}{n}} \cdot \max\left\{\left| \bm{\Psi}_j\right|_F, \left| \widehat{\bm{\Psi}}_j-\bm{\Psi}_j
\right|_{F}\right\}.$$
Observe that the summability of the functional AR kernels $\sum_{j=1}^\infty |\psi_j(u,v)|_{\mathcal{L}^2,w_{n,j}}<\infty$ immediately yields $\sum_{j=1}^\infty |\bm{\Psi}_j|_F <\infty$. Then comparing to \eqref{convergence_rate}, we obtain that
$$\left|\widehat{\psi}_j^{\rm sep}-\psi_j\right|_{\mathcal{L}^2,\widehat{w}_{n,j}} =o_\Pr\left(\left|\widehat{\psi}_j-\psi_j\right|_{\mathcal{L}^2,\widehat{w}_{n,j}}\right).$$ 
As a result, one can conclude that $\Pr(\sqrt{n} \widehat{T}_{3,j}>\widehat{q}_{1-\alpha}^{3,j,\ast} \mid H_0^{\text{sep},j}) \to \alpha$ by \cref{empirical_thm}. On the other hand, when $H_a^{\text{sep},j}$ holds, there exists some constant $c>0$ such that $\Pr(|\widehat{\psi}_j^{\rm sep}-\psi_j|_{\mathcal{L}^2,\widehat{w}_{n,j}}\ge c)> 0$.
Similarly on the event $\mathcal{A}_n$ and with $n\to \infty$, we have
    \begin{align*}
        & \Pr(\sqrt{n} \widehat{T}_{3,j}>
        \widehat{q}_{1-\alpha}^{3,j,\ast} \mid H_a^{\text{sep},j}) = 1- \Pr(\sqrt{n} \widehat{T}_{3,j}\le \widehat{q}_{1-\alpha}^{3,j,\ast} \mid H_a^{\text{sep},j})\\
         \ge & 1- \Pr\left( \left| \left|\widehat{\psi}_j^{\rm sep}-\psi_j\right|_{\mathcal{L}^2,\widehat{w}_{n,j}} -  \left|\widehat{\psi}_j-\psi_j\right|_{\mathcal{L}^2,\widehat{w}_{n,j}}\right| \le \widehat{q}_{1-\alpha}^{3,j,\ast}/\sqrt{n}  \mid H_a^{\text{sep},j}\right)\to 1.
    \end{align*}
    $\hfill \square$
 
\section{Commonly used basis functions}\label{app_basis}
\def\theequation{B.\arabic{equation}}	
\setcounter{equation}{0}

Here, we list two commonly used orthogonal basis functions which are employed in our paper.

\begin{example}[Fourier bases]
    For $t\in[0,1]$,
	$$\alpha_k(t)=\begin{cases}
		1, &k=0,\\
		\sqrt{2}\cos(k\pi t), &k\ge 1.
        \end{cases}$$
\end{example}

	\begin{example}[Normalized Legendre polynomials \cite{Bell04}]\label{example1}
		The Legendre polynomial of degree $n$ can be obtained using Rodrigue's formula 
		$$P_n(t)=\frac{1}{2^n n!}\frac{\dee^n}{\dee t^n}(t^2-1)^n,~-1\le t \le 1.$$ For $t\in [0,1]$, the normalized Legendre polynomials turn out to be
		\begin{align*}
		\alpha_k(t)=\begin{cases}
		1,&k=0,\\
		\sqrt{2k+1}P_k(2t-1),&k>0.
		\end{cases}
		\end{align*}
	\end{example}

\section{Examples of Functional Time Series Models}\label{app_examples}
\def\theequation{D.\arabic{equation}}	
\setcounter{equation}{0}

	Here, we will provide two models to  illustrate how to calculate the physical dependence measure $\delta_x(l,q)$ defined in \cref{ass_dep} of the main paper.
	\subsection{FMA$(\infty)$ model}
	\begin{example}[Functional MA$(\infty)$ model] \label{exam1}
		Let $\eta_i(v)$ be i.i.d. centered and continuous Gaussian random functions with $\sup_{v\in[0,1]}\EE\Vert\eta_i(v)\Vert_2<\infty$.  For each integer $m\ge 0$, let $\beta_m(u,v)=a_m\beta_m^\ast(u,v)$ where $\{a_m\}$ is a positive deterministic sequence with $\sum_{m=0}^\infty a_m<\infty$ and $\beta_m^\ast(\cdot,\cdot)$ is a $\mathcal{C}([0,1]^2)$ deterministic function such that $|\beta^*_m(u,v)|\le C$ for all $u,v$ and $m$ and some finite constant $C$. Consider the functional MA$(\infty)$ model,
		\begin{equation}\label{exp1}
		Y_i(u)=\sum_{m=0}^\infty \int_0^1 \beta_m(u,v)\eta_{i-m}(v)\dee v.
		\end{equation}
		Let  $C(u,v):=\EE(\eta_i(u)\eta_i(v))$ be the covariance function of $\eta_i(u)$. Let $v_1(u), v_2(u),\cdots$ and the corresponding $\lambda_1\ge \lambda_2\ge \cdots$ be the eigenfunctions and eigenvalues of $C(u,v)$. By the basis expansion method, we can write $\beta_m^\ast(u,v)=
		\sum_{j=1}^\infty
        \sum_{k=1}^\infty b_{j,k}^m\alpha_j(u)v_k(v)$ and $\eta_i(v)=\sum_{k=1}^\infty\eta_{i,k}v_k(v)$, 
        then we have $$x_{i,j}f_j=\sum_{m=0}^\infty a_m\left(\sum_{k=1}^\infty
		b_{j,k}^m\eta_{i-m,k}\right).$$ Here $f^2_j=\sum_{m=0}^\infty a_m^2\theta_{jm}$ is the variance of the random coefficient $\widetilde{x}_{i,j}$, where $\theta_{jm}:=\sum_{k=1}^\infty
		(b_{j,k}^m)^2 \lambda_k$. Let $\eta_i=(\eta_{i,j})_{j\ge 1}$. Similar to the discussion of \citep[Example 1]{CZ2022}, we obtain that $\delta_x(l,q)=O(a_l)$ for any given $q\ge 2$ if $\sum_{k=0}^\infty a_k^2\theta_{jk}\ge C\theta_{jm}$ for sufficiently large $j$ and $m$. Consequently, the physical dependence measure for the functional MA($\infty$) models $\delta_x(l,q)\le C(l+1)^{-\tau}$ when we choose $a_l=(l+1)^{-\tau}$.
	\end{example}
	
	\subsection{FAR(1) model}
	In this paper, we focus on the discussion when the physical dependence measure is of polynomial decay, that is, \cref{ass_dep} of the main article holds true. However, all our results can be extended to the case when it is of exponential decay
	\begin{equation}\label{expon}
	\delta_x(l,q)\le C\rho^l, 0<\rho<1.
	\end{equation}
	Next, we will demonstrate an example of FAR(1) model to verify this exponential decay \eqref{expon} of the dependence measures. 
	\begin{example}[Functional AR(1) model] \label{exam2}
		Let $\epsilon_i(u)$ be i.i.d. centered and continuous Gaussian random functions with $\sup_{t\in[0,1]}\EE\Vert\epsilon_i(u)\Vert_2<\infty$. Consider the following model \begin{equation}\label{model}
		Y_i(u)=\int_0^1 B(u,v)Y_{i-1}(v)\dee v+\epsilon_i(u),
		\end{equation}
		where $B(u,v): [0,1]^2\to \mathbb{R}$ is a continuous, symmetric function satisfying $\int_0^1\int_0^1B^2(u,v)\dee u\dee v<\infty$ and $\int_0^1\int_0^1B(u,v)x(u)x(v)\dee u\dee v\ge 0$ with any random function $x(u)\in \mathcal{L}^2([0,1])$. Thus $B(u,v)$ is called a symmetric and positive-definite kernel on $[0,1]^2$. Define $C(u,v):=\EE(\epsilon_i(u)\epsilon_i(v))$ as the covariance function of $\epsilon_i(u)$. Let $v_1(u), v_2(u),\cdots$ and the corresponding $\lambda_1\ge \lambda_2\ge \cdots$ be the eigenfunctions and eigenvalues of $C(u,v)$. By the basis expansion method, we can write
		$Y_i(u)=\sum_{k=1}^\infty \widetilde{x}_{i,k}v_k(u),~
		B(u,v)=\sum_{k=1}^\infty b_kv_k(u)v_k(v)$ and 
		$\epsilon_i(u)=\sum_{k=1}^\infty \epsilon_{i,k}v_k(u)$. Consequently we have
        \begin{equation} \label{basis_expan}\widetilde{x}_{i,k}=b_k\widetilde{x}_{i-1,k}+\epsilon_{i,k}=\sum_{m=0}^\infty b_k^m\epsilon_{i-m,k}.
        \end{equation}
		Here denote $f^2_k={\rm Var}(\widetilde{x}_{i,k})=\left(\sum_{m=0}^\infty b_k^{2m}\right)\lambda_k$. If we let $\rho:=\sup_k|b_k|\in (0,1)$, then $f_k^2\ge \lambda_k$. Further observe that $\epsilon_{i,k}$ are independent Gaussian random variables across $k$, hence $(\epsilon_{i-l,k}-\epsilon_{i-l,k}^\ast)$ is normally distributed with mean 0 and variance  $2\lambda_k$, then we have 
		\begin{align*}
		\Vert x_{i,k}-x_{i,k}^\ast\Vert_q&=\Vert b_k^l(\epsilon_{i-l,k}-\epsilon_{i-l,k}^\ast)/f_k\Vert_q\\
		&\le C_q\sqrt{2\lambda_k}b_k^l/
        \sqrt{\lambda}_k\le C\rho^l.
		\end{align*}
	\end{example}

\section{Additional technical results} \label{secD}
\def\theequation{E.\arabic{equation}}
\setcounter{equation}{0}

In this section, we provide additional technical results that further support the theoretical results established in the main article.

\subsection{VAR approximation and related}
Throughout our methodology, the functional AR approximation theory in \cref{approx_ar} is motivated by a VAR approximation result on the scaled multivariate time series $\{\bm{x}_i\}$. Recall the best linear predictions $\widehat{\bm{x}}_i=
\sum_{j=1}^{i-1}\bm{\Phi}_j\bm{x}_{i-j}, i=2,\ldots,n$. Let $\bm{x}_{i-1}^{(i)}=(\bm{x}_{i-1}^\top,\ldots,\bm{x}_1^\top)^\top\in \mathbb{R}^{(i-1)p}$ be a block vector and $\Gamma^{(i)}={\rm Cov}(\bm{x}_{i-1}^{(i)},\bm{x}_{i-1}^{(i)})\in \mathbb{R}^{(i-1)p\times (i-1)p}$ be the covariance matrix of $\bm{x}_{i-1}^{(i)}$. Then denote $\bm{\Phi}^{(i)}=(\bm{\Phi}_1^\top,\ldots,
\bm{\Phi}_{i-1}^\top)^\top \in \mathbb{R}^{(i-1)p\times p}$ and $\bm{\gamma}^{(i)}={\rm Cov}(\bm{x}_{i-1}^{(i)},\bm{x}_i)\in \mathbb{R}
^{(i-1)p\times p}$, we have the Yule-Walker equation
$$\bm{\Phi}^{(i)}=[\bm{\Gamma}^{(i)}]^{-1}
\bm{\gamma}^{(i)}.$$
Now, we state the VAR approximation theory in the proposition below.
\begin{proposition}\label{var_approx}
    Under Assumptions \ref{ass_conti_u}, \ref{weak_depen_components} and \ref{updc}, a rich class of short memory stationary multivariate time series can be well approximated by a stationary white-noise-driven multivariate AR process of slowly diverging order $b$ as
    \begin{equation}\label{approx_var}
        \bm{x}_i=\sum_{j=1}^{\min\{i-1,b\}} \bm{\Phi}_j\bm{x}_{i-j}+\bm{\epsilon}_i+\bigO_\Pr\left(p^{1/2}b^{-\tau+2}(\log b)^{\tau-1}\right).
    \end{equation}
\end{proposition}

To validate the above proposition, we introduce the following lemma that bounds the decay rate of the coefficient matrices $\bm{\Phi}_j, j\ge 1$.

\begin{lemma}\label{phi_rate}
Suppose Assumptions \ref{ass_conti_u}, \ref{weak_depen_components} and \ref{updc} hold, then for the VAR process $\bm{x}_i=\sum_{j=1}^{i-1}
\bm{\Phi}_j\bm{x}_{i-j}+\bm{\epsilon}_i$, there exists some constant $C>0$ such that $$\Vert \bm{\Phi}_j\Vert \le C\left(\frac{j}{\log j +1}\right)^{-\tau+1}, \quad j\ge 1.$$
\end{lemma}

This lemma is an application of \citep[Proposition 2]{CZ2023} to stationary multivariate time series. A careful check of the proof of the latter proposition reveals that the arguments remain valid for stationary processes in the present setting. For brevity, the proof of \cref{phi_rate} is omitted and the proof of \cref{var_approx} is straightforward from \cref{phi_rate}.
\bigskip

\noindent
On the other hand, it is clear that the random vector 
$\bm{\epsilon}_i=(\epsilon_{i,1},...,
\epsilon_{i,p})^\top$ is the best linear prediction error in the VAR model. Denote $\epsilon_{i,k}=\tilde{G}_k(\mathcal{F}_i)$ where $\tilde{G}_k$ is some measurable function for $k=1,...,p$, similar to $G_k$ defined in \cref{def1}. Then we have the following lemma on its physical dependence measure.
	
\begin{lemma}\label{eps_depend}
Suppose Assumptions \ref{ass_conti_u}, \ref{weak_depen_components}--\ref{ass_dep} hold true, denote the physical dependence measure of 
$\{\epsilon_{i,k}\}_{i\in\mathbb{Z}}$ as $\delta_\epsilon(l,q):=\max_{1\le k\le p}\Vert \tilde{G}_k(\mathcal{F}_i)-\tilde{G}_k(\mathcal{F}_{i,l})\Vert_q$ and we assume that $\max_{i,k}\Vert
\epsilon_{i,k}\Vert_q<\infty,~q\ge 2$. Then, there exists some constant $C>0$ such that $\delta_\epsilon(l,q)\le C\sqrt{p}l^{-\tau+2}(\log l+1)^{\tau-1},~l\ge 1$.
	\end{lemma}
	\noindent
	\textit{Proof}:
	Rewrite the VAR model in (12) of the main article as $\bm{x}_i=\sum_{h=1}^{i-1}\bm{\Phi}_{h}\bm{x}_{i-h}+\bm{\epsilon}_i$ and denote $\bm{\Phi}_{h}(k,m)$ as the $(k,m)$th element of the matrix $\bm{\Phi}_{h}$. Moreover, denote $x_{i,k}^\ast=
    G_k(\mathcal{F}_{i,l})$ and $\epsilon_{i,k}^\ast=\tilde{G}_k(\mathcal{F}_{i,l})$ where $\mathcal{F}_{i,l}=(\mathcal{F}_{i-l-1}, \eta_{i-l}^\ast, \ldots, \eta_i)$. On this basis, we can derive that for every $k=1,\ldots,p$,
	\begin{align*}
	& \Vert \epsilon_{i,k}-\epsilon_{i,k}^\ast \Vert_q\\
    =&\left\Vert
	x_{i,k}-x_{i,k}^\ast+\sum_{h=1}^{i-1}\sum_{m=1}^p
	\bm{\Phi}_{h}(k,m)(x_{i-h,m}^\ast-x_{i-h,m})\right\Vert_q\\
	\le &\Vert x_{i,k}-x_{i,k}^\ast\Vert_q+
    \left\Vert\sum_{h=1}^{i-1}
	\sum_{m=1}^p\bm{\Phi}_{h}(k,m)(x_{i-h,m}-x_{i-h,m}^\ast)\right\Vert_q\\
\le & \Vert x_{i,k}-x_{i,k}^\ast\Vert_q +
\sum_{h=1}^{i-1}\Big(\sum_{m=1}^p|\bm{\Phi}_{h}(k,m)|\Big)\max_{1\le m\le p}
\left\Vert x_{i-h,m}-x_{i-h,m}^\ast\right\Vert_q\\
	\le &C(l+1)^{-\tau}+\sum_{h=1}^l
    \sqrt{p}\left\Vert \bm{\Phi}_h\right\Vert
    \left\Vert x_{i-h,m}-x_{i-h,m}^\ast\right\Vert_q\\
    \le &C\sqrt{p}l^{-\tau+2}(\log l+1)^{\tau-1},
	\end{align*} 
    where the last inequality follows by \cref{ass_dep} and Lemma 1. As a result, $\delta_\epsilon(l,q)\le C\sqrt{p}l^{-\tau+2}(\log l+1)^{\tau-1}$.
    $\hfill \square$
\bigskip

Next, we aim to demonstrate the consistency of the coefficient matrix estimators $\{\widehat{\bm{\Phi}}_j\}_{j=1}^b$ in the VAR approximation.

\begin{proposition}\label{esti_consistency}
Suppose Assumptions \ref{updc} and \ref{ass_dep} hold, we have
$$\left\Vert
\widehat{\bm{\Phi}}_j-\bm{\Phi}_j\right\Vert \le 
C\sqrt{bp^2\log(n)/n},\quad j=1,\ldots,b.$$  
\end{proposition}
\noindent
	\textit{Proof}: We will mainly adopt the proof strategy of \citep[Theorem 1]{liu2013probability} to derive the consistency rate for the estimation. Denote 
    $$\bm{S}_n=\frac{1}{n}{\bm{X}^\top \bm{X}}=\frac{1}{n}\sum_{i=b+1}^n
    \bm{x}_i^{(b)}(\bm{x}_i^{(b)})^\top,\quad 
    \bar{\bm{S}}_n=\frac{1}{n}\EE[{\bm{X}^\top \bm{X}}]=\frac{1}{n}\sum_{i=b+1}^n
    \EE\left[\bm{x}_i^{(b)}(\bm{x}_i^{(b)})^\top\right].$$ The following proof consists of two steps: (i) establish that $\bm{S}_n$ is bounded away from zero; (ii) Derive the convergence rate of $\widehat{\bm{\Phi}}_j$.

To investigate the Step (i), we first denote the $j$-dependent sequence as $$\bm{S}_{n,j}=\frac{1}{n}\sum_{i=b+1}^n
\EE[\bm{x}_i^{(b)}(\bm{x}_i^{(b)})^\top|\eta_{i-j-1},...,\eta_{i-1}],$$ namely $\EE[\bm{x}_i^{(b)}(\bm{x}_i^{(b)})^\top|\eta_{i-j-1},...,\eta_{i-1}]$ and $\EE[\bm{x}_{i'}^{(b)}(\bm{x}_{i'}^{(b)})^\top|\eta_{i'-j-1},...,\eta_{i'-1}]$ are independent if $|i-i'|>j$. Then we can decompose the difference between $\bm{S}_n$ and $\bar{\bm{S}}_n$ into three parts: 
\begin{align}\label{decomp_3}
&\bm{S}_n-\bar{\bm{S}}_n\notag \\
=&(\bm{S}_n-\bm{S}_{n,n})+\sum_{j=2}^n(\bm{S}_{n,j}-\bm{S}_{n,j-1})+
(\bm{S}_{n,1}-\bar{\bm{S}}_n)=:{\rm I}+{\rm II}+{\rm III}.
\end{align} 
Now, we start with the first term I. In general, notice that
\begin{align*}
\Vert \bm{S}_n-\bm{S}_{n,j}\Vert &=
\left\Vert\frac{1}{n}
\sum_{i=b+1}^n \left(\bm{x}_i^{(b)}
(\bm{x}_i^{(b)})^\top-\EE[\bm{x}_i^{(b)}
(\bm{x}_i^{(b)})^\top|\eta_{i-j-1},...,\eta_{i-1}]\right)\right\Vert\\
&\le \left|\frac{1}{n}
\sum_{i=b+1}^n \left(\bm{x}_i^{(b)}
(\bm{x}_i^{(b)})^\top-\EE[\bm{x}_i^{(b)}
(\bm{x}_i^{(b)})^\top|\eta_{i-j-1},...,\eta_{i-1}]\right)\right|_F,
\end{align*}
where the second inequality follows by $\Vert\bm{A}\Vert\le |\bm{A}|_F$ for some matrix $\bm{A}$. Then, the above bound of the difference boils down to evaluating its each entry. Denote $y_{i}^{(kl)}$ as the $(k,l)$-th element of the matrix $\bm{x}_i^{(b)}(\bm{x}_i^{(b)})^\top$ for $k,l=1,...,bp$. Similarly, we can write  $y_{i,j}^{(kl)}$ as the $(k,l)$-th element of the matrix $\EE[\bm{x}_i^{(b)}(\bm{x}_i^{(b)})^\top|\eta_{i-j-1},...,\eta_{i-1}]$. Furthermore, observe that $\{y_i^{(kl)}-y_{i,j}^{(kl)}\}_{i=b+1}^n$ is a martingale difference sequence for any $j=1,\ldots,n$.

Let $\delta_y(h,q):=\max_{i,k,l}\Vert y_{i}^{(kl)}-\tilde{y}_{i}^{(kl)}\Vert_q$ be the physical dependence measure of $z_i^{(kl)}$ where $\tilde{z}_i^{(kl)}$ is the coupling random variable of $z_i^{(kl)}$, then we have $\delta_z(h,q)\le C(h+1)^{-\tau}$ for $h\ge 0$ by Assumption 4 of the paper and the definition of $z_{i}^{(kl)}$. Denote $S_n^{(kl)}$ and $S_{n,n}^{(kl)}$ as the $(k,l)$th entries of $\bm{S}_n$ and $\bm{S}_{n,n}$, respectively. Then followed by \citep[Lemma A.1]{LiuLin09}, we have
	$$\left\Vert S_n^{(kl)}-S_{n,n}^{(kl)}\right\Vert_q
    \le C\sum_{h=n}^{\infty}
	\delta_z(h,q)/\sqrt{n}=\bigO(n^{-\tau+1/2}),$$ which also implies $\left|S_n^{(kl)}-S_{n,n}^{(kl)}\right|=\bigO_\Pr(n^{-\tau+1/2})$ and hence $\Vert \bm{S}_n-\bm{S}_{n,n}\Vert\le Cbpn^{-\tau+1/2}$.
 
To deal with the second term II, define $$\Delta_n^{(kl)}:=\sum_{j=2}^n
 \Lambda_{n,j}^{(kl)}:=
	\sum_{j=2}^n\left(\sum_{i=b+1}^n
	(y_{i,j}^{kl}-y_{i,j-1}^{kl})/n\right),~1\le k,l\le bp.$$ 
	Now let $z_{i,j}^{(kl)}=\sum_{h=(i-1)j+b+1}^{(ij+b)\wedge n}(y_{h,j}^{(kl)}-y_{h,j-1}^{(kl)})/n$ where $a\wedge b:=\min(a,b)$ for two real numbers $a$ and $b$. With $l_0=\lfloor \frac{n-b}{j}\rfloor$, we can obtain 
	$$\left|\Lambda_{n,j}^{(kl)}\right|=\left|\sum_{i=1}^{l_0} z_{i,j}^{(kl)}\right|=\left|\sum_{i~{\text is~odd}}z_{i,j}^{(kl)}
	+\sum_{i~{\text is~even}}z_{i,j}^{(kl)}\right|.$$
	Observe that $z_{1,j},z_{3,j},...$ are independent and $z_{2,j},z_{4,j},...$ are also independent by the definition of the $j$-dependent sequence. Since $\{y_{i,j}^{(kl)}-y_{i,j-1}^{(kl)}\}_{i=b+1}^n$ is a martingale difference sequence with respect to $\sigma(\eta_{i-j-1},\ldots,\eta_{i-1})$, it follows by \cref{rosenthal} and the triangle inequality that
 {\small
 \begin{equation}\label{lambda}
	\left\Vert\Lambda_{n,j}^{(kl)}\right\Vert_q\le \frac{14.7q}{\log q}\left\{
    \left(\sum_{i~{\text is~odd}}\left\Vert z_{i,j}^{(kl)}\right\Vert_2^2\right)^{1/2}+
    \left(\sum_{i~{\text is~odd}}\Vert z_{i,j}^{(kl)}\Vert_q^q\right)^{1/q}+
	\left( \sum_{i~{\text is~even}}\left\Vert z_{i,j}^{(kl)}\right\Vert_2^2\right)^{1/2}
	+\left(\sum_{i~{\text is~even}}\Vert z_{i,j}^{(kl)}\Vert_q^q\right)^{1/q}\right\}.
	\end{equation}
 }

Armed with \cref{mar_concentration} and Jensen's inequality, we have
	\begin{align*}
	&\Vert z_{i,j}^{(kl)}\Vert_q^q\le (q-1)^{q/2}j^{q/2}\delta_y^q(j,q)/n^q,\\
	&\Vert z_{i,j}^{(kl)}\Vert_2^2\le j\delta_y^2(j,2)/n^2.	
	\end{align*}
	Thus \cref{lambda} implies that for $2\le j\le n$,
	\begin{align*}
    \left\Vert \Lambda_{n,j}^{(kl)}\right\Vert_q \le& \frac{29.4q}{\log q}\left(\delta_z(j,2)\sqrt{jl_0/2n^2}+
    (l_0/2)^{1/2}(q-1)^{1/2}j^{1/2}\delta_y(j,q)/n\right)\\
    \le& C\left\{(j+1)^{-\tau}/\sqrt{n}+j^{-\tau+1/2-1/q}n^{-1+1/q}
    \right\}.
    \end{align*}
	Consequently, we have $$\Vert \Delta_n^{(kl)}\Vert_q\le \sum_{j=2}^n
    \left\Vert\Lambda_{n,j}^{(kl)}\right\Vert_q=\bigO(1/\sqrt{n}), \quad q\ge 2,$$
    and $$\left\Vert \sum_{j=2}^n (\bm{S}_{n,j}-\bm{S}_{n,j-1})\right\Vert\le Cbpn^{-1/2}.$$
	
	Finally, we focus on the last term III. Note that
	$$\Vert \bm{S}_{n,1}-\bar{\bm{S}}_n\Vert = 
	\left\Vert \frac{1}{n}\sum_{i=b+1}^n \left(\EE[\bm{x}_i^{(b)}
 (\bm{x}_i^{(b)})^\top|\eta_{i-1}]-
	\EE[\bm{x}_i^{(b)}(\bm{x}_i^{(b)})^\top]\right)\right\Vert .$$ Since $\big\{\EE[\bm{x}_i^{(b)}(\bm{x}_i^{(b)})^\top|\eta_{i-1}]-
	\EE[\bm{x}_i^{(b)}(\bm{x}_i^{(b)})^\top]\big\}$ are independent across $i$, then it will be controlled by \cref{lemma_berns}. By elementary calculations, we obtain $R\le C/n$ and $\sigma^2\le Cp/n$, then
	\begin{equation*}
	\Pr\left(\left\Vert \frac{1}{n}\sum_{i=b+1}^n \left(\EE[\bm{x}_i^{(b)}(\bm{x}_i^{(b)})^\top
 |\eta_{i-1}]-\EE[\bm{x}_i^{(b)}(\bm{x}_i^{(b)})^\top]\right)\right\Vert	
 \ge t\right)\le
	bp{\rm exp}\left\{-\frac{t^2/2}{\sigma^2+Rt/3}\right\} 
	\end{equation*}
	converges to 0 by choosing $t=\bigO\left(\sqrt{bp\log n/n}\right)$. Combining all three deviation terms, we conclude that
    \begin{align}\label{projection_matrix_diff}
\Vert \bm{S}_n-\bar{\bm{S}}_n\Vert&\le \Vert \bm{S}_n-\bm{S}_{n,n}\Vert+
\left\Vert\sum_{j=2}^n (\bm{S}_{n,j}-\bm{S}_{n,j-1})\right\Vert+
\Vert\bm{S}_{n,1}-\bar{\bm{S}}_n\Vert \notag\\
&\le C\left(bpn^{-\tau+1/2}
+bpn^{-1/2}+\sqrt{\frac{bp\log n}{n}}\right)\notag \\
&=\bigO_\Pr\left(bpn^{-1/2}\right)
=o_\Pr(1).
\end{align}
Furthermore, with the UPDC condition in Assumption \ref{updc} of the main article and \cref{lemma_w}, we have 
$$\lambda_{\min}(\bm{S}_n)\ge \lambda_{\min}(\bar{\bm{S}}_n)-\lambda_{\min}(\bm{S}_n-\bar{\bm{S}}_n)\ge \kappa_1-Cbpn^{-1/2}>0.$$

On the other hand, we aim to establish the consistency of coefficient matrix estimators $\{\widehat{\bm{\Phi}}_j\}_{j=1}^b$. Recall $\widehat{\bm{\beta}}=\left(\frac{\bm{X}^\top\bm{X}}{n}\right)^{-1}
\left(\frac{\bm{X}^\top\bm{Y}}{n}\right)$, and the deviation between $\widehat{\bm{\beta}}$ and $\bm{\beta}$ becomes
\begin{equation}\label{eq_diff}
\widehat{\bm{\beta}}=\bm{\beta}+
\left(\frac{\bm{X}^\top\bm{X}}{n}\right)^{-1}
\left(\frac{\bm{X}^\top\bm{\epsilon}}{n}\right)+\left(\frac{\bm{X}^\top\bm{X}}{n}\right)^{-1}
\left(\frac{\bm{X}^\top\bm{\xi}}{n}\right)=:\bm{\beta}+\bm{\Phi}^M+\bm{\Phi}^R,
\end{equation}
where $\bm{\xi}=\big(\bm{0},
\bm{\Phi}_{b+1}\bm{x}_{1},
\sum_{j=b+1}^{b+2}\bm{\Phi}_{j}\bm{x}_{b+3-j},
\ldots,\sum_{j=b+1}^{n-1}\bm{\Phi}_j\bm{x}_{n-j}\big)^\top$. From our earlier discussion, one can obtain that $\Vert \bm{S}_n^{-1}\Vert=\bigO_\Pr(1)$. Further denote $\bm{\Phi}_j^R$ as the $j$th block matrix of $\bm{\Phi}^R$ and by elementary calculation, we have $\Vert \bm{\Phi}_j^R\Vert = C\sqrt{p}b^{-\tau+1}(\log b)^{\tau-1}$.

Now, we will explore the error bound of each block of $\bm{\Phi}^M$, denoted by $\bm{\Phi}_j^M$. Similarly since $\Vert \bm{S}_n^{-1}\Vert=\bigO_\Pr(1)$, it is easy to find that $\Vert \bm{\epsilon}^\top\bm{X}/n\Vert=
\bigO_\Pr(\sqrt{bp^2/n})$. For some constant $M>0$, denote the
event $\mathcal{B}_n$ on which $\Vert 
\bm{\epsilon}^\top\bm{X}/n\Vert
\Vert (\bm{X}^\top\bm{X}/n)^{-1}\Vert \le M$, then one can obtain $\Pr(\mathcal{B}_n^c)=o(1)$. Then it yields that for every $j=1,\ldots,b$,
\begin{align}
    &\Pr\left(\Vert \bm{\Phi}_j^M\Vert \ge 2C\sqrt{bp^2\log(n)/n}\right)
    \notag \\
    \le & \Pr\left(\Big\{\Vert \bm{\Phi}_j^M\Vert \ge 2C\sqrt{bp^2\log(n)/n}\Big\}
    \cap \mathcal{B}_n \right)+\Pr(\mathcal{B}_n^c)
    \notag \\
    \le & \Pr\left(\left\Vert \frac{\bm{\epsilon}^\top\bm{X}}{n}(\bm{S}_n^{-1}-\bar{\bm{S}}_n^{-1})
    \bm{E}_j\right\Vert \ge C\sqrt{bp^2\log(n)/n}\right) \label{first_bound}\\
    &\qquad +\Pr\left(\left\Vert 
    \frac{\bm{\epsilon}^\top\bm{X}}{n}
    \bar{\bm{S}}_n^{-1}
    \bm{E}_j\right\Vert \ge C\sqrt{bp^2\log(n)/n}
    \right)+o(1).
    \label{second_bound}
\end{align}
According to the aforementioned discussion and \eqref{projection_matrix_diff}, we conclude that Eq.s \eqref{first_bound} and \eqref{second_bound} converge to zero as $n$ tends to infinity. In conclusion, we complete the proof. $\hfill \square$  

\subsection{Gaussian approximation and multiplier bootstrap}\label{add_theory}
Similar to the process $\widetilde{\bm{B}}_n^{\rm oracle,1}(u,v)$ introduced in \eqref{eq_target}, we define an intermediate process as 
\begin{align}
    \label{det24}
\bm{B}_{n}^{\rm oracle}(u,v):=\bm{A}_{f}(v)
\bm{\Sigma}^{-1}\widetilde{\bm{I}}{\rm diag}(\bm{Z}_n^b)\widetilde{\bm{E}}
\bm{\alpha}_{f+}(u).
\end{align} In this subsection, we will establish a Gaussian approximation theory for the weighted $\mathcal{L}^2$-norm of $\bm{B}_{n}^{\rm oracle,1}$ uniformly over all quantiles and a wide class of weight functions. This theory is based on the uniform Gaussian approximation results over all Euclidean convex sets for sums of stationary and weakly dependent time series of moderately high dimensions. The result extends the corresponding findings for independent and $m$-dependent data established in \cite{Bentkus03}, \cite{Fang15} and \cite{Fang16} among others, which may be of separate interest.

Next, we define 
$$
\bm{U}_n^b:=\frac{1}{\sqrt{n}}\sum_{i=b+1}^n\bm{u}_{i} ~, 
$$
where $\{\bm{u}_{i}\}_{i=b+1}^n$ is a sequence of $bp^2$-dimensional Gaussian random vectors that is independent of $\{\bm{z}_{i}\}_{i=b+1}^n$ and preserves the covariance structure of $\{\bm{z}_{i}\}_{i=b+1}^n$. With this notation we consider 
$$
\bm{B}_n^u(u,v)=
\bm{A}_{f}(v)
\bm{\Sigma}^{-1}
\widetilde{\bm{I}}{\rm diag}(\bm{U}_n^b)\widetilde{\bm{E}}
\bm{\alpha}_{f+}(u)
$$
as the Gaussian analogue of \eqref{det24}.
	Now, we state the Gaussian approximation result for the weighted $\mathcal{L}^2$ norm of the target statistic $\bm{B}_n^{\rm oracle}$ in the following theorem.
	
	\begin{theorem}\label{thm_gaussian}
		Let  Assumptions \ref{ass_conti_u}, \ref{weak_depen_components}--\ref{ass_design_matrix} be satisfied and suppose the smallest eigenvalue of the covariance matrix $\bm{\Xi}_n$ is bounded below by some constant $\kappa_4>0$, there exists a constant $C>0$ such that
		\begin{equation}
		\label{pGaus1}
        \begin{split}
        & \sup \limits_{
	x\in\mathbb{R},
    \bm{w}_n\in\mathcal{W}}
    \left|
	\Pr\left(|
 \bm{B}_n^{\rm oracle}|_
 {\mathcal{L}^2,\bm{w}_n}\le x\right)-\Pr\left(|\bm{B}_n^u|_
 {\mathcal{L}^2,\bm{w}_n}\le x\right)\right|  \\
 &~~~~~~~~~~~~~~~~~~~~~~~~~~~~~~~~
 \le C\min\left\{
        (bp^2)^{7/4}n^{-\frac{1}{2}+\frac{9}{2q}+\frac{2}{\tau-1}}, (bp^2)^{\frac{5}{6}} n^{\frac{2}{3 r}-\frac{1}{3}}(\log n)^{\frac{1}{3}}\right\}, 
        \end{split}
\end{equation}
        where 
        $$
        \frac{1}{r}=
    \max\left\{\frac{1}{q},\frac{1}{q\sqrt{\tau+1} }+\left(\frac{1}{2}-\frac{1}{q\sqrt{\tau+1} }\right) \max \Big\{\frac{2}{q\sqrt{\tau+1}}, \frac{1}{\tau}\big(\frac{1}{2}-\frac{1}{q}\big)\Big\}\right\}.
    $$
	\end{theorem}
Theorem \ref{thm_gaussian} implies that the distribution of $\left|\bm{B}_n^{\rm oracle}\right|_{\mathcal{L}^2, \bm{w}_n}$ can be well approximated by that of $\left|\bm{B}_n^u\right|_{\mathcal{L}^2, \bm{w}_n}$ uniformly over all quantiles and weight functions in $\mathcal{W}$ when $q$ and $\tau$ are large enough. In particular, it shows that for  sufficiently large $q$ and $\tau$ and for  sufficiently small $b, p$ are the first term on the right-hand side of \eqref{pGaus1} converges at the rate $\bigO(n^{-1/2+\omega_1})$ for arbitrarily small $\omega_1>0$. On the other hand, when $q$ and $\tau$ are sufficiently large, $b\asymp \log(n)$ and the last term on the right-hand side of \eqref{pGaus1} becomes $p^{5/3}n^{-1/3+\omega_2}$ for arbitrary small $\omega_2>0$. Therefore, the bounds in \eqref{pGaus1} vanish when the dimension of the random vector $\bm{x}_i$ grows as large as  $p \asymp n^{1/5-\omega}$ for arbitrarily small $\omega>0$.
\medskip

\noindent
\textit{Proof}:
According to \cref{comparison_new}, we will provide two upper bounds for the Gaussian approximation. First, we start to follow \cref{comparison} to derive first approximation bound. For the second approximation bound in \eqref{pGaus1}, we leverage the Wasserstein Gaussian approximation result finish the proof.
 
    For the first approximation bound,  recall $\bm{z}_{i}={\rm vec}
    \big((\bm{x}_{i}^{(b)})\bm{\epsilon}_i^\top\big)$. Since $\{\bm{x}_{i}^{(b)}\}$ and $\{\bm{\epsilon}_i\}$ are both stationary processes, we can rewrite $\bm{z}_{i}$ into a physical representation of the stationary multivariate time series, i.e.,
	$$\bm{z}_{i}=\bm{L}(\mathcal{F}_i),$$ where $\bm{L}=(L_1,...,L_{bp^2})^\top$ is a measurable vector function. Here, we denote the dependence measure of the component $\{z_{i,j}\}_{j=1}^{bp^2}$ for $i=b+1,\ldots,n$ as $$\delta_z(l,q):=\max_{1\le k\le bp^2}\Vert L_k(\mathcal{F}_i)-L_k(\mathcal{F}_{i,l})\Vert_q,$$ 
	where $\mathcal{F}_{i,l}$ is defined in \cref{def1}. Under \cref{ass_dep} and \cref{eps_depend}, we will obtain that
	\begin{align*}
	\delta_z(l,q/2)\le&
    \max_{1\le j\le b}
    \max_{1\le h\le p}\Vert x_{i-j,h}-x_{i-j,h}^\ast\Vert_{q}
 \max_{1\le m\le p}\Vert \epsilon_{i,m}\Vert_{q}+\max_{1\le j\le b}
 \max_{1\le h\le p}\Vert x_{i-j,h}^\ast\Vert_{q} \max_{1\le m\le p}\Vert \epsilon_{i,m}-\epsilon_{i,m}^\ast\Vert_{q}\\
	\le& C\sqrt{p}
    l^{-\tau+2}(\log l+1)^{\tau-1},
	\end{align*}
 where $x_{i-j,h}^\ast$ and $\epsilon_{i,m}^\ast$ are i.i.d. copies of $x_{i-j,h}$ and $\epsilon_{i,m}$, respectively. On the other hand, denote $$T_x^{w_n}=\left\{\bm{S}\in \mathbb{R}^{bp^2}:
\frac{1}{b}\sum_{j=1}^{b} \left|\bm{E}_j^\top\bm{A}_f(v)\bm{\Sigma}^{-1}
\widetilde{\bm{I}}{\rm diag} (\bm{S})\widetilde{\bm{E}}\bm{\alpha}_{f+}(u)\big/w_{n,j}(u,v)\right|_{\mathcal{L}^2}
	\le x\right\}$$ where $\bm{E}_j\in\mathbb{R}^b$ contains 1 at the $j$th location and 0 at others. Furthermore, let $\mathcal{T}_n=\{T_x^{w_n}: x\in\mathbb{R}, w_{n,j}(u,v)\in \mathcal{W} \text{~for~all~}j=1,\ldots,b\}$. Since $T_x^{w_n}$ is a convex set and $\mathcal{T}_n$ is a collection of convex sets, then we can rewrite the Kolmogorov distance between $\bm{Z}_n^b$ and $\bm{U}_n^b$ on $\mathcal{T}_n$ as
	$$\mathcal{D}(\bm{Z}_n^b,\bm{U}_n^b):=
    \sup_{\bm{w}_n\in\mathcal{W}}
	\sup_{x\in\mathbb{R}}\left|
	\Pr(\bm{Z}_n^b\in T_x^{w_n})-\Pr(\bm{U}_n^b\in T_x^{w_n})\right|=
\sup_{T\in\mathcal{T}_n}\left|
\Pr(\bm{Z}_n^b\in T)-\Pr(\bm{U}_n^b\in T)\right|,$$
	where $T$ is a convex set in $\mathcal{T}_n$.

 To deduce the error bound of Gaussian approximation, we need to make use of the truncation, $m$-dependent approximation techniques and finally apply Lemmas \ref{derivative}--\ref{comparison_new} in \cref{app_auxilary}. Now define the truncated version of $\bm{z}_{i}$ as
	\begin{equation*}
	\bar{\bm{z}}_{i}=\begin{cases}
	\bm{z}_{i},~&|\bm{z}_{i}|\le (bp^2)^{\frac{1}{2}}n^{\frac{3}{2q}},\\
	\bm{0}_{bp^2},~&\text{otherwise}.
	\end{cases}
	\end{equation*}
	Given a large constant $M=M_n$, let the $M$-dependent approximation sequence of $\bar{\bm{z}}_{i}$ be
$$\bar{\bm{z}}_{i}^M=\EE(\bar{\bm{z}}_{i}|\eta_{i-M-1},\cdots,\eta_{i-1}),
	~i=b+1,\cdots,n.$$
	Consequently, we can define the partial sum of the truncated series as $\bar{\bm{Z}}_n=\sum_{i=b+1}^n\bar{\bm{z}}_{i}/\sqrt{n}$ and the $M$-dependent analogue as $\bar{\bm{Z}}_n^M=\sum_{i=b+1}^n\bar{\bm{z}}_{i}^M/\sqrt{n}$. Further denote $\bar{\bm{Z}}_n^\ast=\bar{\bm{Z}}_n-\EE\bar{\bm{Z}}_n$, $\widetilde{\bm{Z}}_n^M=\bar{\bm{Z}}_n^M-\EE\bar{\bm{Z}}_n$ and let $\widetilde{\bm{U}}_n^M$ be a Gaussian random vector preserving the covariance structure of $\widetilde{\bm{Z}}_n^M$. With Lemmas \ref{derivative} and \ref{comparison}, we have
	\begin{align}
	\mathcal{D}(\bm{Z}_n^b,\bm{U}_n^b)\le&~4(bp^2)^{\frac{1}{4}}\sigma_1
	+\sup_{T\in\mathcal{T}_n}
 \left|\EE\left[h_{T,\sigma_1}(\bm{Z}_n^b)- h_{T,\sigma_1}(\bm{U}_n^b)\right]\right| \notag\\
	\le&~4(bp^2)^{\frac{1}{4}}\sigma_1+\sup_{T\in\mathcal{T}_n}
 \left|\EE\left[
	h_{T,\sigma_1}(\bm{Z}_n^b)-h_{T,\sigma_1}(\bar{\bm{Z}}_n^\ast)
	\right]\right|+\sup_{T\in\mathcal{T}_n}\left|
	\EE\left[h_{T,\sigma_1}(\bar{\bm{Z}}_n^\ast)- h_{T,\sigma_1}(\widetilde{\bm{Z}}_n^M)\right]\right|\notag \\
	&{}+\sup_{T\in\mathcal{T}_n}\left|
	\EE\left[h_{T,\sigma_1}(\tilde{\bm{Z}}_n^M)- h_{T,\sigma_1}(\widetilde{\bm{U}}_n^M)\right]\right|
	+\sup_{T\in\mathcal{T}_n}\left|
	\EE\left[h_{T,\sigma_1}(\widetilde{\bm{U}}_n^M)- h_{T,\sigma_1}(\bm{U}_n^b)\right]\right| \notag\\
	\le&~4(bp^2)^{\frac{1}{4}}\sigma_1+\frac{C}{\sigma_1}\EE|\bm{Z}_n^b-
	\bar{\bm{Z}}_n^\ast|+
    \frac{C}{\sigma_1}\EE|\bar{\bm{Z}}_n^\ast-
	\widetilde{\bm{Z}}_n^M|\notag \\
	&{}+\sup_{T\in\mathcal{T}_n}\left|
	\EE\left[h_{T,\sigma_1}(\widetilde{\bm{Z}}_n^M)- h_{T,\sigma_1}(\widetilde{\bm{U}}_n^M)\right]\right|+
	\frac{C}{\sigma_1}\EE|\widetilde{\bm{U}}_n^M-\bm{U}_n^b|\notag\\
	:=&~4(bp^2)^{\frac{1}{4}}\sigma_1+
	\frac{C\sigma_2}{\sigma_1}+
    \frac{C\sigma_3}
	{\sigma_1}+
    \frac{C\sigma_4}
	{\sigma_1}+
 \sup_{A\in\mathcal{T}_n}
 \left|\EE\left[h_{T,\sigma_1}
	(\widetilde{\bm{Z}}_n^M)-h_{T,\sigma_1}(\widetilde{\bm{U}}_n^M)\right]\right|, \label{combine}
	\end{align}
	where $\sigma_2:=\EE|\bm{Z}_n^b-\bar{\bm{Z}}_n^\ast|$, $\sigma_3:=\EE|\bar{\bm{Z}}_n^\ast-
	\widetilde{\bm{Z}}_n^M|$ and $\sigma_4:=\EE|\widetilde{\bm{U}}_n^M-\bm{U}_n^b|$.\\
	(1) Truncation approximation.
	We shall first control the truncation error $\sigma_2$. Note that $\EE\bm{z}_i=0$, then we have
 \begin{align*}
	|\bm{Z}_n^b-\bar{\bm{Z}}_n^\ast|&=
	\frac{1}{\sqrt{n}}\left|\sum_{i=b+1}^n(\bm{z}_{i}-\EE\bm{z}_{i}
	-\bar{\bm{z}}_{i}+\EE\bar{\bm{z}}_{i})\right|\\
	&\le \frac{1}{\sqrt{n}}\left|\sum_{i=b+1}^n(\bm{z}_{i}-
	\bar{\bm{z}}_{i})\right|+\frac{1}{\sqrt{n}}
	\left|\sum_{i=b+1}^n\EE(\bm{z}_{i}-\bar{\bm{z}}_{i})\right|\\
	:&={\rm I}+{\rm II}.
	\end{align*}
	For ${\rm I}$, notice that for any $i=1,...,n$,
	\begin{align*}
	&\Pr\left(|\bm{z}_{i}|>(bp^2)^{\frac{1}{2}}n^{\frac{3}{2q}}\right)=
	\EE\left[\bm{1}\{|\bm{z}_{i}|>(bp^2)^{\frac{1}{2}}n^{\frac{3}{2q}}\}\right]\\
	\le &\EE\left[\left(\frac{|\bm{z}_{i}|}
	{(bp^2)^{\frac{1}{2}}n^{\frac{3}{2q}}}\right)^q\right]=
	(bp^2)^{-\frac{q}{2}}n^{-\frac{3}{2}}\EE|\bm{z}_{i}|^q \\
	=&(bp^2)^{-\frac{q}{2}}
	n^{-\frac{3}{2}}\EE\left|\sum_{j=1}^{bp^2} z_{i,j}^2\right|^{q/2}
	\le (bp^2)^{-\frac{q}{2}}
	n^{-\frac{3}{2}}(bp^2)^{q/2-1}\sum_{j=1}^{bp^2}\EE|z_{i,j}|^q
	\le C_qn^{-\frac{3}{2}},
	\end{align*}
	where the second to last inequality follows from the inequality $\EE|X_1+\cdots+X_{bp^2}|^{q/2}\le (bp^2)^{q/2-1}\sum_{j=1}^{bp^2}\EE|X_j|^{q/2}$ for random variables $\{X_j\}_{j=1}^{bp^2}$ and \cref{ass_moment} in the main paper. Hence, we have $\Pr\left(|\bm{Z}_n^b-\bar{\bm{Z}}_n|=0\right)=1-\bigO(n^{-1/2})$. This implies there exists an order, say $n^{-2}$ such that $\Pr\left(\left|\frac{1}{\sqrt{n}}\sum_{i=b+1}^n
	(\bm{z}_{i}-\bar{\bm{z}}_{i})\right|>n^{-2}\right)\to 0$. As a result, ${\rm I}=o_\Pr(n^{-2})$.
	
	For ${\rm II}$, since for any $i=1,...,n$,
	\begin{align*}
	\EE(\bm{z}_{i}-\bar{\bm{z}}_{i})
	\le& \EE\left[|\bm{z}_{i}|
	\bm{1}\{|\bm{z}_{i}|>(bp^2)^{\frac{1}{2}}
	n^{\frac{3}{2q}}\}\right]\\
	\le& \EE\left[|\bm{z}_{i}|\left(\frac{|\bm{z}_{i}|}
	{(bp^2)^{\frac{1}{2}}n^{\frac{3}{2q}}}\right)^{q-1}\right]\\
	=&(bp^2)^{-\frac{q-1}{2}}n^{-\frac{3}{2}+\frac{3}{2q}}\EE|\bm{z}_{i}|^q\\ 
	\le&C_q (bp^2)^{\frac{1}{2}}n^{-\frac{3}{2}+\frac{3}{2q}},
	\end{align*}
	where the second inequality uses the fact that for the nonnegative random variable $y$ and some number $a>0$, the inequality $y\bm{1}\{y\ge a\}\le y\left(\frac{y}{a}\right)^p$ for any $p>0$ holds true. Consequently, ${\rm II}=\frac{1}{\sqrt{n}}\left|\sum_{i=b+1}^n
    \EE(\bm{z}_{i}-\bar{\bm{z}}_{i})\right|
    =\bigO((bp^2)^{\frac{1}{2}}n^{-1+\frac{3}{2q}})$. Now, by choosing $\beta=(bp^2)^{\frac{1}{2}}n^{-\frac{1}{2}+\frac{3}{2q}}$, then $\sigma_2=\bigO((bp^2)^{\frac{1}{2}}n^{-1+\frac{3}{2q}})$.
	\bigskip
	
	\noindent
	(2) $M$-dependence approximation.
	Next, we will deduce the 
    approximation rate between our original process and its $M$-dependent sequence, i.e., control $\sigma_3$ in \eqref{combine}. Recall the physical dependence measure $\delta_z(l,q/2)$ of ${z}_{i,j}$ and denote $\Theta_{M,q/2}=\sum_{l=M}^\infty\delta_z(l,q/2)$. Let
	$$\bar{\bm{Z}}_n-\bar{\bm{Z}}_n^M=\frac{1}{\sqrt{n}}\sum_{i=b+1}^n
	(\bar{\bm{z}}_{i}-\bar{\bm{z}}_{i}^M)=:\frac{1}{\sqrt{n}}\sum_{i=b+1}^n
	\bar{\bm{z}}_i^\Delta.$$ It is readily seen that $\{\bar{\bm{z}}_i^\Delta,i=b+1,...,n\}$ is a sequence of martingale differences, then we have
	\begin{align*}
	&\left\Vert\bar{\bm{Z}}_n^\ast-\tilde{\bm{Z}}_n^M\right\Vert_{q/2}^2=
	\left\Vert\bar{\bm{Z}}_n-\bar{\bm{Z}}_n^M\right\Vert_{q/2}^2\\ =&
	\left\Vert \frac{1}{\sqrt{n}}\sum_{i=b+1}^n \bar{\bm{z}}_i^\Delta \right\Vert_{q/2}^2
	=\left\{\EE\left[\sum_{j=1}^{bp^2}\left(\frac{1}{\sqrt{n}}\sum_{i=b+1}^n
	\bar{z}_{i,j}^\Delta\right)^2
 \right]^{q/4}\right\}^{4/q}\\
	\le&\left\{(bp^2)^{q/4-1}\sum_{j=1}^{bp^2}\EE\left|\frac{1}{\sqrt{n}}\sum_{i=b+1}^n
	\bar{z}_{i,j}^\Delta\right|^{q/2}\right\}^{4/q}\\
	\le& bp^2\left\Vert \max_{1\le j\le bp^2}\Big | \frac{1}{\sqrt{n}}\sum_{i=b+1}^n\bar{z}_{i,j}^\Delta\Big |\right\Vert_{q/2}^2 \le Cbp^2\Theta_{M,q/2}^2\le Cbp^3M^{-2\tau+6}(\log M)^{2\tau-2},
	\end{align*}
	where $\bar{z}_{i,j}^\Delta$ is the entrywise of vector $\bar{\bm{z}}_i^\Delta$, the first inequality is due to the fact that $\EE|X_1+\cdots+X_{bp^2}|^{q/4}\le (bp^2)^{q/4-1}\sum_{j=1}^{bp^2}\EE|X_j|^{q/4}, q>4$ with $X_j$ being random variables. The third inequality above is followed by Lemma A.1 of \cite{LiuLin09} using Burkholder's inequality. As a result, $\left\Vert\bar{\bm{Z}}_n^\ast-\widetilde{\bm{Z}}_n^M\right\Vert_q
	\le Cbp^{3/2}M^{-\tau+3}(\log M)^{\tau-1}$.
	
	Therefore, we can choose $M$ appropriately such that $M^{-\tau+3}(\log M)^{\tau-1}$ converges to zero fast. For example $M=\bigO\big(n^{\frac{1}{\tau-3}}(\log n)^{\frac{\tau-1}{\tau-3}}\big)$ for $\tau>4$, then we have $\bar{\bm{Z}}_n^\ast-\widetilde{\bm{Z}}_n^M=\bigO_\Pr(bp^{3/2}
    n^{-1})$.
	\medskip
	
	\noindent
	(3) Using \cref{distance} deduce the final result.
	Finally, we will employ \cref{distance} to deal with the last term of \cref{combine}. Let $\widetilde{\bm{Z}}_n^M=\sum_{i=b+1}^n
 \bar{\bm{z}}_{i}^\dagger/\sqrt{n}$, then by the truncation and $m$ dependence approximation techniques, we have $|\bar{\bm{z}}_{i}^\dagger/\sqrt{n}|\le \beta$ and
	$\EE\bar{\bm{z}}_{i}^\dagger/\sqrt{n}=0$. Recall $\bm{\Xi}_n=
	\EE\left(\sum_{i=b+1}^n\bm{z}_{i}\right)
	\left(\sum_{i=b+1}^n\bm{z}_{i}^\top\right)/n$, denote 
 $$\bm{\Xi}_M:={\rm Cov}(\widetilde{\bm{Z}}_n^M)=\EE\left[\sum_{i=b+1}^n
 (\bar{\bm{z}}_{i}^M-
	\EE\bar{\bm{z}}_{i})\right]\left[
 \sum_{i=b+1}^n(\bar{\bm{z}}_{i}^M
	-\EE\bar{\bm{z}}_{i})^\top\right]/n.$$ Since $\EE \bm{z}_i=0$, we find that the difference of covariance matrix between $\widetilde{\bm{Z}}_n^M$ and $\bm{Z}_n$ based on the Frobenius norm turns out to be
	\begin{align*}
	&\left|\bm{\Xi}_n-\bm{\Xi}_M\right|_F\\
	\le &\frac{1}{n}\left\{\left|\EE\left[\sum_{i=b+1}^n\left(\bm{z}_{i}-
	\bar{\bm{z}}_{i}^M\right)\right]
 \left(\sum_{i=b+1}^n\bm{z}_{i}^\top\right)
	\right|_F+\left|\EE\left(\sum_{i=b+1}^n
 \bar{\bm{z}}_{i}^M\right)
	\left[\sum_{i=b+1}^n\left(\bm{z}_{i}-\bar{\bm{z}}_{i}^M\right)^\top\right]
	\right|_F\right.\\
	&+\left.\left|\sum_{i=b+1}^n\EE
	(\bm{z}_{i}-\bar{\bm{z}}_{i})\sum_{i=b+1}^n\EE
	(\bm{z}_{i}-\bar{\bm{z}}_{i})^\top\right|_F\right\}\\
	\le &C\left(b^{1/2}pn^{-2}+
    b^{3/2}p^{5/2}M^{-\tau+3}(\log M)^{\tau-1}+(bp^2)n^{-2+3/q}\right)=\bigO(b^{3/2}p^{5/2}M^{-\tau+3}(\log M)^{\tau-1}),
	\end{align*}
	where the last inequality follows by the error bounds of the truncation approximation, $M$-dependence approximation, and the Cauchy Schwarz inequality. Further denote $\widetilde{U}_{n,j}^M$ and $U_{n,j}$ are components of vector $\widetilde{\bm{U}}_n^M$ and $\bm{U}_n^b$, respectively. Then, 
	\begin{align*}
	\sigma_4&=\EE\sqrt{\sum_{j=1}^{bp^2}(\widetilde{U}_{n,j}^M-U_{n,j})^2}\le \sqrt{\EE\sum_{j=1}^{bp^2}(\widetilde{U}_{n,j}^M-U_{n,j})^2}\\
	&=\sqrt{{\rm Tr}\left[(\bm{\Xi}_M)^{1/2}-
		(\bm{\Xi}_n)^{1/2}\right]^2}=\left|(\bm{\Xi}_M)^{1/2}-(\bm{\Xi}_n)^{1/2}
	\right|_F\\
	&=\bigO(b^{3/2}p^{5/2}M^{-\tau+3}(\log M)^{\tau-1}).
	\end{align*}  Since we assume that the smallest eigenvalue of $\bm{\Xi}_n$ is bounded below by some constant $\kappa_4>0$, we can deduce that
	\begin{align*}
	\lambda_{\min}(\bm{\Xi}_M)&\ge
	\lambda_{\min}(\bm{\Xi}_M-\bm{\Xi}_n)+
	\lambda_{\min}(\bm{\Xi}_n)\\
	&= -\lambda_{\max}(\bm{\Xi}_n-\bm{\Xi}_M)+
	\lambda_{\min}(\bm{\Xi}_n)\\
	&\ge \kappa_4-Cb^{3/2}p^{5/2}M^{-\tau+3}(\log M)^{\tau-1}>0.
	\end{align*}
	Hence, the small eigenvalue of $\bm{\Xi}_M$ is also bounded below by some positive constant, then $\left|(\bm{\Xi}_M)^{-1}\right|\le C$.
	Further note that $n_1=M,~n_2=2M,~n_3=3M$ and together with the Eq. (4.24) in \cite{Fang16}, we have
	\begin{equation}\label{key}
	\mathcal{D}(\widetilde{\bm{Z}}_n^M,\widetilde{\bm{U}}_n^M)\le 4(bp^2)^{\frac{1}{4}}\sigma_1
	+2Cn\beta^3M^2\frac{1}{\sigma_1}[(bp^2)^{\frac{1}{4}}(\sigma_1+3M\beta)+
    \mathcal{D}(\widetilde{\bm{Z}}_n^M,
    \widetilde{\bm{U}}_n^M)].
	\end{equation}
	By substituting $\beta=(bp^2)^{\frac{1}{2}}n^{-\frac{1}{2}+\frac{3}{q}}$ and optimizing $\sigma_1$, we can choose $\sigma_1=\bigO((bp^2)^{\frac{3}{2}}
	n^{-\frac{1}{2}+\frac{9}{2q}}M^2)$. Since $\sigma_3=o(\sigma_4)$, then it suffices to calculate the second and fourth terms in \cref{combine}, which are $\bigO((bp^2)^{-1}n^{-\frac{1}{2}-\frac{3}{q}}M^{-2})$ and $\bigO(p^{-\frac{1}{2}}n^{\frac{1}{2}-\frac{3}{2q}}M^{-\tau+1}(\log M)^{\tau-1})$, respectively. By choosing $M=\bigO(n^{\frac{1}{\tau-1}})$, we have $$\mathcal{D}(\bm{Z}_n^b,\bm{U}_n^b)\le C
	(bp^2)^{\frac{7}{4}}n^{-\frac{1}{2}+\frac{9}{2q}+\frac{2}{\tau-1}}.$$  

    Now, let us derive the second approximation error bound following \cref{comparison_new} and the technical schemes in the proof of \citep[Theorem 3.2]{liu2025wasserstein}. More specifically, denote $\widetilde{\bm{U}}_n$ as an intermediate Gaussian vector with covariance different from ${\rm Cov}(\bm{Z}_n^b)$. Then according to Eq.s (36),(37) in \cite{liu2025wasserstein} and the Chebyshev’s inequality, the Kolmogorov distance $\mathcal{D}(\bm{Z}_n^b,\bm{U}_n^b)$ becomes
    \begin{align*}
        \mathcal{D}(\bm{Z}_n^b,\bm{U}_n^b)& \le 4(bp^2)^{1/4}\sigma+
        \Pr(|\bm{Z}_n^b-\bm{U}_n^b|> \epsilon)\\
        &\le 4(bp^2)^{1/4}\sigma + \Pr(|\bm{Z}_n^b-\widetilde{\bm{U}}_n|> \epsilon/2)+
        \Pr(|\widetilde{
        \bm{U}}_n-\bm{U}_n^b|> \epsilon/2)\\
        &\le 4(bp^2)^{1/4}\sigma + 4C\frac{(bp^2)^2 n^{\frac{2}{r}-1} \log (n)}{\sigma^2 }+4C\frac{(bp^2)^2 n^{\frac{2}{r}-1}}{\sigma^2}\\
        &= 4(bp^2)^{1/4}\sigma + \bigO\left(\frac{(bp^2)^2 n^{\frac{2}{r}-1} \log (n)}{\sigma^2}\right),
    \end{align*}
    where $\frac{1}{r}=
    \max\left\{\frac{1}{q},\frac{1}{q\sqrt{\tau+1} }+\left(\frac{1}{2}-\frac{1}{q\sqrt{\tau+1} }\right) \max \Big\{\frac{2}{q\sqrt{\tau+1}}, \frac{1}{\tau}\big(\frac{1}{2}-\frac{1}{q}\big)\Big\}\right\}$. By choosing $\sigma=\bigO((bp^2)^{\frac{7}{12}} n^{\frac{2}{3 r}-\frac{1}{3}}(\log n)^{1 / 3})$, we have $$
    \mathcal{D}(\bm{Z}_n^b,\bm{U}_n^b)\le 
    C(bp^2)^{\frac{5}{6}} n^{\frac{2}{3 r}-\frac{1}{3}}(\log n)^{\frac{1}{3}}.$$
    In summary, we finish our proof of \cref{thm_gaussian}.
$\hfill \square$

Next theorem provides the rate of the theoretical multiplier bootstrap approximation.
	\begin{theorem}\label{thm_boots} Let 
    Assumptions \ref{ass_conti_u}, \ref{weak_depen_components}--\ref{ass_design_matrix} be satisfied, assume that 
	the smallest eigenvalue of the matrix $\bm{\Xi}_n$ is bounded below by some constant $\kappa_4>0$ and $m=\bigO\big((n/p)^{1/3}\big)$.  On the event $\mathcal{A}_n$ that defined in \eqref{det26}, we have for the process $\widetilde{\bm{B}}_n^{\rm oracle,2}$ defined in \eqref{det25}
		\begin{equation}\label{eq_compare}
        \begin{split}
    &		\sup \limits_{x\in\mathbb{R}, \bm{w}_n\in \mathcal{W}}\Big|
		\Pr\Big(
  \big|\widetilde{\bm{B}}_n^{\rm oracle,2}\big|
  _{\mathcal{L}^2,\bm{w}_n}
  \le x\mid \Upsilon_{b+1}^{n}
  \Big)-\Pr\Big(\big|
		\bm{B}_n^u \big|
  _{\mathcal{L}^2, \bm{w}_n}\le x\Big)\Big|
   \\
  & ~~~~~~~~~~~~~~~~~~~~~~~~~~~~~~~~~~~~~~~~~~~~~~~~~~~~~~ 
            \le C\left(b^{7/8}p^{5/4}n^{-1/4}\log^{1/2}(n)+bp^{17/6}n^{-1/3}h_n\right).
                    \end{split}
		\end{equation}
	\end{theorem}
    
\cref{thm_boots} states that the distribution of the statistic $|\bm{B}_n^u|_{\mathcal{L}^2,\bm{w}_n}$, which is built on Gaussian random variables, can be well approximated by the conditional distribution of $|\widetilde{\bm{B}}_n^{\rm oracle,2}|_{\mathcal{L}^2,\bm{w}_n}$ given $\Upsilon_{b+1}^{n}$. In particular, the first term on the right-hand side of \eqref{eq_compare} represents the multiplier bootstrap approximation error between the Gaussian vector $\bm{U}_n^b$ and the conditional Gaussian vector $\bm{Z}_n^{\ast,b}\mid \Upsilon_{b+1}^n$, while the second term relates to the covariance-matrix comparison error between $\widetilde{\bm{\Xi}}_n$ and $\bm{\Xi}_n$. Combined with \cref{thm_gaussian}, the above theorem also implies that under certain conditions, the conditional distribution of $|\widetilde{\bm{B}}_n^{\rm oracle,2}|_{\mathcal{L}^2, \bm{w}_n}$ given $\Upsilon_{b+1}^{n}$  approximates the distribution of
  $|\bm{B}_n^{\rm oracle}|_{\mathcal{L}^2, \bm{w}_n}$ uniformly across all quantiles and a wide class of smooth weight functions.
\medskip

\noindent
  \textit{Proof}: 
Rewrite $\widetilde{\bm{B}}_n^{\rm oracle,2}(u,v)=\bm{A}_{f}(v)\widetilde{\bm{\Sigma}}
^{-1}\widetilde{\bm{I}}
{\rm diag}(\bm{Z}_n^{\ast,b})
\widetilde{\bm{E}}\bm{\alpha}_{f+}(u)$ where $\widetilde{\bm{\Sigma}}
^{-1}=(\bm{X}^\top \bm{X}/n)^{-1}$ and define an intermediate process as $\bm{B}_n(u,v)=\bm{A}_{f}(v)\bm{\Sigma}
^{-1}\widetilde{\bm{I}}
{\rm diag}(\bm{Z}_n^{\ast,b})
\widetilde{\bm{E}}\bm{\alpha}_{f+}(u)$. By the triangle inequality, it follows that
\begin{align}
    \mathcal{D}'(\bm{Z}_n^{\ast,b}, \bm{U}_n^b)
    :=&\sup_{x\in R, \bm{w}_n\in \mathcal{W}}\left|\Pr\left(
   \left|\widetilde{\bm{B}}_n^{\rm oracle,2}(u,v)\right|
_{\mathcal{L}^2,\bm{w}_n}\le x\mid 
\Upsilon_{b+1}^n\right)-
\Pr\left(\left|\bm{B}_n^u(u,v)\right|
_{\mathcal{L}^2,\bm{w}_n}\le x\right)\right| \notag\\
\le & \sup_{x\in R, \bm{w}_n\in \mathcal{W}}\left|
\Pr\left(
   \left|\widetilde{\bm{B}}_n^{\rm oracle,2}(u,v)\right|
_{\mathcal{L}^2,\bm{w}_n}\le x\mid 
\Upsilon_{b+1}^n\right)-\Pr\left(
   \left|\bm{B}_n(u,v)\right|
_{\mathcal{L}^2,\bm{w}_n}\le x \mid \Upsilon_{b+1}^n\right)
\right| \notag\\
&\quad + \sup_{x\in \mathbb{R}, \bm{w}_n\in \mathcal{W}}\left|
\Pr\left(
   \left|\bm{B}_n(u,v)\right|
_{\mathcal{L}^2,\bm{w}_n}\le x\mid 
\Upsilon_{b+1}^n\right)- \Pr\left(
   \left|\bm{B}_n^u(u,v)\right|
_{\mathcal{L}^2,\bm{w}_n}\le x\right)\right| \notag\\
=&:I+II \label{boots_error}
\end{align}
For the term $I$, first we define two random vectors $\widetilde{\bm{V}}_n^m:={\rm vec}(\widetilde{\bm{\Sigma}}^{-1}\widetilde{\bm{I}}{\rm diag}(\bm{Z}_n^{\ast,b})\widetilde{\bm{E}})\in \mathbb{R}^{bp^2}$ and $\bm{V}_n^m:={\rm vec}(\bm{\Sigma}^{-1}\widetilde{\bm{I}}{\rm diag}(\bm{Z}_n^{\ast,b})\widetilde{\bm{E}})\in \mathbb{R}^{bp^2}$. Consider the event 
        $$\widetilde{T}_x^{w_n}=\left\{\bm{S}\in \mathbb{R}^{bp^2}:
\frac{1}{b}\sum_{j=1}^{b} \left|\bm{E}_j^\top\bm{A}_f(v)
\widetilde{\bm{I}}{\rm diag} (\bm{S})\widetilde{\bm{E}}\bm{\alpha}_{f+}(u)\big/w_{n,j}(u,v)\right|_{\mathcal{L}^2}
	\le x\right\}.$$ Moreover, let $\widetilde{\mathcal{T}}_n=\{\widetilde{T}_x^{w_n}: x\in\mathbb{R}, w_{n,j}(u,v)\in \mathcal{W} \text{~for~all~}j=1,\ldots,b\}$. Then followed by Lemmas \ref{derivative} and \ref{comparison},
   we have
   {\footnotesize
        \begin{align*}
        I & =\sup_{x\in\mathbb{R},\bm{w}_n\in\mathcal{W}}\left|
    \Pr\left(\left|\bm{A}_f(v)
    \widetilde{\bm{I}}{\rm diag}(\widetilde{\bm{V}}_n^m)\widetilde{\bm{E}}\bm{\alpha}_{f+}(u)\right|_{\mathcal{L}^2,
    \bm{w}_n}\le x \mid \Upsilon_{b+1}^n\right)-\Pr\left(\left|\bm{A}_f(v)
    \widetilde{\bm{I}}{\rm diag}(\bm{V}_n^m)\widetilde{\bm{E}}\bm{\alpha}_{f+}(u)\right|_{\mathcal{L}^2,\bm{w}_n}\le x \mid \Upsilon_{b+1}^n \right)\right|\\
    &\le 4(bp^2)^{1/4}\sigma + \sup_{T\in \widetilde{\mathcal{T}}_n}
    \left|\EE\left[h_{T,\sigma}(\widetilde{\bm{V}}_n^m)-h_{T,\sigma}(\bm{V}_n^m)\mid \Upsilon_{b+1}^n\right]\right|\\
    &\le 4(bp^2)^{1/4}\sigma
    +\frac{C}{\sigma}\EE\left[\left|\widetilde{\bm{V}}_n^m-\bm{V}_n^m\right|\mid \Upsilon_{b+1}^n\right]\\
    &\le 4(bp^2)^{1/4}\sigma
    +\frac{C}{\sigma}
    \sqrt{\EE\left[(\widetilde{\bm{V}}_n^m-\bm{V}_n^m)^\top(\widetilde{\bm{V}}_n^m-\bm{V}_n^m)\mid \Upsilon_{b+1}^n\right]}\\
    &= 4(bp^2)^{1/4}\sigma+ \frac{C}{\sigma}\sqrt{{\rm Tr}\left([\bm{I}_p\otimes (\widetilde{\bm{\Sigma}}^{-1}-\bm{\Sigma}^{-1})]^2
    \frac{1}{n-m-b+1}\sum_{i=b+1}^{n-m+1}\left(\frac{1}{\sqrt{m}}\sum_{j=i}^{i+m-1}\bm{z}_j\right)\left(\frac{1}{\sqrt{m}}\sum_{j=i}^{i+m-1}\bm{z}_j^\top\right)\right)}\\
    &\le 4(bp^2)^{1/4}\sigma+\frac{C}{\sigma} \left\Vert \bm{I}_p\otimes (\widetilde{\bm{\Sigma}}^{-1}-\bm{\Sigma}^{-1})\right\Vert
    \sqrt{\frac{1}{n-m-b+1}\sum_{i=b+1}^{n-m+1}\left|\frac{1}{\sqrt{m}}\sum_{j=i}^{i+m-1}\bm{z}_j\right|_F^2}\\
    &\le 4(bp^2)^{1/4}\sigma + \frac{C}{\sigma}b^{3/2}p^{2}\log(n)/\sqrt{n},
        \end{align*}}
        where the third inequality uses Jensen's inequality, the second equality results from ${\rm vec}(AB)=(I\otimes B){\rm vec}(A)$ and the last inequality follows by \cref{lemma_design_matrix_inv}. 
        By choosing the optimal order $\sigma=\bigO(b^{5/8}p^{1/2}n^{-1/4}\log^{1/2}(n))$, we have $I\le Cb^{7/8}p^{5/4}n^{-1/4}\log^{1/2}(n)$.

For $II$, since both $\bm{U}_n^b$ and $\bm{Z}_n^{\ast,b} \mid \Upsilon_{b+1}^n$ are Gaussian vectors in $\mathbb{R}^{bp^2}$,
we will utilize the comparison result in \cref{comparison_cov} to provide a sharper bound. By the proof of \cref{empirical_thm}, it follows that on the event $\mathcal{A}_n$,
    $$ II \le C\left|\widetilde{\bm{\Xi}}_n-\bm{\Xi}_n\right|_F\le Cbp^{17/6}n^{-1/3}h_n.$$
    Putting $I$ and $II$ together, we conclude that
   $$\mathcal{D}'(\bm{Z}_n^{\ast,b}, \bm{U}_n^b)\le C\left(b^{7/8}p^{5/4}n^{-1/4}\log^{1/2}(n)+bp^{17/6}n^{-1/3}h_n\right).$$
 $\hfill \square$
\subsection{Auxiliary technical results}
	Here, we aim to derive the approximation error between the matrices $\widetilde{\bm{\Sigma}}$ and $\bm{\Sigma}$.
    \begin{lemma}\label{lemma_design_matrix}
    Under \cref{ass_dep}, we can obtain
    $\Vert\widetilde{\bm{\Sigma}}-
    \bm{\Sigma}\Vert=\bigO_\Pr
    \left(\frac{bp\log(n)}{\sqrt{n}}\right)$.
	\end{lemma}

    \noindent
\textit{Proof}: This proof primarily utilizes a Bernstein-type inequality for sums of independent random matrices (\cref{lemma_berns}). To facilitate it, we first introduce the $m$-dependent approximation sequence to deal with the issue of independence. Specifically, recall that $\widetilde{\bm{\Sigma}}=\bm{X}^\top\bm{X}/n=\sum_{i=b+1}^n \bm{x}_i^{(b)}\big(\bm{x}_i^{(b)}\big)^\top/n$ and we denote $$\bar{\bm{x}}_{i}^{(b)}=
\EE(\bm{x}_{i}^{(b)}|\eta_{i-m-1},...,\eta_{i-1}),~i=b+1,\cdots,n.$$
		It is easy to find that $\bar{\bm{x}}_{i}^{(b)}$ and $\bar{\bm{x}}_{j}^{(b)}$ are independent if $|i-j|>m$. Further denote $\mathbf{X}_i={\rm vec}\Big(\bm{x}_{i}^{(b)}\big(
        \bm{x}_{i}^{(b)}
        \big)^\top\Big)\in \mathbb{R}^{b^2p^2}$, $\bar{\mathbf{X}}_i={\rm vec}\Big(\bar{\bm{x}}_{i}^{(b)}
        \big(\bm{x}_{i}
        ^{(b)}\big)
        ^\top\Big)\in \mathbb{R}^{b^2p^2}$ and
        $\bar{\bm{\Sigma}}=
        \sum_{i=b+1}^n
        \bar{\bm{x}}_i^{(b)}
        \big(\bar{\bm{x}}
        _i^{(b)}\big)^\top$. By Assumption 4 of the paper and the discussion of \cite[Remark 2.3]{Zhang18}, we can derive that 
        $$\Omega_{m,q}:=
        \sum_{l=m}^\infty \max_i\left\Vert 
        |\bm{G}(\mathcal{F}_i)-
        \bm{G}(\mathcal{F}_{i,l})|_\infty\right\Vert_q\le C(b^2p^2)^{1/q}m^{-\tau+1}.$$
        Consequently, we can obtain
        \begin{align*}
            \Pr\Big(\big\Vert
            \widetilde{\bm{\Sigma}}-\bar{\bm{\Sigma}}
            \big\Vert\ge t\Big)&\le 
            \Pr\Big(bp\big|
            \widetilde{\bm{\Sigma}}-\bar{\bm{\Sigma}}\big|_{\max}
            \ge t\Big)\\
            &\le \Pr\left(\Big|\frac{1}{\sqrt{n}}\sum_{i=b+1}^n
            (\mathbf{X}_i-\bar{\mathbf{X}}_i)\Big|_\infty\ge \frac{t\sqrt{n}}{bp}\right)\\
            &\le \frac{C\{\log(b^2p^2)\}^{q/2}\Omega_{m+1,q}^q}{(t\sqrt{n}/bp)^q}\\
            & \le \frac{C\{\log(b^2p^2)\}^{q/2}b^2p^2(m+1)^{-q(\tau-1)}}{(t\sqrt{n}/bp)^{q}}
        \end{align*}
        By choosing $t\sqrt{n}/bp=\bigO\big((bp)
        ^{2/q}m^{-\tau+2}\big)$ and $m$ sufficiently large, we have $\Pr\Big(\big\Vert
            \widetilde{\bm{\Sigma}}-\bar{\bm{\Sigma}}
            \big\Vert\ge t\Big)=o(1)$. Moreover, we conclude that with high probability
            $$\big\Vert
            \widetilde{\bm{\Sigma}}-\bar{\bm{\Sigma}}
            \big\Vert\le C(bp)^{-1+2/q}m^{-\tau+2}/\sqrt{n}.$$
Now, it suffices to control the $m$-dependence approximation sequence. To employ \cref{lemma_berns},  we need to design the blocks of sequence with independent random variables. Namely, define $k_0=\lfloor \frac{n-b}{m}\rfloor$ denote the index set sequences for $i=b+1,\ldots,b+m$ as 
$$\mathcal{I}_i=\begin{cases}
    \{i+km: k=0,1,\ldots,k_0\}, & i+k_0m\le n,\\
    \{i+km: k=0,1,\ldots,k_0-1\}, &\text{otherwise}.
\end{cases}$$
Similar to the proof in \citep[Lemma 7]{CZ2022}, we can calculate
\begin{equation*}
    R_m=\frac{1}{n}\max_i\Big\Vert
    \bar{\bm{x}}_i^{(b)}\big(
    \bar{\bm{x}}_i^{(b)})^\top-\EE\big[\bar{\bm{x}}_i^{(b)}\big(
    \bar{\bm{x}}_i^{(b)})^\top\big]\Big
    \Vert \le \frac{Cbp}{n},
\end{equation*}
and $\sigma_m^2:=\left\Vert \frac{1}{n^2}\sum_{k\in\mathcal{I}_i}\EE\Big[
\bar{\bm{x}}_k^{(b)}
\big(\bar{\bm{x}}_k^{(b)}
\big)^\top-\EE\big\{\bar{\bm{x}}_k^{(b)}
\big(\bar{\bm{x}}_k^{(b)}\big)^\top\big\}
\Big]^2\right\Vert \le \frac{Cb^2p^2}{mn}$. Then it yields that 
\begin{align*}
    &\Pr\left(\left\Vert\bar{\bm{\Sigma}}-\EE\bar{\bm{\Sigma}}
    \right\Vert \ge t\right)\\
    \le & Cm\max_i\Pr\left(
    \left\Vert\frac{1}{n}\sum_{k \in \mathcal{I}_i}\Big\{
    \bar{\bm{x}}_k^{(b)}
    \big(\bar{\bm{x}}_k^{(b)}\big)^\top-
    \EE\left[\bar{\bm{x}}_k^{(b)}
    \big(\bar{\bm{x}}_k^{(b)}\big)\right]
    \Big\}\right\Vert\ge \frac{t}{m}\right)\\
    \le & Cmbp\cdot {\rm exp}\left(
    \frac{-t^2/(2m^2)}{\sigma_m^2+R_mt/(3m)}\right).
\end{align*}
Then by choosing $m=\bigO(\log n)$ and $t=\bigO\Big(\frac{bp\log(n)}{\sqrt{n}}\Big)$, we can obtain
\begin{align*}
    &\left\Vert \widetilde{\bm{\Sigma}}-
    \bm{\Sigma}\right\Vert\\
    \le & \left\Vert
    \widetilde{\bm{\Sigma}}-
    \bar{\bm{\Sigma}}\right\Vert+
    \left\Vert
    \bar{\bm{\Sigma}}-
    \EE\bar{\bm{\Sigma}}\right\Vert+
    \left\Vert
    \EE\bar{\bm{\Sigma}}-
    \bm{\Sigma}\right\Vert\\
    \le &\frac{Cbp\log(n)}{\sqrt{n}}.
\end{align*}
$\hfill \square$  

\begin{lemma}
\label{lemma_design_matrix_inv}
Under Assumptions \ref{ass_dep} and \ref{ass_design_matrix}, we have 
$$\left\Vert \widetilde{\bm{\Sigma}}^{-1}-
    \bm{\Sigma}^{-1}\right\Vert = 
    \bigO_\Pr\left(\frac{bp\log(n)}{\sqrt{n}}\right).$$
\end{lemma}
\noindent
\textit{Proof}:
According to the result in \cref{lemma_design_matrix}, one can obtain 
\begin{align*}
    &\left\Vert
    \widetilde{\bm{\Sigma}}
    ^{-1}-\bm{\Sigma}^{-1}
    \right\Vert\\
    =&\left\Vert
    \widetilde{\bm{\Sigma}}^{-1}(\bm{\Sigma}-\widetilde{\bm{\Sigma}})
    \bm{\Sigma}^{-1}\right\Vert\\
    \le &\left\Vert
    \widetilde{\bm{\Sigma}}^{-1}
    \right\Vert \left\Vert
    \bm{\Sigma}-\widetilde{\bm{\Sigma}}
    \right\Vert \left\Vert
    \bm{\Sigma}^{-1}\right\Vert\\
    \le & \frac{Cbp\log(n)}{\sqrt{n}}.
\end{align*}
 $\hfill \square$

To prove the consistency of the estimated quantities, we denote 
\begin{align*}
    \widehat{
    \bm{\epsilon}}_i&=
    \widehat{\bm{x}}_i-\sum_{j=1}^b\widehat{\bm{\Phi}}_j\widehat{\bm{x}}_{i-j},\\
    \widehat{f}_{k}^2&=\frac{1}{n}\sum_{i=1}^n\left(r_{i,k}-\frac{1}{n}\sum_{i=1}^n r_{i,k}\right)^2,
\end{align*}
where $\widehat{x}_{i,k}=r_{i,k}/\widehat{f}_k$.

\begin{lemma}\label{consis_epsilon}
     Under Assumptions \ref{ass_conti_u}, \ref{weak_depen_components}--\ref{ass_design_matrix} of the main article, we have for each $k=1,\ldots,p$,
     \begin{align*}
    \frac{\widehat{f}_k}{f_k}-1 &= \bigO_\Pr(n^{-1/2}),\\
    \max_{b+1\le i\le n}|\widehat{\epsilon}_{i,k}-\epsilon_{i,k}|&=
    \bigO_\Pr\big(bp^{3/2}n^{-1/2+2/q}\big).
     \end{align*}
 \end{lemma}
\noindent
\textit{Proof}: Note that $\EE[x_{i,k}]=0$ and $f_{k}^2=\EE[r_{i,k}^2]$. Then it yields that
$$ \widehat{f}_k^2=\frac{1}{n}\sum_{i=1}^n\left(
r_{i,k}-\frac{1}{n}\sum_{i=1}^n r_{i,k}\right)^2=f_k^2\left(
1+\frac{1}{n}\sum_{i=1}^n (x_{i,k}^2-1)-\Big(\frac{1}{n}\sum_{i=1}^n x_{i,k}\Big)^2\right).$$
Let $y_{i,k}=x_{i,k}^2-1$ and define $\delta_y(l,\cdot)$ as the physical dependence measure of $y_{i,k}$. Then by Assumption 4 of the main article, we can obtain
$$\delta_y(l,q/2)=\max_{1\le k\le p}\Vert y_{i,k}-y_{i,k}^\ast\Vert_{q/2}\le 
\max_{1\le k\le p}\Vert x_{i,k}+x_{i,k}^\ast\Vert_q \max_{1\le k\le p}\Vert x_{i,k}-x_{i,k}^\ast\Vert_q \le C(l+1)^{-\tau}.$$
According to \citep[Theorem 1]{wu2007strong}, we have for $k=1,\ldots,p$,
    $$\left\Vert \sum_{i=1}^n y_{i,k}/n \right\Vert_{q/2}=\bigO(1/\sqrt{n}),\quad \left\Vert \Big(\sum_{i=1}^n x_{i,k}/n \Big)^2\right\Vert_{q/2}
    =\bigO(1/n),\qquad q>4.$$ Consequently, one can obtain
    $\Vert\widehat{f}_k^2
    /f_k^2-1\Vert_{q/2}
    =\bigO(1/\sqrt{n})$ for $q>4$. By elementary calculation, we conclude that $\widehat{f}_k/f_k-1=\bigO_\Pr(n^{-1/2})$ for $k=1,\ldots,p$.

    On the other hand, to investigate the estimation consistency of the residual $\bm{\epsilon}_{i,k}$, recall $$\bm{\xi}=(\bm{0}_{p\times 1}, \bm{\Phi}_{b+1}\bm{x}_1,\sum_{j=b+1}^{b+2}\bm{\Phi}_j\bm{x}_{b+3-j},\ldots,\sum_{j=b+1}^{n-1}\bm{\Phi}_j\bm{x}_{n-j})^\top.$$ Then based on the definition of the residual and the multiple linear regression model built in Section 3, we have
    \begin{align*}
        \widehat{\bm{\epsilon}}-\bm{\epsilon}&=\left(\bm{Y}-\widehat{\bm{X}}\widehat{\bm{\beta}}\right)-(\bm{Y}-\bm{X}\bm{\beta}-\bm{\xi})\\
        &=(\bm{X}-\widehat{\bm{X}})\widehat{\bm{\beta}}-\bm{X}(\widehat{\bm{\beta}}-\bm{\beta})+\bm{\xi}.
    \end{align*}
    Furthermore, denote $\bm{E}_i$ is an $(n-b)$-dimensional vector with the $i$th entry being 1 and others being 0, $\bm{V}_k$ is a $p$-dimensional vector with the $k$th entry being 1 and others being 0. Armed with \eqref{eq_diff}, it turns out that for any $k=1,\ldots,p$,
    \begin{align*}
        &\Vert \widehat{\epsilon}_{i,k}-\epsilon_{i,k}\Vert_{q/2}\\
        \le & \Vert \bm{E}_i^\top
        (\bm{X}-\widehat{\bm{X}})\widehat{\bm{\beta}}\bm{V}_k\Vert_{q/2}+\Vert \bm{E}_i^\top \bm{X}(\bm{\Phi}^M+\bm{\Phi}^R)\bm{V}_k\Vert_{q/2}+\Vert \bm{E}_i^\top\bm{\xi}
        \bm{V}_k\Vert_{q/2}\\
        \le &\Vert \bm{E}_i^\top
        (\bm{X}-\widehat{\bm{X}})\Vert_q\Vert \widehat{\bm{\beta}}\Vert \Vert \bm{V}_k\Vert_{q}+
        \Vert \bm{E}_i^\top \bm{X}\Vert_q\left\Vert \left(\frac{\bm{X}^\top\bm{X}}{n}\right)^{-1}\right\Vert \left\Vert\sum_{i=b+1}^n \bm{x}_i^{(b)}
        \bm{\epsilon}_i^\top\bm{V}_k/n\right\Vert_q+\Vert \bm{E}_i^\top \bm{X}\Vert_q \Vert \bm{\Phi}^R\bm{V}_k\Vert_q+\Vert \bm{E}_i^\top \bm{\xi}\bm{V}_k\Vert_{q/2}\\
        \le & C\left(\sqrt{bp/n}+\sqrt{b^2p^3/n}+\sqrt{bp^2/n}b^{-\tau+2}(\log b)^{\tau-2}+\sqrt{p}b^{-\tau+2}(\log b)^{\tau-2}\right)\\
        \le & Cbp^{3/2}/\sqrt{n}.
    \end{align*}
    By the maximum inequality, we obtain that $\Vert \max_{b+1 \le i\le n}|\widehat{\epsilon}_{i,k}-
    \epsilon_{i,k}|\Vert_{q/2}\le Cbp^{3/2}n^{-1/2+2/q}$ for $q>4$ and hence conclude the estimation consistency $\max_{b+1 \le i\le n}|\widehat{\epsilon}_{i,k}-
    \epsilon_{i,k}|=\bigO_\Pr(bp^{3/2}n^{-1/2+2/q})$ for all $k=1,\ldots,p$.
$\hfill \square$

    \section{Auxiliary lemmas}\label{app_auxilary}
    In this section, we will introduce some lemmas used in our proofs of the main results. 

Denote $\mathcal{T}$ as the collection of all convex sets in $\mathbb{R}^p$. For $T\in \mathcal{T}$, let $h_T(x)=I_T(x)$ and denote the smoothed function
$h_{T,\sigma}(\omega)=\psi\left({\rm dist}(\omega,T)/\sigma\right)$ where ${\rm dist}(\omega,T)=\inf_{v\in T}|\omega-v|$ and $$\psi(x)=
\begin{cases}
	1 & x<0,\\
	1-2x^2 &0\le x<\frac{1}{2},\\
	2(1-x)^2 &\frac{1}{2}\le x<1,\\
	0 &x\ge1.
\end{cases}$$
Followed by the notation in \cite{Fang15}, define
$$\mathcal{D}(\mathscr{L}(\bm{W}), \mathscr{L}(\bm{Z}))=\sup _{A \in \mathcal{A}}|\mathbb{P}(\bm{W} \in A)-\mathbb{P}(\bm{Z} \in A)|,$$ where
where $\bm{W}, \bm{Z}$ are Gaussian random vectors and $\mathcal{A}$ denotes the collection of all convex sets in $\mathbb{R}^d$.

\begin{lemma}[Lemma 2.3 of \cite{Bentkus03}]\label{derivative}
	The function $h_{T,\sigma}(\omega)$ as defined above has the property
	$$\left|\nabla h_{T,\sigma}(\omega)\right|\le 2\sigma^{-1}~\text{for all}~\omega \in \mathbb{R}^{p}.$$ 
	\end{lemma}

\begin{lemma}[Lemma 4.2 of \cite{Fang15}]\label{comparison}
	For any $p$-dimensional random vector $\bm{W}$,
	$$\mathcal{D}(\bm{W},\bm{Z})\le 4p^{\frac{1}{4}}\sigma+\sup_{T\in\mathcal{T}}
	|\EE[h_{T,\sigma}(\bm{W})-h_{T,\sigma}(\bm{Z})]|,$$
	where $\bm{Z}$ is a $p$-dimensional standard Gaussian vector.
\end{lemma}

\begin{lemma}[Lemma B.1 of \cite{liu2025wasserstein}]\label{comparison_new}
For any $p$-dimensional random vector $\bm{W}$,
	$$\mathcal{D}(\bm{W},\bm{Z})\le 4p^{\frac{1}{4}}\sigma+\min\left\{\sup_{T\in\mathcal{T}}
	|\EE[h_{T,\sigma}(\bm{W})-h_{T,\sigma}(\bm{Z})]|,\Pr(|\bm{W}-\bm{Z}|>\sigma)\right\},$$
	where $\bm{Z}$ is a $p$-dimensional standard Gaussian vector.
\end{lemma}

\begin{lemma}[Lemma 4.4 of \cite{liu2025wasserstein}]\label{comparison_cov}
Suppose $\bm{X}$ and $\bm{Y}$ are zero mean Gaussian vectors in $\mathbb{R}^d$, with positive definite covariance matrices $\bm{\Sigma}_X$ and $\bm{\Sigma}_Y$, respectively. Suppose the smallest
eigenvalue of either $\bm{\Sigma}_X$ or $\bm{\Sigma}_Y$ is bounded below by $\lambda_\ast>0$, then
$$\mathcal{D}(\bm{X},\bm{Y})\le \frac{3}{2}\min\left\{1,\lambda_\ast^{-1}|\bm{\Sigma}_X-\bm{\Sigma}_Y|_F\right\}.$$
\end{lemma}

\begin{lemma}[Remark 2.2 of \cite{Fang16}]\label{distance}
	Let $\bm{W}=\sum_{i=1}^n\bm{X}_i$ be a sum of $p$-dimensional random vectors such that $\EE(\bm{X}_i)=0$ and ${\rm Cov}(\bm{W})=\bm{\Sigma}_w$. Suppose $\bm{W}$ can be decomposed as follows:\\
	1. $\forall i \in [n]$, $\exists i\in N_i \subset [n]$, such that $\bm{W}-\bm{X}_{N_i}$ is independent of $\bm{X}_i$, where $[n] = \{1,\cdots,n\}$.\\
	2. $\forall i \in [n]$, $j\in N_i$, $\exists N_i \subset N_{ij} \subset [n]$, such that $\bm{W}-\bm{X}_{N_{ij}}$ is independent of 
	$\{\bm{X}_i, \bm{X}_j\}$.\\
	3. $\forall i \in [n]$, $j \in N_i,~k \in N_{ij},~\exists N_{ij} \subset N_{ijk} \subset [n]$ such that $\bm{W}-\bm{X}_{N_{ijk}}$ is independent of 
	$\{\bm{X}_i,\bm{X}_j,\bm{X}_k\}$.
	
	Suppose further that for each $i\in [n],~j\in N_i,~k\in N_{ij},~
	|\bm{X}_i|\le \beta,~|N_i|\le n_1,~|N_{ij}|\le n_2,~|N_{ijk}|\le n_3$,
	where $|\cdot|$ is the Euclidean norm of a vector or the operator norm of a matrix. Then there exists a universal constant $C$ such that
	$$\mathcal{D}(\bm{W},\bm{\Sigma}_\omega^{1/2}\bm{Z})\le Cp^{\frac{1}{4}}n|\bm{\Sigma}_w^{-1/2}|^3
		\beta^3n_1(n_2+n_3/p),$$
		where $\bm{Z}$ is a $p$-dimensional standard Gaussian random vector.
\end{lemma}

In the proof of Proposition 1, we employ the following concentration inequalities.
	\begin{lemma}[Theorem 4.1 and its following Remark of \cite{johnson1985best}]\label{rosenthal}
		A version of the Rosenthal inequality for independent variables $\{X_i\}$:
		$$\left\Vert \sum_{i=1}^n X_i\right\Vert_q\le \frac{14.7q}{\log q}(\mu_{n,2}^{1/2}+\mu_{n,q}^{1/q}),$$ where 
		$\mu_{n,q}=\sum_{i=1}^{n}\EE|X_i|^q$.
	\end{lemma}

\begin{lemma}[Theorem 2.1 of \cite{rio2009moment}]\label{mar_concentration}
	A version of the Burkholder inequality for martingale differences $\{X_i\}$:
	$$\left\Vert \sum_{i=1}^n X_i\right\Vert_q^2
	\le (q-1)\sum_{i=1}^n\Vert X_i\Vert_q^2.$$
\end{lemma}

    \begin{lemma}[Weyl's inequality]\label{lemma_w}
    	Let $\bm{A}$ and $\bm{B}$ be Hermitian $n\times n$ matrices with eigenvalues $\lambda_1(\bm{A})\ge \cdots\ge \lambda_n(\bm{A}),~\lambda_1(\bm{B})\ge \cdots\ge \lambda_n(\bm{B})$, respectively. If $1\le k\le i\le n$ and $1\le l\le n-i+1$, then
    	$$\lambda_{i+l-1}(\bm{A})+\lambda_{n-l+1}(\bm{B})\le 
    	\lambda_i(\bm{A}+\bm{B})\le \lambda_{i-k+1}(\bm{A})+\lambda_k(\bm{B}).$$
    	In particular,
    	$$\lambda_i(\bm{A})+\lambda_n(\bm{B})\le \lambda_i(\bm{A}+\bm{B})\le 
    	\lambda_i(\bm{A})+\lambda_1(\bm{B}).$$
    \end{lemma}
    
	Lastly, we will mention the Bernstein-type inequality for independent random matrices.
	\begin{lemma}[Theorem 6.1.1 of \cite{Tropp15}]\label{lemma_berns}
		Let $\{\bm{\Xi}_i\}_{i=1}^n$ be a finite sequence of independent random matrices with dimensions $d_1\times d_2$. Assume $\EE(\bm{\Xi}_i)=\bm{0}$ for each $i$, $\max_{1\le i\le n}\Vert\bm{\Xi}_i\Vert\le R_n$ and define 
		$$\sigma_n^2=\max\left\{\left\Vert\sum_{i=1}^n\EE\left(
		\bm{\Xi}_i\bm{\Xi}_i^\top\right)\right\Vert,
		\left\Vert\sum_{i=1}^n\EE\left(\bm{\Xi}_i^\top\bm{\Xi}_i\right)
		\right\Vert\right\},$$ where the norm represents the largest singular value. Then for all $t>0$,
		$$\Pr\left(\left\Vert\sum_{i=1}^n\bm{\Xi}_i\right\Vert\ge t\right)\le (d_1+d_2)\exp\left(\frac{-t^2/2}{\sigma_n^2+R_nt/3}\right).$$
	\end{lemma}

\bibliographystyle{chicago}
\bibliography{test}

@Article{CZ2023,
  title={Optimal Short-Term Forecast for Locally Stationary Functional Time Series},
 author={Cui, Yan and Zhou, Zhou},
 journal={IEEE Transactions on Information Theory},
 volume  = {71},
 number  = {4},
 pages   = {2872-2887},
 year={2025}
}

@Book{horvath2012inference,
  title     = {Inference for Functional Data with Applications},
  publisher = {Springer Science \& Business Media},
  year      = {2012},
  author    = {Horv{\'a}th, Lajos and Kokoszka, Piotr},
}

@Article{Shang14,
  author  = {Han Lin Shang},
  title   = {A survey of functional principal component analysis},
  journal = {AStA Advances in Statistical Analysis},
  year    = {2014},
  volume  = {98},
  number  = {1},
  pages   = {121--142},
  issn    = {1863-8171},
  doi     = {10.1007/s10182-013-0213-1},
}

@Article{Chen07,
  author    = {Chen, Xiaohong},
  title     = {Large sample sieve estimation of semi-nonparametric models},
  journal   = {Handbook of Econometrics},
  year      = {2007},
  volume    = {6},
  pages     = {5549--5632},
  publisher = {Elsevier},
}

@Article{Trefethen2008,
  author  = {Lloyd N. Trefethen},
  title   = {Is {G}auss Quadrature Better than {C}lenshaw–{C}urtis?},
  journal = {SIAM review},
  year    = {2008},
  volume  = {50},
  number  = {1},
  pages   = {67--87},
  issn    = {0036-1445},
  doi     = {10.1137/060659831},
}

@Book{Meyer90,
  title        = {Ondelettes et o\'{p} erateurs. Actualit\'{e}s math\'{e}matiques.},
  publisher    = {Hermann, Paris},
  year         = {1990},
  author       = {Yves Meyer},
  howpublished = {I. Actualit´ es Math´ ematiques},
}

@Article{WX11,
  author  = {Haiyong Wang and Shuhuang Xiang},
  title   = {On the convergence rates of {L}egendre approximation},
  journal = {Mathematics of Computation},
  year    = {2011},
  volume  = {81},
  number  = {278},
  pages   = {861--877},
  issn    = {0025-5718},
  doi     = {10.1090/s0025-5718-2011-02549-4},
}

@Book{KR17,
  title     = {Introduction to Functional Data Analysis},
  publisher = {Chapman and Hall/CRC},
  year      = {2017},
  author    = {Piotr Kokoszka and Matthew Reimherr},
}

@Article{LiHsing07,
  author    = {Li, Yehua and Hsing, Tailen},
  title     = {On rates of convergence in functional linear regression},
  journal   = {Journal of Multivariate Analysis},
  year      = {2007},
  volume    = {98},
  number    = {9},
  pages     = {1782--1804},
  publisher = {Elsevier},
}

@Article{Hall06,
  author    = {Hall, Peter and Hosseini-Nasab, Mohammad},
  title     = {On properties of functional principal components analysis},
  journal   = {Journal of the Royal Statistical Society: Series B (Statistical Methodology)},
  year      = {2006},
  volume    = {68},
  number    = {1},
  pages     = {109--126},
  publisher = {Wiley Online Library},
}

@Article{Aue2015,
  author  = {Alexander Aue and Diogo Dubart Norinho and Siegfried H\"{o}rmann},
  title   = {On the Prediction of Stationary Functional Time Series},
  journal = {Journal of the American Statistical Association},
  year    = {2015},
  volume  = {110},
  number  = {509},
  pages   = {378--392},
  issn    = {0162-1459},
  doi     = {10.1080/01621459.2014.909317},
}

@article{CXW13,
	author = {Xiaohui Chen and Mengyu Xu and Wei Biao Wu},
	title = {{Covariance and precision matrix estimation for high-dimensional time series}},
	volume = {41},
	journal = {The Annals of Statistics},
	number = {6},
	publisher = {Institute of Mathematical Statistics},
	pages = {2994 -- 3021},
	year = {2013},
	doi = {10.1214/13-AOS1182},
	URL = {https://doi.org/10.1214/13-AOS1182}
}

@Book{BD91,
  title     = {Time series: Theory and Methods (Second Edition).},
  publisher = {Springer-Verlag},
  year      = {1991},
  author    = {P. Brockwell and R. Davis},
  volume    = {Second Edition},
}

@Article{CLZ16,
  author  = {T. Tony Cai and Weidong Liu and Harrison H. Zhou},
  title   = {Estimating sparse precision matrix: Optimal rates of convergence and adaptive estimation},
  journal = {The Annals of Statistics},
  year    = {2016},
  volume  = {44},
  number  = {2},
  pages   = {455--488},
  issn    = {0090-5364},
  doi     = {10.1214/13-aos1171},
}

@Article{Wu05,
  author    = {Wu, Wei Biao},
  title     = {Nonlinear system theory: Another look at dependence},
  journal   = {Proceedings of the National Academy of Sciences},
  year      = {2005},
  volume    = {102},
  number    = {40},
  pages     = {14150--14154},
  publisher = {National Acad Sciences},
}

@article{zhou2023,
    author = {Zhou, Zhou and Dette, Holger},
    title = {Statistical inference for high-dimensional panel functional time series},
    journal = {Journal of the Royal Statistical Society Series B: Statistical Methodology},
    volume = {85},
    number = {2},
    pages = {523--549},
    year = {2023},
    month = {04},
    issn = {1369-7412},
    doi = {10.1093/jrsssb/qkad015},
}

@article{hormann2010weakly,
  title={Weakly dependent functional data},
  author={H{\"o}rmann, Siegfried and Kokoszka, Piotr},
  journal={The Annals of Statistics},
  volume={38},
  number={3},
  pages={1845--1884},
  year={2010},
  publisher={Institute of Mathematical Statistics}
}

@Article{Bentkus03,
  author    = {Bentkus, Vidmantas},
  title     = {On the dependence of the {B}erry--{E}sseen bound on dimension.},
  journal   = {Journal of Statistical Planning and Inference},
  year      = {2003},
  volume    = {113},
  number    = {2},
  pages     = {385--402},
  publisher = {Elsevier}
}

@Article{Fang15,
  author    = {Fang, Xiao and R{\"o}llin, Adrian.},
  title     = {Rates of convergence for multivariate normal approximation with applications to dense graphs and doubly indexed permutation statistics},
  journal   = {Bernoulli},
  year      = {2015},
  volume    = {21},
  number    = {4},
  pages     = {2157--2189},
  publisher = {Bernoulli Society for Mathematical Statistics and Probability},
}

@Article{Fang16,
  author    = {Fang, Xiao},
  title     = {A multivariate {CLT} for bounded decomposable random vectors with the best known rate},
  journal   = {Journal of Theoretical Probability},
  year      = {2016},
  volume    = {29},
  number    = {4},
  pages     = {1510--1523},
  publisher = {Springer},
}

@Article{LiuLin09,
  author    = {Liu, Weidong and Lin, Zhengyan},
  title     = {Strong approximation for a class of stationary processes},
  journal   = {Stochastic Processes and their Applications},
  year      = {2009},
  volume    = {119},
  number    = {1},
  pages     = {249--280},
  publisher = {Elsevier},
}

@article{zhang2016white,
  title={White noise testing and model diagnostic checking for functional time series},
  author={Zhang, Xianyang},
  journal={Journal of Econometrics},
  volume={194},
  number={1},
  pages={76--95},
  year={2016},
  publisher={Elsevier}
}

@article{kokoszka2017inference,
  title={Inference for the autocovariance of a functional time series under conditional heteroscedasticity},
  author={Kokoszka, Piotr and Rice, Gregory and Shang, Han Lin},
  journal={Journal of Multivariate Analysis},
  volume={162},
  pages={32--50},
  year={2017},
  publisher={Elsevier}
}

@article{kim2023,
author = {Mihyun Kim and Piotr Kokoszka and Gregory Rice},
title = {{White noise testing for functional time series}},
volume = {17},
journal = {Statistics Surveys},
number = {none},
publisher = {Amer. Statist. Assoc., the Bernoulli Soc., the Inst. Math. Statist., and the Statist. Soc. Canada},
pages = {119--168},
year = {2023},
doi = {10.1214/23-SS143},
URL = {https://doi.org/10.1214/23-SS143}
}

@article{alvarez2025goodness,
  title={A goodness-of-fit test for functional time series with applications to Ornstein-Uhlenbeck processes},
  author={{\'A}lvarez-Li{\'e}bana, J and L{\'o}pez-P{\'e}rez, A and Gonz{\'a}lez-Manteiga, W and Febrero-Bande, M},
  journal={Computational Statistics \& Data Analysis},
  volume={203},
  pages={108092},
  year={2025},
  publisher={Elsevier}
}

@article{wu2024frequency,
  title={Frequency detection and change point estimation for time series of complex oscillation},
  author={Wu, Hau-Tieng and Zhou, Zhou},
  journal={Journal of the American Statistical Association},
  volume={119},
  number={547},
  pages={1945--1956},
  year={2024},
  publisher={Taylor \& Francis}
}

@Book{Bell04,
  title     = {Special Functions for Scientists and Engineers (Dover Books on Mathematics)},
  publisher = {Dover Publications},
  year      = {2004},
  author    = {W. W. Bell},
}

@Article{CZ2022,
  author  = {Cui, Y. and Zhou, Z.},
  title   = {Simultaneous inference for time series functional linear regression},
  journal = {arXiv:2207.11392},
  year    = {2026},
}

@article{liu2013probability,
  title={Probability and moment inequalities under dependence},
  author={Liu, Weidong and Xiao, Han and Wu, Wei Biao},
  journal={Statistica Sinica},
  pages={1257--1272},
  year={2013},
  publisher={JSTOR}
}

@article{johnson1985best,
  title={Best constants in moment inequalities for linear combinations of independent and exchangeable random variables},
  author={Johnson, William B and Schechtman, Gideon and Zinn, Joel},
  journal={The Annals of Probability},
  pages={234--253},
  year={1985},
  publisher={JSTOR}
}

@article{rio2009moment,
  title={Moment inequalities for sums of dependent random variables under projective conditions},
  author={Rio, Emmanuel},
  journal={Journal of Theoretical Probability},
  volume={22},
  number={1},
  pages={146--163},
  year={2009},
  publisher={Springer}
}

@Book{Tropp15,
  title     = {An Introduction to Matrix Concentration Inequalities},
  publisher = {Now Publishers, Inc.},
  year      = {2015},
  author    = {Tropp, Joel A.},
  journal   = {Foundations and Trends in Machine Learning arXiv preprint arXiv:1501.01571, 2015},
}

@Article{Zhang18,
  author    = {Zhang, Xianyang and Cheng, Guang},
  title     = {Gaussian approximation for high dimensional vector under physical dependence},
  journal   = {Bernoulli},
  year      = {2018},
  volume    = {24},
  number    = {4A},
  pages     = {2640--2675},
  publisher = {Bernoulli Society for Mathematical Statistics and Probability},
}

@article{politis2004automatic,
  title={Automatic block-length selection for the dependent bootstrap},
  author={Politis, Dimitris N and White, Halbert},
  journal={Econometric reviews},
  volume={23},
  number={1},
  pages={53--70},
  year={2004},
  publisher={Taylor \& Francis}
}

@article{wu2007strong,
  title={Strong invariance principles for dependent random variables},
  author={Wu, Wei Biao},
journal={The Annals of Probability},
  volume={35},
  number={6},
  pages={2294--2320},
year={2007}
}

@article{liu2025wasserstein,
  title={Wasserstein and Convex Gaussian Approximations for Non-stationary Time Series of Diverging Dimensionality},
  author={Liu, Miaoshiqi and Yang, Jun and Zhou, Zhou},
  journal={arXiv preprint arXiv:2506.08723},
  year={2025}
}

@article{zhang2024simultaneous,
  title={Simultaneous Inference for Non-Stationary Random Fields, with Application to Gridded Data Analysis},
  author={Zhang, Yunyi and Zhou, Zhou},
  journal={arXiv preprint arXiv:2409.01220},
  year={2024}
}

@book{hsing2015theoretical,
  title={Theoretical foundations of functional data analysis, with an introduction to linear operators},
  author={Hsing, Tailen and Eubank, Randall L},
  volume={997},
  year={2015},
  publisher={Wiley Online Library}
}

@Article{Hall07,
  author    = {Hall, Peter and Horowitz, Joel L.},
  title     = {Methodology and convergence rates for functional linear regression},
  journal   = {The Annals of Statistics},
  year      = {2007},
  volume    = {35},
  number    = {1},
  pages     = {70--91},
  publisher = {Institute of Mathematical Statistics}
}

@book{lutkepohl2013introduction,
  title={Introduction to multiple time series analysis},
  author={L{\"u}tkepohl, Helmut},
  year={2013},
  publisher={Springer Science \& Business Media}
}

@article{kim2024projection,
  title={Projection-based white noise and goodness-of-fit tests for functional time series},
  author={Kim, Mihyun and Kokoszka, Piotr and Rice, Gregory},
  journal={Statistical Inference for Stochastic Processes},
  volume={27},
  number={3},
  pages={693--724},
  year={2024},
  publisher={Springer}
}

@article{characiejus2020general,
  title={A general white noise test based on kernel lag-window estimates of the spectral density operator},
  author={Characiejus, Vaidotas and Rice, Gregory},
  journal={Econometrics and Statistics},
  volume={13},
  pages={175--196},
  year={2020},
  publisher={Elsevier}
}

@article{gonzalez2023testing,
  title={Testing the goodness-of-fit of a functional autoregressive model},
  author={Gonz{\'a}lez-Manteiga, Wenceslao and Ruiz-Medina, Mar{\'\i}a Dolores and L{\'o}pez-P{\'e}rez, AM and {\'A}lvarez-Li{\'e}bana, Javier},
  journal={arXiv:2303.09644},
  year={2026}
}

@article{MESTRE2021107108,
title = {Functional time series model identification and diagnosis by means of auto- and partial autocorrelation analysis},
journal = {Computational Statistics \& Data Analysis},
volume = {155},
pages = {107108},
year = {2021},
author = {Guillermo Mestre and José Portela and Gregory Rice and Antonio {Muñoz San Roque} and Estrella Alonso}
}

@article{chiou2016multivariate,
  title={Multivariate functional linear regression and prediction},
  author={Chiou, Jeng-Min and Yang, Ya-Fang and Chen, Yu-Ting},
  journal={Journal of Multivariate Analysis},
  volume={146},
  pages={301--312},
  year={2016},
  publisher={Elsevier}
}

@article{chang2024autocovariance,
  title={An autocovariance-based learning framework for high-dimensional functional time series},
  author={Chang, Jinyuan and Chen, Cheng and Qiao, Xinghao and Yao, Qiwei},
  journal={Journal of Econometrics},
  volume={239},
  number={2},
  pages={105385},
  year={2024},
  publisher={Elsevier}
}

@Article{horvath14,
  author    = {Horv{\'a}th, Lajos and Kokoszka, Piotr and Rice, Gregory},
  title     = {Testing stationarity of functional time series},
  journal   = {Journal of Econometrics},
  year      = {2014},
  volume    = {179},
  number    = {1},
  pages     = {66--82},
  publisher = {Elsevier},
}

@Book{Bosq2000,
  author    = {Bosq, Denis},
  title     = {Linear Processes in Function Spaces: Theory and Applications},
  series    = {Lecture Notes in Statistics},
  volume    = {149},
  publisher = {Springer},
  address   = {New York},
  year      = {2000},
  doi       = {10.1007/978-1-4612-1154-9},
}

@article{KlepschEtAl2017,
  author  = {Klepsch, Johannes and Kl{\"u}ppelberg, Claudia and Wei, Taoran},
  title   = {Prediction of Functional {ARMA} Processes with an Application
             to Traffic Data},
  journal = {Econometrics and Statistics},
  year    = {2017},
  volume  = {1},
  pages   = {128--149},
  doi     = {10.1016/j.ecosta.2016.10.009},
}

@article{Bosq2014,
  author  = {Bosq, Denis},
  title   = {Computing the best linear predictor in a {Hilbert} space.
             Applications to general {ARMAH} processes},
  journal = {Journal of Multivariate Analysis},
  volume  = {124},
  pages   = {436--450},
  year    = {2014},
  doi     = {10.1016/j.jmva.2013.11.013}
}

@article{Kuenzer2024,
  author  = {Kuenzer, Thomas},
  title   = {Estimation of functional {ARMA} models},
  journal = {Bernoulli},
  volume  = {30},
  number  = {1},
  pages   = {117--142},
  year    = {2024},
  doi     = {10.3150/23-BEJ1591}
}

@article{KuehnertRiceAue2026,
  author  = {K{\"u}hnert, Sebastian and Rice, Gregory and Aue, Alexander},
  title   = {Estimating invertible processes in {Hilbert} spaces,
             with applications to functional {ARMA} processes},
  journal = {Bernoulli},
  volume  = {32},
  number  = {2},
  pages   = {1523--1546},
  year    = {2026},
  doi     = {10.3150/25-BEJ1918}
}

@article{Mas2007,
  author  = {Mas, Andr{\'e}},
  title   = {Weak convergence in the functional autoregressive model},
  journal = {Journal of Multivariate Analysis},
  volume  = {98},
  number  = {6},
  pages   = {1231--1261},
  year    = {2007},
  doi     = {10.1016/j.jmva.2006.05.010}
}

@article{GabrysKokoszka2007,
  author  = {Gabrys, Robertas and Kokoszka, Piotr},
  title   = {Portmanteau Test of Independence for Functional Observations},
  journal = {Journal of the American Statistical Association},
  year    = {2007},
  volume  = {102},
  number  = {480},
  pages   = {1338--1348},
  doi     = {10.1198/016214507000001111},
}

@article{HlavkaEtAl2021,
  author  = {Hl{\'a}vka, Zden{\v e}k and Hu{\v s}kov{\'a}, Marie
             and Meintanis, Simos G.},
  title   = {Testing Serial Independence with Functional Data},
  journal = {TEST},
  year    = {2021},
  volume  = {30},
  number  = {3},
  pages   = {603--629},
  doi     = {10.1007/s11749-020-00732-0},
}

@article{KokoszkaReimherr2013,
  author  = {Kokoszka, Piotr and Reimherr, Matthew},
  title   = {Determining the order of the functional autoregressive model},
  journal = {Journal of Time Series Analysis},
  volume  = {34},
  number  = {1},
  pages   = {116--129},
  year    = {2013},
  doi     = {10.1111/j.1467-9892.2012.00816.x}
}
\end{document}